\documentclass[12pt,english]{article}

\usepackage[toc,page,header]{appendix} 
\usepackage{minitoc} 

\usepackage{ae,aecompl}

\usepackage[T1]{fontenc}
\usepackage[latin9]{inputenc}
\usepackage{geometry}
\usepackage{color}
\usepackage{amsmath}
\usepackage{amstext}
\usepackage{amsthm}
\usepackage{amssymb}
\usepackage{bm} 
\usepackage{graphicx}
\usepackage{mathtools} 
\usepackage{setspace}
\usepackage[authoryear]{natbib}
\usepackage{bibunits} 
\defaultbibliographystyle{apalike}
\usepackage[]{mdframed}
\usepackage{commath} 
\mdfdefinestyle{myenvs}{%
  hidealllines=true,%
  nobreak=true, 
  leftmargin=0pt,
  rightmargin=0pt,
  innerleftmargin=0pt,
  innerrightmargin=0pt,
} 
\PassOptionsToPackage{normalem}{ulem}
\usepackage{ulem}
\usepackage[unicode=true,pdfusetitle,
bookmarks=true,bookmarksnumbered=false,bookmarksopen=false,breaklinks=true,pdfborder={0 0 0},pdfborderstyle={},backref=false,colorlinks=true]{hyperref}
\hypersetup{linkcolor=magenta, urlcolor=blue, citecolor=blue, pdfstartview={FitH}, unicode=true}

\def\tr{\textcolor{red}}

\usepackage{comment} 
\usepackage{float}

\makeatletter
\numberwithin{equation}{section}
\numberwithin{figure}{section}
\theoremstyle{plain}
\newmdtheoremenv[style=myenvs]{lyxalgorithm}{Algorithm}
\newtheorem{lem}{\protect\lemmaname}
\theoremstyle{plain}
\newtheorem{assumption}{\protect\assumptionname}
\theoremstyle{remark}
\newtheorem{rem}{\protect\remarkname}
\theoremstyle{plain}
\newtheorem{cor}{\protect\corollaryname}
\theoremstyle{plain}
\newtheorem{thm}{\protect\theoremname}
\theoremstyle{definition}

\usepackage{apptools}
\AtAppendix{\counterwithin{assumption}{section}}
\AtAppendix{\counterwithin{cor}{section}}
\AtAppendix{\counterwithin{equation}{section}}
\AtAppendix{\counterwithin{lem}{section}}
\AtAppendix{\counterwithin{lyxalgorithm}{section}}
\AtAppendix{\counterwithin{rem}{section}}
\AtAppendix{\counterwithin{thm}{section}}

\def\E{\mathrm{E}}
\def\P{\mathrm P}

\def\N{\mathbb N}
\def\R{\mathbb R}
\def\Z{\mathbb Z}

\def\bfA{\mathbf{A}}
\def\be{\boldsymbol{e}}

\def\bF{\boldsymbol{F}}
\def\bG{\boldsymbol{G}}

\def\bI{\mathbf{I}}
\def\bS{\boldsymbol{S}}
\def\bu{\boldsymbol{u}}

\def\bX{\boldsymbol{X}}
\def\bfX{\mathbf{X}}
\def\bY{\boldsymbol{Y}}
\def\bZ{\boldsymbol{Z}}
\def\bbeta{\boldsymbol{\beta}}
\def\bdelta{\boldsymbol{\delta}}
\def\bepsilon{\boldsymbol{\varepsilon}}
\def\bseta{\boldsymbol{\eta}}
\def\bmu{\boldsymbol{\mu}}
\def\bSigma{\boldsymbol{\Sigma}}
\def\bTheta{\boldsymbol{\Theta}}
\def\bUpsilon{\boldsymbol{\Upsilon}}

\DeclareMathOperator{\supp}{supp}
\DeclareMathOperator{\sign}{sign}

\newcommand*{\argmin}{\operatornamewithlimits{argmin}\limits}

\newcommand\normn[1]{\lVert#1\rVert}

\usepackage{paralist}

\makeatother

\providecommand{\assumptionname}{Assumption}
\providecommand{\corollaryname}{Corollary}
\providecommand{\examplename}{Example}
\providecommand{\lemmaname}{Lemma}
\providecommand{\remarkname}{Remark}
\providecommand{\theoremname}{Theorem}

\begin{document}


\title{Data-Driven Tuning Parameter Selection for High-Dimensional Vector Autoregressions
\thanks{First version: October 2022. We thank Robert Adamek, Peter Boswijk, Abhishek Chakrabortty, Lan Gao, Frank Kleibergen, Bent Nielsen, Anders Rahbek and Mikkel S{\o}lvsten for insightful comments and discussion. We also thank participants at the Aarhus
Workshops in Econometrics I and IV, the  2024 Bristol econometrics study group conference, the 35\textsuperscript{th} (EC)\textsuperscript{2} conference, and seminar participants at University of Amsterdam for suggestions and discussion.
Moreover, we thank Zifan Li for providing details about the simulation
experiment in \citet{wong_lasso_2020}. Pedersen gratefully acknowledges
financial support from Independent Research Fund Denmark (FSE 0133-00162B).}   
}

\author{Anders Bredahl Kock\footnote{University of Oxford  and Aarhus Center for Econometrics (ACE); e-mail:
\tt{anders.kock@economics.ox.ac.uk}.} \and Rasmus S{\o}ndergaard Pedersen\footnote{University of Copenhagen and Danish Finance Institute; e-mail:
\tt{rsp@econ.ku.dk}.} \and Jesper
Riis-Vestergaard S{\o}rensen\footnote{University of Copenhagen and Aarhus Center for Econometrics (ACE); e-mail:
\tt{jrvs@econ.ku.dk}.}}

\maketitle

\begin{abstract}
Lasso-type estimators are routinely used to estimate high-dimensional time series models. The theoretical guarantees established for these estimators typically require the penalty level to be chosen in a suitable fashion often depending on \emph{unknown} population quantities. Furthermore, the resulting estimates and the number of variables retained in the model depend crucially on the chosen penalty level. However, there is currently no theoretically founded guidance for this choice in the context of high-dimensional time series. Instead, one resorts to selecting the penalty level in an ad hoc manner using, e.g., information criteria or cross-validation. We resolve this problem by considering estimation of the perhaps most commonly employed multivariate time series model, the linear vector autoregressive (VAR) model, and propose versions of the Lasso, post-Lasso, and square-root Lasso estimators with penalization chosen in a \emph{fully data-driven} way. The theoretical guarantees that we establish for the resulting estimation and prediction errors match those currently available for methods based on infeasible choices of penalization. We thus provide a first solution for choosing the penalization in high-dimensional time series models.

\end{abstract} 




\section{Introduction}
Multivariate time series play a fundamental role in
many areas of research. The quintessential approach to modelling these is based on the linear
VAR model; see, e.g., \citet{lutkepohl_new_2005}. In order to fully capture the dynamics of the time series at hand and to decrease the risk of omitting variables, one frequently includes many explanatory variables.
These considerations result in a high-dimensional VAR model wherein the number of parameters can exceed the number of observations. Consequently, the parameters can no longer be estimated by least squares, and as an alternative there has been a surge of research on Lasso-type estimators [\citet{tibshirani_regression_1996}]. For instance, assuming independent and identically distributed (i.i.d.)~Gaussian innovations in the VAR, consistency and oracle inequalities for the Lasso
have been studied by \citet{han_transition_2013}, \citet{basu_regularized_2015}, \citet{kock_oracle_2015} and \citet{davis_sparse_2016}, among
others. The assumption of Gaussian innovations has been
relaxed in several papers, including \citet{song_large_2011}, \citet{wong_lasso_2020} and \citet{miao_var_2022}. In
particular, \citet{wong_lasso_2020} derive consistency results for
the Lasso in VAR models under suitable mixing and moment conditions, whereas \citet{miao_var_2022} derive rates of convergence and oracle properties for factor augmented VAR models,
relying on results for weakly dependent processes [in the sense of
\citet{wu_nonlinear_2005}]. \citet{masini_VAR_2022} establish oracle inequalities
in the case of martingale difference innovations that are not necessarily
mixing. \citet{chernozhukov_lasso_2021} consider Lasso-driven inference for time series and spatial models, while \citet{adamek_lasso_2022} consider inference based on the desparsified Lasso when the data-generating process allows near-epoch dependence. Further related papers are \citet{han_direct_2015}, \citet{guo_var_2016} and \citet{wu_performance_2016}. 

To implement Lasso-type estimators one must choose the penalization level
$\lambda$, and the theoretical results listed above are derived under
a suitable choice of this level. However, the ``right'' choice of this tuning parameter typically depends on \emph{unknown} population quantities such as mixing coefficients or other coefficients quantifying
the dependence structure of the data generating process {[}\citet{wong_lasso_2020},
\citet{babii_timeseries_2022}, \citet{masini_VAR_2022}{]}, the population
covariance matrix of the observed variables and innovations {[}\citet{kock_oracle_2015},
\citet{MEDEIROS_highdim_2016}{]}, the population coefficient matrix
{[}\citet{basu_regularized_2015}{]}, or other quantities depending
on the data generating process [\citet{chernozhukov_lasso_2021}, \citet{adamek_lasso_2022}, \citet{miao_var_2022}]. Consequently, despite the large amount of research on penalized estimation of VAR models, there is currently no firm guidance on how to choose~$\lambda$ in practice. Instead, one often resorts to information criteria [e.g.~\citet{kock_oracle_2015}, \citet{MEDEIROS_highdim_2016} and
\citet{masini_VAR_2022}], cross-validation [e.g. \citet{wong_lasso_2020}, \citet{babii_timeseries_2022} and \citet{miao_var_2022}] or other methods without providing theoretical guarantees for these. Crucially, such choices of penalization do not  necessarily satisfy the conditions imposed in the theoretical results. Therefore, strictly speaking, the theoretical guarantees provided are only valid for fortuitous choices of~$\lambda$. In this paper we resolve this problem by proposing a data-driven way of choosing~$\lambda$ along with prediction and estimation error guarantees for the resulting weighted Lasso estimator of the parameters in large VAR models. As upper bounds on the estimation error play a crucial role in establishing the validity of inference based on debiasing in high-dimensional models [see \citet{javanmard_confidence_2014}, \citet{van_de_geer_asymptotically_2014} and \citet{zhang_confidence_2014}], our results open the door for inference in VAR models with a data-driven choice of~$\lambda$.

The penalization algorithm we study originates from \citet{belloni_sparse_2012}, who consider regressions in high dimensions with independent data. There are several challenges in adapting this algorithm to time series data: First, the analysis of \citet{belloni_sparse_2012} relies on the independence as well as certain high-level conditions on the explanatory variables and their relation to the model errors. In VAR models one cannot impose such conditions as their validity is completely determined by the model (which also generates the explanatory variables). Instead, imposing only primitive conditions, we carefully
take into account the inherent dependence in the VAR. Second, as heavy tails are omnipresent in many time series, we allow for the possibility of certain types of heavy-tailed
innovation distributions, that is, so-called sub-Weibull innovations. To accommodate dependence and heavy tails,
we establish a novel maximal inequality for centered sums of dependent sub-Weibull
random variables, that may be of
independent interest. The validity of the penalty loadings (i.e.~Lasso weights) proposed in the algorithm is then established by proving
that these are close to certain infeasible ideal loadings that are
constructed by means of blocking-based self-normalization. In order to verify that
the ideal loadings are well-behaved, we rely on recent moderate deviation
theory for self-normalized block-sums of weakly dependent processes [in the \citet{wu_nonlinear_2005} functional dependence sense]
derived by \citet{chen_self-normalized_2016} and \citet{gao_refined_2022}.

To alleviate the shrinkage bias introduced by the Lasso, one often refits the parameters of variables selected by the Lasso using least squares. The performance of the resulting post-Lasso depends crucially on the Lasso variable selection. As this selection depends on the chosen level of penalization, the post-Lasso itself depends on this tuning parameter. We show that the post-Lasso following the weighted Lasso implied by our data-driven tuning parameter choice obeys the same performance guarantees as the weighted Lasso. Note, however, that these guarantees are established under an additional assumption (not needed for the weighted Lasso) on certain sparse eigenvalues of the population regressor covariance matrix.

Finally, the square-root Lasso (henceforth:~sqrt-Lasso) of~\cite{belloni2011square} is a popular alternative to the Lasso for i.i.d.~data in the absence of conditional heteroskedasticity. Apart from the work of~\cite{sqrtJTSA}, who study a linear model with fixed regressors but dependent error terms, the theoretical performance of the sqrt-Lasso has not been studied for time series and, in particular, it is not clear how to choose its tuning parameter. We resolve this problem by providing a fully data-driven implementation of the sqrt-Lasso with performance guarantees matching those of the other estimators studied in the absence of conditional heteroskedasticity.

\subsection*{Outline}

In Section \ref{sec:Model-and-Estimation} we present the weighted Lasso estimator and describe the data-driven tuning parameter selection. Section \ref{sec:Ass+Res} presents the assumptions and performance guarantees for the weighted Lasso. Sections~\ref{sec:postlasso} and \ref{sec:SqrtLasso} cover the post- and sqrt-Lasso, respectively. Sections \ref{sec:Simulations} and \ref{sec:Empirical-Illustration} contain simulations and an empirical illustration, respectively. Proofs, technical lemmas, and additional simulation output are contained in the supplementary appendices.

\subsection*{Notation}
For $k\in\N$, we write
$[k]:=\{1,\dotsc,k\}$.~For $\bdelta\in\R^{k}$, we denote its $\ell_{r}$-norm, by $\norm{\bdelta}_{\ell_{r}}:=(\sum_{j=1}^k|\delta_j|^r)^{1/r},r\in[1,\infty)$, and $\norm{\bdelta}_{\ell_{\infty}}:=\max_{j\in[k]}|\delta_j|$. For~$\mathcal{S}\subseteq [k]$ non-empty,~$\bdelta_\mathcal{S}\in\R^{|\mathcal{S}|}$ is the subvector of~$\bdelta$ picked out by by~$\mathcal{S}$. When applied to a real matrix $\bfA$, the aforementioned norms are understood as the induced (operator) norms. We use $\bfA_{k:k+l,m:m+n}$ to denote the submatrix picked out by rows $k$ through $k+l$ and columns $m$ through $m+n$. For a square matrix $\bfA\in\mathbb{R}^{k\times k}$,
we denote its spectrum (set of eigenvalues) by~$\Lambda(\bfA)$ and write $\rho(\bfA):=\max\{|\lambda|:\lambda\in\Lambda(\bfA)\}$ for its spectral radius (with
$\left|\cdot\right|$ the complex modulus). If $\bfA\in\R^{k\times k}$ is symmetric, we write $\Lambda_{\min}(\bfA)$
and $\Lambda_{\max}(\bfA)$ for the minimum and maximum eigenvalues, respectively. 

For a random scalar $X$, we denote its $L_r$-norm by $\norm{X}_{r}:=(\mathrm{E}[|X|^{r}])^{1/r}$, $r\in[1,\infty)$, with $\E[\cdot]$ denoting the expectation operator.
For $\alpha\in(0,\infty)$, we define the sub-Weibull($\alpha$) norm of $X$ as $\norm{X}_{\psi_{\alpha}}:=\sup_{r\in[1,\infty)}r^{-1/\alpha}\norm{X}_{r}$. 
A random scalar $X$ is said to be sub-Weibull$(\alpha)$ if $\norm{X}_{\psi_{\alpha}}<\infty$.\footnote{One may show that the space of sub-Weibull$(\alpha)$ random variables is complete with respect to $\left\Vert \cdot \right\Vert _{\psi_{\alpha}}$. Consequently, it holds that the norm $\left\Vert \cdot \right\Vert _{\psi_{\alpha}}$ is countably sub-additive, which we will use repeatedly.}
For a $k$-dimensional random vector $\bX$, we define its joint sub-Weibull($\alpha$)
norm $\Vert \bX\Vert _{\psi_{\alpha}}:=\sup\{\Vert \bu^{\top}\bX\Vert _{\psi_{\alpha}}:\left\Vert \bu\right\Vert _{\ell_{2}}=1\}$
and call $\bX$ jointly sub-Weibull($\alpha$) if $\left\Vert \bX\right\Vert _{\psi_{\alpha}}<\infty$.

For non-random numbers $a_n$ and positive numbers $b_n,n\in\N,$ we write
$a_n=o(1)$ if $a_n\to0$ as~$n\to\infty$, and $a_n\lesssim b_n,$ if the sequence $a_n/b_n$ is
bounded by a constant. For random variables $V_n$ and positive numbers $b_n,$ we write
$V_n\lesssim_{\P} b_n,$ if the sequence $V_n/b_n$ is bounded in probability. We take $n\geqslant 3$ and $p\geqslant 2$ throughout the manuscript and reserve the word ``constant'' for non-random quantities that do not depend on $n$.

\section{Model and Penalization Algorithm\label{sec:Model-and-Estimation}}
We study the $q$\textsuperscript{th}-order $(q\in\N)$ VAR model given by
\begin{equation}
\bY_{t}=\sum_{j=1}^{q}\bTheta_{0j}\bY_{t-j}+\bepsilon_{t},\quad t\in\mathbb{Z},\label{eq:VARvector}
\end{equation}
yielding a stochastic process $\{\bY_t\}_{t\in\Z}$ which is (strictly) stationary under the assumptions in Section~\ref{sec:Ass+Res}. Here~$\bY_{t}:=(Y_{t,1},\dotsc,Y_{t,p})^{\top}$ is a random
vector of length $p$, $\{\bTheta_{0j}\}_{j=1}^q$ are $p\times p$ (unknown)
coefficient matrices, and with  $\bepsilon_{t}:=(\varepsilon_{t,1},\dotsc,\varepsilon_{t,p})^{\top}$, $\{\bepsilon_t\}_{t\in\Z}$
is a sequence of innovations.
Given observations $\{\bY_{t}\}_{t=-(q-1)}^{n}$ from (\ref{eq:VARvector}),
the objective is to estimate $\{\bTheta_{0j}\}_{j=1}^q$, while
allowing (but not requiring) that the number of elements~$p^{2}q$ in $\{\bTheta_{0j}\}_{j=1}^q$
is larger than the (effective) sample size
$n$. Throughout we take $p\geqslant2$ and $n\geqslant3$.\footnote{We consider the VAR model in (\ref{eq:VARvector}), allowing the number
of output variables $p$, the (common) distribution of $\bepsilon_{t}$,
and, hence, that of $\bY_{t}$ to depend on the sample size $n$. That is, we consider an array $\{\{\bY_{t}^{(n)}\}_{t\in\mathbb{Z}}\}_{n\in\N}$ of stochastic processes, each process $\{\bY_{t}^{(n)}\}_{t\in\mathbb{Z}}$ presumed strictly stationary,
in which each $\bY_{t}^{(n)}=(Y_{t,1}^{(n)},\dotsc,Y_{t,p_{n}}^{(n)})^\top,t\in\Z$, is a random
element of $\R^{p_{n}}$. To ease notation, we henceforth suppress the
$n$ superscript.}

The process in (\ref{eq:VARvector}) may be written in companion form,
\begin{equation}
\bZ_{t}=\widetilde{\bTheta}_{0}\bZ_{t-1}+\widetilde{\bepsilon}_{t},\quad t\in\Z,\label{eq:VARcompanion}
\end{equation}
for
\[
\underbrace{\bZ_{t}}_{pq\times1}:=\begin{pmatrix}
\bY_{t}\\
\bY_{t-1}\\
\vdots\\
\bY_{t-(q-1)}
\end{pmatrix},
\quad
\underbrace{\widetilde{\bTheta}_{0}}_{pq\times pq}:=\begin{pmatrix}\bTheta_{01} & \bTheta_{02} & \cdots &  & \bTheta_{0q}\\
\mathbf{I}_{p} & \boldsymbol{0}_{p\times p} & \cdots &  & \boldsymbol{0}_{p\times p}\\
\boldsymbol{0}_{p\times p} & \mathbf{I}_{p} & \boldsymbol{0}_{p\times p} &  & \vdots\\
\vdots & \ddots & \ddots & \ddots \\
\boldsymbol{0}_{p\times p} & \cdots & \boldsymbol{0}_{p\times p} & \mathbf{I}_{p} & \boldsymbol{0}_{p\times p}
\end{pmatrix},
\quad\underbrace{\widetilde{\bepsilon}_{t}}_{pq\times1}:=\begin{pmatrix}\bepsilon_{t}\\
\boldsymbol{0}_{p\times 1}\\
\vdots\\
\boldsymbol{0}_{p\times 1}
\end{pmatrix},
\]
and $\mathbf{I}_p\in\R^{p\times p}$ being the identity matrix. In the special case of $q=1$, we interpret $\widetilde{\bTheta}_{0}$ as $\bTheta_{01}$ and $\widetilde{\bepsilon}_{t}$ as $\bepsilon_{t}$.
For $i\in[p]$, the $i$\textsuperscript{th} row in (\ref{eq:VARvector})
is given by
\begin{equation}
Y_{t,i}=\bZ_{t-1}^{\top}\bbeta_{0i}+\varepsilon_{t,i},\label{eq:VARequation}
\end{equation}
where the vector $\bbeta_{0i}^{\top}$ of length $pq$ is the $i$\textsuperscript{th} row of the companion
matrix $\widetilde{\bTheta}_{0}$. In a low-dimensional setting where~$\bbeta_{0i}^{\top}$ is of fixed length (and short), estimation can be
done by equationwise least squares with
the same $n\times pq$ regressor matrix
$\bfX:=[\bZ_{0}:\cdots:\bZ_{n-1}]^{\top}$ for all $i\in[p]$. However, when $pq>n$, the $pq\times pq$ Gram
matrix $\bfX^{\top}\bfX$ has reduced rank, and Lasso-type estimators have been studied as an alternative under various sparsity assumptions. Nevertheless, as discussed in the introduction, the theoretical guarantees hitherto established for these are valid only for specific tuning parameter choices depending on unknown population quantities. As a result---despite the sensitivity of shrinkage estimators to the tuning parameter choice---one must currently resort to choosing this parameter by methods without theoretical guarantees when implementing shrinkage estimators in time series models. We provide a solution to this problem in the context of a weighted Lasso by proposing an algorithm for tuning parameter selection and  explicitly incorporate the data-driven tuning parameter choice into our theoretical guarantees.

A weighted Lasso estimator satisfies
\begin{equation}
\widehat{\bbeta}_{i}:=\widehat{\bbeta}_{i}(\lambda,\widehat{\bUpsilon}_{i})\in\argmin_{\bbeta\in\R^{pq}}\left\{ \widehat Q_i(\bbeta) +\frac{\lambda}{n}\Vert\widehat{\bUpsilon}_{i}\bbeta\Vert_{\ell_{1}}\right\},\label{eq:LASSOVector}
\end{equation}
where~$\widehat Q_i$ is the sample average squared error loss function
\begin{equation}\label{eq:Qihat}
\widehat Q_i(\bbeta):=\frac{1}{n}\sum_{t=1}^{n}\left(Y_{t,i}-\bZ_{t-1}^{\top}\bbeta\right)^{2},\quad\bbeta\in\R^{pq},
\end{equation}
for $i\in [p]$, $\lambda\in(0,\infty)$ is a penalty
level and each $\widehat{\bUpsilon}_{i}:=\mathrm{diag}(\widehat{\upsilon}_{i,1},\dotsc,\widehat{\upsilon}_{i,pq})$
is a diagonal matrix of data-dependent penalty loadings $\widehat\upsilon_{i,j}\in(0,\infty),j\in[pq]$. Both the penalty level and loadings will be specified in Algorithm \ref{alg:Data-Driven-Penalization} below. Although the minimization problem in (\ref{eq:LASSOVector}) may have multiple solutions, we sometimes refer to such a $\widehat{\bbeta}_{i}$ as \textit{the} Lasso estimator. The results established below apply to any such (measurable) minimizer.\footnote{By the measurable selection theorem 6.7.22~in \citet{pfanzagl},~$\widehat{\bbeta}_{i}$ in~\eqref{eq:LASSOVector} can be chosen measurable (provided the~$\widehat{\bUpsilon}_{i}$ are chosen measurable).}

Note that much of the literature on the
Lasso and its variants employs the same penalty loadings for all regressors (e.g.~$\widehat{\upsilon}_{i,j}\equiv1$).
Equal weighting implicitly treats the regressors as either being on
the same scale or having been brought onto the same scale by some
preliminary transformation of the data, which is then typically abstracted
from in the theoretical analysis. Taking serious the effect of such a preliminary transformation is not trivial as it can alter the already intricate dependence structure of the time series at hand. By incorporating data-dependent loadings, our algorithm is not subject to this caveat.

Adapting \citet[Algorithm A.1]{belloni_sparse_2012} for independent data to our setting,
we consider the following data-driven penalization. Let
\begin{equation}
c:=1.1\quad\text{and}\quad\gamma_{n}=0.1/\ln(\max\{n,pq\})\label{eq:tuning_c_and_gamma}
\end{equation}
and let $K\in\mathbb{N}_{0}$ denote a fixed number of loading updates.\footnote{Although Theorem \ref{thm:Rates-for-LASSO-data-driven-loadings} is valid for all choices of~$c>1$, our simulations (Section \ref{sec:Simulations}) indicate that choices just above~$1$ lead to the best performance in practice. Furthermore, as long as~$c$ is close to~$1$, the performance of Algorithm \ref{alg:Data-Driven-Penalization} is not sensitive to the exact choice. This is documented in Figure~\ref{fig:cchoice} in Section \ref{sec:appsim} of the supplementary appendix, which also shows that the same remarks on the choice of~$c$ are valid for the post- and sqrt-Lasso.~$c=1.1$ was suggested in~\cite{belloni_sparse_2012}.} Following \citet{belloni_sparse_2012}, we fix $K$ at 15 in our simulations and
empirical illustration. To state Algorithm \ref{alg:Data-Driven-Penalization} below, let $\Phi$ be the standard Gaussian cumulative distribution function.
\begin{mdframed}
\begin{lyxalgorithm}
[\textbf{Data-Driven Penalization}]\label{alg:Data-Driven-Penalization}\ \\
\textbf{Initialize}: Specify the penalty level~$\lambda$ in~\eqref{eq:LASSOVector} as
\begin{equation}
\lambda^{\ast}_n := 2c\sqrt{n}\Phi^{-1}\big(1-\gamma_n/(2p^{2}q)\big),\label{eq:PenaltyLevelPractice}
\end{equation}
and specify the initial penalty loadings as 
\begin{equation}
\widehat{\upsilon}_{i,j}^{\left(0\right)} := \sqrt{\frac{1}{n}\sum_{t=1}^{n}Y_{t,i}^{2}Z_{t-1,j}^{2}},\quad i\in[p],\quad j\in[pq].\label{eq:PenaltyLoadingsInitial}
\end{equation}
Use $\lambda^{\ast}_n$ and $\widehat{\bUpsilon}_{i}^{\left(0\right)}:=\mathrm{diag}(\widehat{\upsilon}_{i,1}^{\left(0\right)},\dotsc,\widehat{\upsilon}_{i,pq}^{\left(0\right)})$
to compute a Lasso estimate $\widehat{\bbeta}_{i}^{\left(0\right)}:=\widehat{\bbeta}_{i}(\lambda^{\ast}_n,\widehat{\bUpsilon}_{i}^{\left(0\right)})$
via (\ref{eq:LASSOVector}) for each $i\in[p]$. Store the residuals $\widehat{\varepsilon}_{t,i}^{\left(0\right)}:=Y_{t,i}-\bZ_{t-1}^{\top}\widehat{\bbeta}_{i}^{\left(0\right)},t\in[n],i\in[p]$, and set $k=1$.

\noindent\textbf{Update:}
While $k\leqslant K$, specify the penalty loadings as
\begin{equation}
\widehat{\upsilon}_{i,j}^{\left(k\right)}:=\sqrt{\frac{1}{n}\sum_{t=1}^{n}(\widehat{\varepsilon}_{t,i}^{\left(k-1\right)})^2Z_{t-1,j}^{2}},\quad i\in[p],\quad j\in[pq].\label{eq:PenaltyLoadingsRefined}
\end{equation}
Use $\lambda_n^{\ast}$ and $\widehat{\bUpsilon}_{i}^{\left(k\right)}:=\mathrm{diag}(\widehat{\upsilon}_{i,1}^{\left(k\right)},\dotsc,\widehat{\upsilon}_{i,pq}^{\left(k\right)})$
to compute a Lasso estimate $\widehat{\bbeta}_{i}^{\left(k\right)}:=\widehat{\bbeta}_{i}(\lambda_n^{\ast},\widehat{\bUpsilon}_{i}^{\left(k\right)})$
via (\ref{eq:LASSOVector}) for each $i\in[p]$. Store the residuals $\widehat{\varepsilon}_{t,i}^{\left(k\right)}:=Y_{t,i}-\bZ_{t-1}^{\top}\widehat{\bbeta}_{i}^{\left(k\right)},t\in[n],i\in[p]$,
and increment $k\leftarrow k+1.$
\end{lyxalgorithm}
\end{mdframed}
\noindent Note that Algorithm \ref{alg:Data-Driven-Penalization} does not require any knowledge of the degree of dependence in the process~$\{\bY_{t}\}_{t\in\Z}$ (as quantified by mixing coefficients or other unknown population quantities). This feature of the penalization is in contrast to current recommendations in the literature. 

To discuss error rates for the implied Lasso estimator, let $\left\Vert \cdot\right\Vert _{2,n}$
be the prediction norm
\[
\Vert \bdelta\Vert_{2,n}:=\sqrt{\frac{1}{n}\sum_{t=1}^{n}\left(\bZ^{\top}_{t-1}\bdelta\right)^{2}}
\]
of the vector $\bdelta\in\R^{pq}$,
and let
\begin{equation}
s:=\max_{i\in[p]}\sum_{j=1}^{pq}\textbf{1}(\beta_{0i,j}\neq0),\label{eq:def_s_sparsity}
\end{equation}
denote the sparsity number, defined as the largest cardinality of the
support of $\bbeta_{0i}$ across $i\in[p]$. (Without loss of generality, we
take $s\geqslant 1$.) Then, under the primitive conditions given in Section~\ref{sec:Ass+Res}, any
Lasso estimators $\widehat\bbeta_i:=\widehat\bbeta_{i}(\lambda_n^{\ast},\widehat\bUpsilon_i^{(K)}),i\in[p]$, in (\ref{eq:LASSOVector}) based on the penalty level $\lambda_n^{\ast}$ in (\ref{eq:PenaltyLevelPractice})
and the (final) penalty loadings $\widehat\bUpsilon_i^{(K)}$ arising from Algorithm \ref{alg:Data-Driven-Penalization} satisfy
\begin{align}
\max_{i\in[p]}\Vert\widehat{\bbeta}_{i}-\bbeta_{0i}\Vert_{2,n} & \lesssim_{\mathrm{P}}\sqrt{\frac{s\ln\left(pn\right)}{n}},\label{eq:lasso_rates_1}\\
\max_{i\in[p]}\Vert\widehat{\bbeta}_{i}-\bbeta_{0i}\Vert_{\ell_{1}} & \lesssim_{\mathrm{P}}\sqrt{\frac{s^{2}\ln\left(pn\right)}{n}}\quad\text{and}\label{eq:lasso_rates_2}\\
\max_{i\in[p]}\Vert\widehat{\bbeta}_{i}-\bbeta_{0i}\Vert_{\ell_{2}} & \lesssim_{\mathrm{P}}\sqrt{\frac{s\ln\left(pn\right)}{n}}.\label{eq:lasso_rates_3}
\end{align}
The formal statement, including a set of sufficient conditions, is given
in Theorem \ref{thm:Rates-for-LASSO-data-driven-loadings} below. Importantly, the performance guarantees in~\eqref{eq:lasso_rates_1}--\eqref{eq:lasso_rates_3} match those currently available based on \emph{infeasible} penalty level choices. 

\begin{rem}[Including Intercepts]
In applications one typically includes model intercepts in an unpenalized manner. In Section \ref{subsec:IncludingIntercepts}, we modify Algorithm \ref{alg:Data-Driven-Penalization} for this purpose.\hfill$\diamondsuit$
\end{rem}

\section{Assumptions and Performance Guarantees for Algorithm \ref{alg:Data-Driven-Penalization}}\label{sec:Ass+Res}

For the remainder of the paper we invoke (a subset of) the following assumptions.
\begin{assumption}[\textbf{Innovations}] \label{assu:Innovations} There is a $\overline p\in\N$, a constant $\overline q\in\N_0$, a measurable mapping $\bF:\R^{\overline p\cdot\overline q}\to\R^p$, an i.i.d.~process $\{\bseta_t\}_{t\in\Z}$, each $\bseta_t$ taking values in $\R^{\overline p}$, and constants $a_1,a_2,\alpha\in(0,\infty)$ such that: \begin{inparaenum}[(1)]\item\label{enu:EpsIID} $\bepsilon_t=\bF(\bseta_{t-\overline q},\dotsc,\bseta_t)$ for each $t\in\Z$.
\item\label{enu:EpsSubWeibullBeta} $\Vert \bepsilon_{0}\Vert _{\psi_{\alpha}}\leqslant a_{1}$. \item\label{enu:EigValsOfSigmaEpsBddAwayFromZero} With $\mathcal{F}^{\bseta}_t$ denoting the $\sigma$-field generated by $\{\bseta_u\}_{u=-\infty}^{t}$, for each $t\in\Z$, $\E[\bepsilon_t \mid \mathcal{F}^{\bseta}_{t-1}]=\mathbf{0}_p$ almost surely and $\min_{i\in[p]}\E[\varepsilon_{t,i}^2 \mid \mathcal{F}^{\bseta}_{t-1}]\geqslant a_{2}$ almost surely. 
\end{inparaenum}
\end{assumption}
Assumption \ref{assu:Innovations}.\ref{enu:EpsIID} implies that the innovations $\bepsilon_t$ are causal and strictly stationary but allows them to be $\overline{q}$-dependent. Here $\overline{q}=0$ corresponds to independent errors, as imposed in, e.g., \citet{kock_oracle_2015}, \citet[Example 1]{wong_lasso_2020} and \citet{miao_var_2022}. Assumption \ref{assu:Innovations}.\ref{enu:EpsSubWeibullBeta} states that $\bepsilon_{0}$ belongs to the sub-Weibull family of distributions with tail parameter (at least) $\alpha$, as imposed in \citet{wong_lasso_2020} and \citet{masini_VAR_2022}.
Special cases are the sub-exponential $(\alpha=1)$ and sub-Gaussian
$(\alpha=2)$ families. If $\bepsilon_{0}$
is Gaussian---as imposed in, e.g., \citet{han_transition_2013}, \citet{basu_regularized_2015}, \citet{kock_oracle_2015} and \citet{davis_sparse_2016}---then Assumption \ref{assu:Innovations}.\ref{enu:EpsSubWeibullBeta}
holds if, in addition,~$\Lambda_{\max}(\bSigma_{\bepsilon})$ is bounded from above, with $\bSigma_{\bepsilon}$ being the covariance matrix of~$\bepsilon_{0}$.
Assumption \ref{assu:Innovations}.\ref{enu:EigValsOfSigmaEpsBddAwayFromZero} parallels part of \citet[Assumption A2]{masini_VAR_2022} and  \citet[Assumption A.1(i)]{miao_var_2022}.

\begin{assumption}
[\textbf{Companion Matrix}]\label{assu:Companion}\textbf{ }\begin{inparaenum}[(1)]\item\label{enu:SpectralRadius}~The spectral
radius $\rho(\widetilde{\bTheta}_{0})<1$. 
\item\label{enu:RowNormDecay} There are constants $b_1, b_2\in(0,\infty)$ and $\tau\in(0,1]$ such that
$\max_{i\in[p]}\Vert(\widetilde{\bTheta}_{0}^{h})_{i,1:p}\Vert_{\ell_{2}}\leqslant b_{1}\mathrm{e}^{-b_{2}h^\tau}$
for all $h\in\mathbb{N}$. 
\item\label{enu:lag-order} The lag-order $q$ is constant.
\end{inparaenum}
\end{assumption}

Assumption \ref{assu:Companion}.\ref{enu:SpectralRadius} is standard and combined with Assumption \ref{assu:Innovations} ensures that the process $\{\bY_{t}\}_{t\in\Z}$ in (\ref{eq:VARvector}) [as well as
$\{\bZ_{t}\}_{t\in\Z}$ in (\ref{eq:VARcompanion})] is strictly stationary. In particular,
$\{\bY_{t}\}_{t\in\Z}$
 has the moving
average representation
\[
\bY_{t}=\sum_{\ell=0}^{\infty}(\widetilde{\bTheta}_{0}^{\ell})_{1:p,1:p}\bepsilon_{t-\ell},
\]
where the series converges absolutely almost surely (and in sub-Weibull norm).

Assumption
\ref{assu:Companion}.\ref{enu:RowNormDecay} is milder than \citet[Assumption A1]{masini_VAR_2022}, which imposes an exponential decay on the $\ell_1$-norms of the rows of $(\widetilde{\bTheta}_{0}^{h})_{1:p,1:p}$ as $h$ grows (as opposed to our decay placed on the~$\ell_2$-norms). Likewise, the assumption is also milder than \citet[Assumption A.1(vi)]{miao_var_2022}, which imposes an exponential decay on the spectral norm of $(\widetilde{\bTheta}_{0}^{h})_{1:p,1:p}$ (exceeding the row-wise $\ell_2$-norms).
\footnote{See \Citet[Lemma 4]{masini_VAR_2022} for
further conditions sufficient for Assumption \ref{assu:Companion}.\ref{enu:RowNormDecay}.} 

Assumption \ref{assu:Companion}.\ref{enu:lag-order} states that the lag-order $q$ is constant. This appears to be reasonable for most practical purposes as~$q$ is often thought of as small relative to~$n$.\footnote{This point was also made by \citet{kock_oracle_2015}, although their analysis allows for growing $q$.} Indeed, in our application in Section \ref{sec:Empirical-Illustration} we have~$n=758$ and as a result of monthly sampling~$q=12$ often suffices. We use that $q$ does not depend on $n$ to exploit recent moderate deviation results in \citet{gao_refined_2022} for self-normalized sums based on geometric moment contracting processes. Assumption \ref{assu:Companion} therefore seems hard to relax with currently available methods.

Assumptions \ref{assu:Innovations} and \ref{assu:Companion} suffice to show that the population covariance matrix $\bSigma_{\bZ}:=\E[\bZ_{0}\bZ_{0}^\top]$ of $\bZ_0$ exists (in $\R^{pq\times pq}$).\footnote{See Lemma \ref{lem:ZandEpsZNormBnds} in the appendix for details.} We further assume:
\begin{assumption}
[\textbf{Covariance}]\label{assu:CovarianceZ}
There is a constant $d\in(0,\infty)$ such that
$\Lambda_{\min}(\bSigma_{\bZ})\geqslant d$.
\end{assumption}
Assumption \ref{assu:CovarianceZ} ensures that $\bSigma_{\bZ}$ has full rank and is implied by, e.g., \citet[Assumption A.1]{miao_var_2022}. As demonstrated in \citet[Section 3]{basu_regularized_2015}, Assumption \ref{assu:CovarianceZ} follows from additional (boundedness) conditions on~$\widetilde{\bTheta}_{0}$, although these are not necessary.\footnote{For the purpose of Theorem \ref{thm:Rates-for-LASSO-data-driven-loadings}, we note that Assumption \ref{assu:CovarianceZ} is used in the proof of Lemma \ref{lem:Restricted-Eigenvalue-Bound}. Inspection of its proof reveals that one can replace Assumption \ref{assu:CovarianceZ} by a milder requirement that only certain \textit{restricted} eigenvalues of $\bSigma_{\bZ}$ are bounded away from zero. Restricted eigenvalue conditions appear in, e.g., \citet{kock_oracle_2015} and \citet[DESIGN(3)a]{MEDEIROS_highdim_2016}.}
\begin{assumption}
[\textbf{Growth and Sparsity}]\label{assu:RowSparsity} For  constants $\alpha\in(0,\infty)$  and $\tau\in(\textstyle{\frac{1}{2}},1]$ satisfying respectively Assumptions \ref{assu:Innovations}.\ref{enu:EpsSubWeibullBeta} and \ref{assu:Companion}.\ref{enu:RowNormDecay}, we have $(\ln p)^{C(\alpha,\tau)}=o(n)$ and $s^2(\ln(pn))^{1/\tau+\widetilde{C}(2,\alpha)}=o(n)$, where $s$ is given in (\ref{eq:def_s_sparsity}), $\widetilde{C}(x,\alpha):=\max\{(2x/\alpha),(x+\alpha)/\alpha\}$ and
\[
C(\alpha,\tau):=\max\left\{\frac{\widetilde{C}(4,\alpha)(1+4\tau)}{4\tau-1},\frac{1+4\tau}{4\tau-2},\frac{1+4\tau}{\tau}\right\}.
\]
\end{assumption}
Assumption \ref{assu:RowSparsity}
imposes sparsity on each row $\bbeta_{0i}^\top$ of the horizontally concatenated parameter matrices
$[\bTheta_{01}:\cdots:\bTheta_{0q}]$ as in, e.g., \citet{kock_oracle_2015},
\citet{chernozhukov_lasso_2021}, and \citet{masini_VAR_2022}.~To help digest the growth and sparsity conditions in Assumption \ref{assu:RowSparsity}, note that in the special case of sub-Gaussian innovations $(\alpha=2)$ and exponential decay in the  $\ell_2$-norms of the rows of $(\widetilde{\bTheta}_{0}^{h})_{1:p,1:p}$ as $h$ grows $(\tau=1)$, the exponents equal $C(\alpha,\tau)=6.66\dotsc$ and $1/\tau+\widetilde{C}(2, \alpha)=3$.

We can now state the properties of the weighted Lasso estimator with data-driven tuning parameter selection based on Algorithm~\ref{alg:Data-Driven-Penalization}. 

\begin{thm}[\textbf{Convergence Rates for Lasso with Data-Driven Penalization}]\label{thm:Rates-for-LASSO-data-driven-loadings}
Let Assumptions \ref{assu:Innovations}, \ref{assu:Companion}, \ref{assu:CovarianceZ} and \ref{assu:RowSparsity} hold and fix~$K\in\N_0$. Then Lasso estimators $\widehat\bbeta_i(\lambda_n^\ast,\widehat\bUpsilon_i^{(K)}),i\in[p]$, arising from the data-driven penalization in Algorithm \ref{alg:Data-Driven-Penalization} satisfy the rates in~(\ref{eq:lasso_rates_1})--(\ref{eq:lasso_rates_3}).
\end{thm}

The estimation and prediction error rates guaranteed by Theorem~\ref{thm:Rates-for-LASSO-data-driven-loadings} match those currently available in the literature building on infeasible penalty parameter choices. Algorithm~\ref{alg:Data-Driven-Penalization} is thus the first data-driven tuning parameter selection method for which theoretical guarantees have been provided in the context of dependent data. Upper bounds on the estimation error are also crucial ingredients in establishing valid inference based on the debiasing/desparsification methodology originating in~\citet{javanmard_confidence_2014}, \citet{van_de_geer_asymptotically_2014} and \citet{zhang_confidence_2014}. This methodology has been used in the context of time series in \citet{adamek_lasso_2022}. Theorem~\ref{thm:Rates-for-LASSO-data-driven-loadings} therefore opens the door to inference in high-dimensional VAR models with data-driven penalty parameter choice. 

In establishing Theorem~\ref{thm:Rates-for-LASSO-data-driven-loadings}, we cannot call upon the large existing body of maximal and concentration inequalities for independent random variables. Instead, we establish a new maximal inequality for sums of dependent random variables (Lemma~\ref{lem:MaximalInequalityFunctionalDependence} in the appendix). This result is then used to show that the penalty loadings of Algorithm~\ref{alg:Data-Driven-Penalization} are close to certain ideal, yet infeasible, blocking-based loadings (Lemma~\ref{lem:AsymptoticValidityDataDrivenPenaltyLoadings}). Note that although these ideal loadings involve blocking, Algorithm~\ref{alg:Data-Driven-Penalization} does not, meaning that no choice of block size is needed in practice. The ideal loadings can, in turn, be used as self-normalizing factors in an application of the recent moderate deviation theorems of \cite{gao_refined_2022} for dependent random variables. Using these results, we show that~$\lambda_n^*$ is a high-probability upper bound on the maximum of such self-normalized sums. Section~\ref{sec:proofsteps} in the appendix explains further steps and challenges involved in establishing Theorem~\ref{thm:Rates-for-LASSO-data-driven-loadings}.

\begin{rem}[Penalty Loading Updates]
The error rates (\ref{eq:lasso_rates_1})--(\ref{eq:lasso_rates_3}) guaranteed by Theorem \ref{thm:Rates-for-LASSO-data-driven-loadings} are valid for any fixed $K\in\N_0$ in Algorithm \ref{alg:Data-Driven-Penalization}. Preliminary and unreported simulations indicate that there are typically finite-sample benefits from updating the loadings (as opposed to no updating, $K=0$). The same exercises indicate that, it suffices to use~$K=5$ updates, as neither the loadings nor the parameter estimates change much for larger~$K$.\hfill$\diamondsuit$
\end{rem}

\section{Performance Guarantees for the Post-Lasso}\label{sec:postlasso}
To alleviate shrinkage bias, the \textit{post-Lasso}~$\widetilde{\bbeta}_{i}(\lambda,\widehat{\bUpsilon}_{i})$, say, uses least squares to refit the coefficients of the variables~$\widehat{T}_i=\widehat{T}_i\del[1]{\widehat{\bbeta}_i(\lambda,\widehat{\bUpsilon}_i)}:=\supp\del[1]{\widehat{\bbeta}_i(\lambda,\widehat{\bUpsilon}_i)}$ selected by a Lasso with generic tuning parameter~$\lambda$ (and generic penalty loadings~$\widehat{\bUpsilon}_i$) in each equation. Here~$\supp(\bm{b})=\cbr[0]{j\in[pq]:b_j\neq 0}$ denotes the support of~$\bm{b}\in\R^{pq}$. Therefore, the not necessarily unique, post-Lasso satisfies\footnote{The post-Lasso~$\widetilde{\bbeta}_{i}$ need not be unique since: 
\begin{enumerate}[(i)]
    \item The first step Lasso estimator need not be unique and the elements of~$\mathrm{argmin}_{\bbeta\in\R^{pq}}\{\widehat Q_i(\bbeta) +(\lambda/n)\Vert\widehat{\bUpsilon}_{i}\bbeta\Vert_{\ell_{1}}\}$
    need not select the same variables, cf.~\cite{tibshirani2013lasso}. As a results, the post-Lasso depends on which Lasso estimator one picks.
    \item Even for a given~$\widehat T_i$ it need not be the case that $\mathrm{argmin}_{\bbeta\in\R^{pq}:\supp(\bbeta)\subseteq \widehat{T}_i}\widehat{Q}_{i}(\bbeta)$
    is unique as~$(\bfX_{\widehat T_i})^\top \bfX_{\widehat T_i}$ may not have full rank.
\end{enumerate}
Nevertheless, just like Theorem~\ref{thm:Rates-for-LASSO-data-driven-loadings}, Theorem~\ref{thm:postLasso} below applies to any post-Lasso estimator and guarantees uniqueness with probability tending to one.}
\begin{equation}\label{eq:PostLassom}
\widetilde{\bbeta}_{i}\in\argmin_{\mathclap{\substack{\bbeta\in\R^{pq}:\\ \supp(\bbeta)\subseteq \widehat{T}_i}}} \widehat{Q}_{i}\del[0]{\bbeta},\qquad i\in[p].    
\end{equation}
All coefficients classified as zero by the Lasso are also set to zero by the post-Lasso, i.e.~$\widetilde{\bbeta}_{i,\widehat{T}_i^c}=\bm{0}_{|\widehat{T}_i^c|}$. In particular,~$\widetilde{\bbeta}_{i}=\bm{0}_{pq}$ if~$\widehat{T}_i=\emptyset$.

Since the set of variables selected by the ``first stage'' Lasso~$\widehat{T}_i=\widehat{T}_i\del[1]{\widehat{\bbeta}_i(\lambda,\widehat{\bUpsilon}_i)}$ depends on the choice of~$\lambda$, the performance of the ``second stage'' post-Lasso also depends on this tuning parameter. Thus, to obtain a data-driven implementation of the post-Lasso with theoretical performance guarantees, one needs a data driven tuning parameter selection for the first stage Lasso with performance guarantees. For the penalty choice~$\lambda_n^*$ and loadings~$\widehat{\bUpsilon}_i^{(K)}$ stemming from Algorithm~\ref{alg:Data-Driven-Penalization} we have provided such guarantees for the first stage Lasso~$\widehat{\bbeta}_i(\lambda_n^*,\widehat{\bUpsilon}_i^{(K)})$ in Theorem~\ref{thm:Rates-for-LASSO-data-driven-loadings} for any fixed~$K\in\N_0$. 

Writing~$\bfX_{I}\in\R^{n\times|I|}$ for the matrix consisting of all columns of~$\bfX$ with column index in~$I\subseteq [pq]$ and letting~$\bm{y}_i=(Y_{1,i},\hdots,Y_{n,i})^\top$, the post-Lasso estimator of the coefficients of the variables retained by the Lasso can be uniquely written as
\begin{align}\label{eq:PLdef}
	\widetilde{\bbeta}_{i,\widehat T_i}
	=
	\del[1]{(\bfX_{\widehat T_i})^\top \bfX_{\widehat T_i}}^{-1}(\bfX_{\widehat T_i})^\top \bm{y}_{i},
\end{align}
if~$\widehat T_i\neq \emptyset$ and~$(\bfX_{\widehat T_i})^\top \bfX_{\widehat T_i}$ has full rank $(|\widehat{T}_i|)$. To state Theorem~\ref{thm:postLasso} below, let $\|\bdelta\|_{\ell_0}:=|\supp(\bdelta)|$ be the $\ell_0$-``norm'' of $\bdelta$, denote
\begin{align}\label{eq:Dm}
	\mathfrak{D}(m):=\cbr[1]{\bdelta\in\R^{pq}:\Vert\bdelta\Vert_{\ell_0}\leqslant m, \Vert\bdelta\Vert_{\ell_2}= 1}\qquad\text{for } m\in[1,\infty),
\end{align}
and define the~\emph{$m$-sparse eigenvalue} of a symmetric~$pq\times pq$ matrix~$\mathbf{A}$ as
\begin{equation}
	\phi_{\max}(m,\mathbf{A}):=\max_{\bdelta\in\mathfrak{D}(m)}\bdelta^\top\mathbf{A}\bdelta.\label{eq:largestmsparseeigenvalue}
\end{equation}
Imposing that $\limsup_{n\to\infty}\phi_{\max}(s\ln(n), \bSigma_{\bZ})<\infty$ in addition to Assumptions~\ref{assu:Innovations}--\ref{assu:RowSparsity}, the post-Lasso following the Lasso implemented with the penalty parameter~$\lambda_n^\ast$ and loadings~$\widehat\bUpsilon_i^{(K)}$ from Algorithm~\ref{alg:Data-Driven-Penalization} obeys the following performance guarantees:
\begin{thm}[\textbf{Convergence Rates for Post-Lasso with Data-Driven Penalization}]\label{thm:postLasso}
	Let Assumptions \ref{assu:Innovations}, \ref{assu:Companion}, \ref{assu:CovarianceZ} and \ref{assu:RowSparsity} hold and fix~$K\in\N_0$. Assume that~$\limsup_{n\to\infty}\phi_{\max}(s\ln(n), \bSigma_{\bZ})<\infty$. Then post-Lasso estimators~$\widetilde{\bbeta}_{i}(\lambda_n^\ast,\widehat\bUpsilon_i^{(K)}),i\in[p]$, from~\eqref{eq:PostLassom} following Lasso estimators $\widehat\bbeta_i(\lambda_n^\ast,\widehat\bUpsilon_i^{(K)}),i\in[p]$, arising from the data-driven penalization in Algorithm \ref{alg:Data-Driven-Penalization}, satisfy the rates in~\eqref{eq:lasso_rates_1}--\eqref{eq:lasso_rates_3}.
	\end{thm}
	\emph{When~$\lambda$ is chosen appropriately}, Theorem~\ref{thm:postLasso} provides the same asymptotic guarantees for the post-Lasso as Theorem~\ref{thm:Rates-for-LASSO-data-driven-loadings} does for the first stage Lasso. However, note that Theorem~\ref{thm:postLasso} imposes the assumption~$\limsup_{n\to\infty}\phi_{\max}(s\ln(n), \bSigma_{\bZ})<\infty$, which is \emph{not} imposed for the Lasso in Theorem~\ref{thm:Rates-for-LASSO-data-driven-loadings}. An assumption of this type was employed already for independent data in~\citet*[Corollary 1]{belloni_sparse_2012} in order to establish rates of convergence in estimation and prediction error for the post-Lasso. This additional assumption is used to (i) upper bound  the number of irrelevant variables retained by the Lasso and (ii) to control the omitted variable bias of the second stage least squares estimator (as the first stage Lasso may not retain all relevant variables).
 Thus, although the post-Lasso performs well in many of our simulations (see Section \ref{sec:Simulations}) with~$\lambda_n^*$ as dictated by Algorithm \ref{alg:Data-Driven-Penalization}, it does impose an additional assumption compared to the Lasso. Moreover, the same simulations show that post-Lasso can perform worse than the latter if this additional assumption is not satisfied.

\section{Performance Guarantees for the Sqrt-Lasso}\label{sec:SqrtLasso}
In the context of i.i.d.~data the sqrt-Lasso  was proposed by~\cite{belloni2011square} to construct a shrinkage estimator for which a theoretically justifiable tuning parameter choice does not depend on the (unknown) scale of the error term(s).\footnote{The sqrt-Lasso coincides with the \textit{scaled Lasso} of \citet{sun_scaled_2012} after concentrating the error standard deviation out of their objective function.} For this reason, it has become a popular alternative to the Lasso. We propose a data-driven implementation of the sqrt-Lasso for high-dimensional VAR models, which also accounts for the scales of the regressors~$\bm{Z}_{t-1}$ through penalty loadings~$\widehat{\bUpsilon}_{i}$ not present in the original sqrt-Lasso.\footnote{\citet{belloni2011square} treat the regressors as fixed/conditioned on throughout (\textit{ibid.}, p.~793). This view allows one to normalize their scales to any desired values. In the context of time series such normalization is not innocuous, as regressor pre-processing can substantially alter the dependence structure. Hence, care needs to be taken in handling the regressor scales when implementing the sqrt-Lasso for time series.} In particular, we study the following sqrt-Lasso
\begin{align}\label{eq:srtLasso}
	\dot{\bbeta}_{i}:=\dot{\bbeta}_{i}(\lambda,\widehat{\bUpsilon}_{i})\in\argmin_{\bbeta\in\R^{pq}}\left\{ \del[1]{\widehat Q_i(\bbeta)}^{1/2} +\frac{\lambda}{n}\Vert\widehat{\bUpsilon}_{i}\bbeta\Vert_{\ell_{1}}\right\}
\end{align}
and set the penalty level~$\lambda=\lambda_n^*/2$ [cf.~\eqref{eq:PenaltyLevelPractice}] and loadings $\widehat{\bUpsilon}_{i}=\dot{\bUpsilon}=\mathrm{diag}(\dot\upsilon_1,\hdots,\dot\upsilon_{pq})$ with $\dot\upsilon_j:=(n^{-1}\sum_{t=1}^nZ^2_{t-1,j})^{1/2}$, $j\in[pq]$.\footnote{In contrast to the loadings~$\{\widehat{\bUpsilon}_i\}_{i\in[p]}$ from Algorithm \ref{alg:Data-Driven-Penalization}, the sqrt-Lasso loadings~$\dot{\bUpsilon}$ do not depend on~$i$.}
Crucially, the implementation of the proposed version of the sqrt-Lasso does \emph{not} depend on any population unknowns---it is fully data-driven.
\begin{thm}[\textbf{Convergence Rates for Sqrt-Lasso with Data-Driven Penalization}]\label{thm:Rates-for-sqrtLASSO-data-driven-loadings}
Let Assumptions \ref{assu:Innovations}, \ref{assu:Companion}, \ref{assu:CovarianceZ} and \ref{assu:RowSparsity} hold with $\overline q=0$. Then sqrt-Lasso estimators $\dot\bbeta_i(\lambda_n^\ast/2,\dot\bUpsilon),i\in[p]$, arising from~\eqref{eq:srtLasso} satisfy the rates in~(\ref{eq:lasso_rates_1})--(\ref{eq:lasso_rates_3}).
\end{thm}
The proposed sqrt-Lasso~$\dot\bbeta_i(\lambda_n^\ast/2,\dot\bUpsilon)$ that adjusts to the scales of the regressors~$\bZ_{t-1}$ through the choice of loadings~$\dot{\bUpsilon}$ obeys the same performance bounds as the (post-)Lasso. The condition $\overline q=0$ used in Theorem \ref{thm:Rates-for-sqrtLASSO-data-driven-loadings} implies that the $\{\bepsilon_t\}_{t\in\Z}$ are i.i.d.~and, in particular, conditionally homoskedastic given the past outcomes. Conditional homoskedasticity is a crucial element in the analysis of the sqrt-Lasso for even i.i.d.~observations, cf.~\citet{belloni2011square}. Just like for i.i.d.~data, we shall see in the simulations in Section \ref{sec:Simulations} that the sqrt-Lasso can be inferior to the Lasso in the presence of conditional heteroskedasticity.

\section{Simulations\label{sec:Simulations}}
We next explore the finite-sample behavior of the methods studied in a list of experiments.

\subsection{Designs}
Seven experimental designs (labelled A--G) are considered, all of which have zero intercepts. The number of Monte Carlo replications is $1{,}000$ and the effective sample sizes are $n\in\{100, 200, \dotsc, 1000\}$. We consider system sizes $p\in\{16,32,64,128\}$, the largest of which was chosen to mimic the 127 series used in our empirical illustration (Section \ref{sec:Empirical-Illustration}).

Designs \hyperlink{Design A}{A}, \hyperlink{Design B}{B} and \hyperlink{Design C}{C} below are inspired by \citet[Section 5]{kock_oracle_2015}. For these designs, the innovations $\{\bepsilon_{t}\}_{t\in\Z}$ are i.i.d.~mean zero Gaussian with a diagonal covariance matrix $\bSigma_{\bepsilon}=\sigma_{\varepsilon}^2\mathbf{I}_{p}$ and $\sigma_{\varepsilon}=0.1$. The remaining details of these designs are:
\begin{itemize}
\item\textbf{Design \hypertarget{Design A}{A}:} The data-generating process (DGP) is a VAR(1) with diagonal matrix $\bTheta_{01}=(0.5)\mathbf{I}_{p}$
implying a spectral radius $\rho(\bTheta_{01})=0.5$. This setting
is very row sparse as $s=1.$
\item\textbf{Design \hypertarget{Design B}{B}:} The DGP is a VAR(1) with coefficient matrix having the Toeplitz structure $\bTheta_{01,i,j}=(-1)^{|i-j|}(0.4)^{1+|i-j|},(i,j)\in[p]^2$. Since~$s=p$ this design leads to a violation of exact/strong sparsity, but the entries decay exponentially fast in magnitude as one moves away from the
diagonal. It holds that $\rho(\bTheta_{01})=0.9$.
\item\textbf{Design \hypertarget{Design C}{C}:} The DGP is a VAR(4) with both~$\bTheta_{01}$
and~$\bTheta_{04}$  block-diagonal matrices with diagonal blocks of size $4\times4$ with all entries equal
to $0.15$ and $-0.1,$ respectively. $\bTheta_{02}=\bTheta_{03}=\mathbf{0}_{p\times p}$. Thus, $s=8$ and the spectral radius of the companion matrix is~$\rho(\widetilde{\bTheta}_{0})=0.9.$\footnote{We deviate from the $5\times5$ blocks in \citet[Experiment B, p.~333]{kock_oracle_2015}
to ensure that the block sizes are divisors of~$p$.}
\end{itemize}
 As in Design \hyperlink{Design A}{A}, Designs \hyperlink{Design D}{D}, \hyperlink{Design E}{E} and \hyperlink{Design F}{F} all involve a VAR(1) with $\bTheta_{01}=(0.5)\mathbf{I}_p$ and $\sigma_\varepsilon=0.1$. They differ from Design \hyperlink{Design A}{A} in terms of~$\bSigma_{\bepsilon}$, the heavy-tailedness of the~$\bepsilon_t$, or allowing for conditional heteroskedasticity. Design \hyperlink{Design G}{G} studies a local-to-unity variant of Design \hyperlink{Design A}{A}. The details of Designs \hyperlink{Design D}{D}--\hyperlink{Design G}{G} are as follows:
\begin{itemize}
\item\textbf{Design \hypertarget{Design D}{D}:} The $\{\bepsilon_{t}\}_{t\in\Z}$ are i.i.d.~centered Gaussian with covariance matrix  $\bSigma_{\bepsilon,i,j}=\sigma_{\varepsilon}^2\cdot(0.9)^{\mathbf{1}(i\neq j)},(i,j)\in[p]^2$. This covariance structures implies strong correlation among the regressors. Specifically, the largest $m$-sparse eigenvalue $\phi_{\max}(m,\bSigma_{\bY})$ [see \eqref{eq:largestmsparseeigenvalue}] grows linearly with $m$. Thus, $\lim_{n\to\infty}\phi_{\max}(s\ln(n),\bSigma_{\bY})=\infty$, i.e.~an assumption used to establish performance guarantees for the post-Lasso in Theorem \ref{thm:postLasso} is violated. 
\item\textbf{Design \hypertarget{Design E}{E}:} The 
$\{\bepsilon_{t}\}_{t\in\Z}$ are i.i.d.~with $\boldsymbol{\varepsilon}_t\sim(\sigma_{\varepsilon}/\sqrt{5/3})\cdot t_5(\mathbf{0}_p,[(0.9)^{\mathbf{1}(i\neq j)}])$, where $t_5(\mathbf{0}_p,[(0.9)^{\mathbf{1}(i\neq j)}])$ denotes the multivariate Student distribution with five degrees of freedom, all-zero locations, and scales $(0.9)^{\mathbf{1}(i\neq j)},(i,j)\in[p]^2$. This design conflicts with Assumption \ref{assu:Innovations} in that the innovation tails are heavier than sub-Weibull.
\item\textbf{Design \hypertarget{Design F}{F}:} This design investigates the effect of conditional heteroskedasticity. Specifically, let $\{\bseta_t\}_{t\in\Z}$ be i.i.d.~with $\bseta_t\sim \mathrm{N}(\mathbf{0}_p,\mathbf{I}_p)$. Let $\bepsilon_t = \bSigma_{\bepsilon}^{1/2}(\bseta_{t-1}) \bseta_t$, where $\bSigma_{\bepsilon}(\bseta_{t-1}):=\mathrm{diag}(\{\sigma_{\varepsilon,i}^2(\bseta_{t-1})\}_{i\in[p]})$ with conditional standard deviations taking the form
\[
\sigma_{\varepsilon,i}(\bseta_{t-1}):=\sigma_{\varepsilon}\cdot\begin{cases}
\mathrm{e}^{-1.5|\eta_{t-1,i}|+1.5|\eta_{t-1,i+1}|}, & \text{if}\;i\in[p-1],\\
\mathrm{e}^{-1.5|\eta_{t-1,p}|+1.5|\eta_{t-1,1}|}, & \text{if}\;i=p.
\end{cases}
\]
This design violates Assumption \ref{assu:Innovations} since the $\{\bepsilon_{t}\}_{t\in\Z}$ have conditional variances that can become arbitrarily small/large. Assumption \ref{assu:Innovations} can be enforced by censoring the~$\sigma_{\varepsilon,i}(\bseta_{t-1})$ from above and below by an arbitrarily large (small) number. 
\item\textbf{Design \hypertarget{Design G}{G}:} The $\{\bepsilon_{t}\}_{t\in\Z}$ are as in Design \hyperlink{Design D}{D} but $\bTheta_{01}=(1-5/n)\mathbf{I}_p$. This is a local to unit root design which violates Assumption \ref{assu:Companion} since $\sup_{n\in\N}\max_{i\in[p]}\|(\bTheta_{01}^{h})_{i,1:p}\|_{\ell_{2}}=\sup_{n\in\N}(1-5/n)^{h}=1$ for all $h\in\N$.
\end{itemize}

\subsection{Implementation and Performance Measure}
The Lasso, post-Lasso, and sqrt-Lasso are implemented as follows:
\begin{itemize}
    \item The weighted Lasso  $\widehat\bbeta_i(\lambda_n^\ast,\widehat\bUpsilon_i^{(K)})$ in \eqref{eq:LASSOVector} tuned via Algorithm~\ref{alg:Data-Driven-Penalization}.
    \item The post-Lasso $\check\bbeta_i(\lambda_n^\ast,\check\bUpsilon_i^{(K)})$ in \eqref{eq:postestim} tuned via Algorithm~\ref{alg:Data-Driven-Penalization-with-Refitting}, refitting in every step.\footnote{Theorem~\ref{alg:Data-Driven-Penalization-with-Refitting} provides performance guarantees for this algorithm that are identical to those in Theorem~\ref{thm:postLasso} for the post-Lasso which refits $\widehat\bbeta_i(\lambda_n^\ast,\widehat\bUpsilon_i^{(K)})$ from the final step of Algorithm~\ref{alg:Data-Driven-Penalization} alone.}
    \item The sqrt-Lasso $\dot{\bbeta}_{i}(\lambda_n^\ast/2,\dot{\bUpsilon})$ in \eqref{eq:pensqrtLasso}.
\end{itemize}
We specify $c$ and $\gamma_n$ as in \eqref{eq:tuning_c_and_gamma} and $K=15$.

For all procedures, the intercept is (correctly) enforced to be zero and the number of lags included is the smallest one ensuring that the model is correctly specified (i.e.~four in Design \hyperlink{Design C}{C} and one for all other designs). The number of parameters to be estimated therefore ranges from $256$ in a VAR(1) of size $p=16$ to 65{,}536 in a VAR(4) of size $p=128.$ 

All simulations are carried out in \texttt{R} with user-written functions for each of the three above-mentioned estimators, taking as (basic) inputs the time series data (a matrix of dimensions $(n+q)\times p$) and desired lag length $(q)$.\footnote{Each Lasso problem \eqref{eq:LASSOVector} is solved using \texttt{glmnet}. For the sqrt-Lasso, we downloaded the Matlab\textsuperscript{\textregistered} implementation of the coordinatewise method from Alexandre Belloni's website and translated it to \texttt{R}.}

\subsection{Results}\label{subsec:Results}
We study the maximum row-wise $\ell_{2}$-estimation error $\max_{i\in[p]}\Vert\widehat{\bbeta}_i-\bbeta_{0i}\Vert_{\ell_{2}}$
for which we report the average across the 1,000 Monte Carlo replications relative to that of the weighted Lasso (henceforth: Lasso). Thus, numbers less than one mean a procedure outperforms the Lasso. Figure \ref{fig:SimulationsRelativeErrors} plots the relative errors as a function of the sample size $n$, system size $p$ and design. The raw (non-relative) estimation errors can be found in Figure \ref{fig:SimulationsAbsoluteErrors} of Section \ref{sec:appsim} in the appendix.

\begin{figure}\label{fig:SimulationsRelativeErrors}\caption{Average Estimation Errors relative to Weighted Lasso}
\centering{}\includegraphics[width=.90\textwidth]{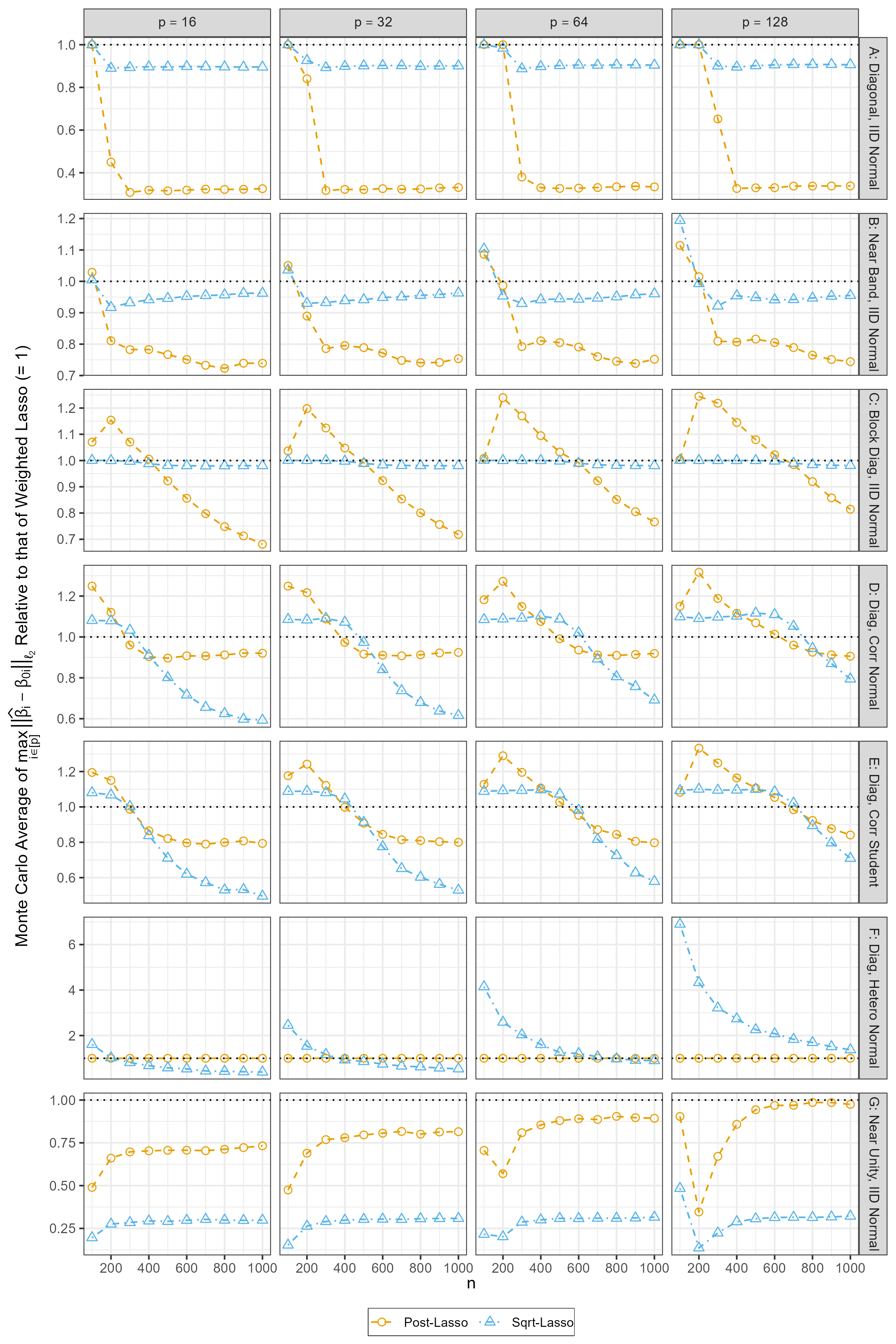}
\end{figure}

In Design \hyperlink{Design A}{A} each equation has only one relevant variable, which is uncorrelated with the irrelevant ones. Thus, the Lasso does well in terms of model selection and, as a result, the post-Lasso works very well here (for all~$p$). Design \hyperlink{Design B}{B} is qualitatively similar, but note that the sqrt- and post-Lasso are less precise than the Lasso for~$p=128$ and~$n$ sufficiently small. In the block-diagonal Design \hyperlink{Design C}{C}, the Lasso and the sqrt-Lasso perform similarly, but for~$p/n$ large the post-Lasso performs worse.

Design \hyperlink{Design D}{D} is challenging in the sense that the single relevant variable in each equation is highly correlated with the irrelevant ones. This makes model selection difficult and explains why the post-Lasso does relatively worse in this setting. The sqrt-Lasso also suffers for~$p/n$ large. Design \hyperlink{Design E}{E} adds heavy tails to Design \hyperlink{Design D}{D}. This leaves the relative estimation errors unaffected.

Design \hyperlink{Design F}{F} confirms that the sqrt-Lasso may suffer substantially under heteroskedasticity (in particular for large~$p$), cf.~the discussion surrounding Theorem \ref{thm:Rates-for-sqrtLASSO-data-driven-loadings}. Finally, Design \hyperlink{Design G}{G} modifies Design \hyperlink{Design A}{A} to a near unit root design. This results in a reversal of the relative performance of the post- and sqrt-Lasso.

\section{Empirical Illustration\label{sec:Empirical-Illustration}}
We apply the methods from the above simulations to forecast a large
set of macroeconomic variables using the Federal Reserve Economic Data (FRED) monthly data (MD)
database. This database is maintained and regularly updated by the Federal
Reserve Bank of St.~Louis and described in detail on Michael W.~McCracken's
website.\footnote{\href{https://research.stlouisfed.org/econ/mccracken/fred-databases/}{https://research.stlouisfed.org/econ/mccracken/fred-databases/}. See also \citet{mccracken_FRED_2016}.} The data is pre-processed in a standard manner using the
Matlab\textsuperscript{\textregistered} code on McCracken's website leaving us with $758$ observations on $p=127$
macroeconomic variables covering March 1959 through April 2022.\footnote{The data pre-processing amounts to carrying out (deterministic) stationarity
inducing transformations (\texttt{prepare\_missing.m}), then removing
outliers (\texttt{remove\_outliers.m}), and finally replacing missing
values with the unconditional average of the corresponding series
(as in the initialization of \texttt{factors\_em.m}).}

\subsection{Forecasting}

In each of the last $120$ months (i.e.~ten years) of the sample we
forecast the $p=127$ variables one month ahead.
Specifically, we estimate VAR$(q)$ models of orders $q\in[12]$
using a rolling estimation window of size $n=758-120-12=626,$ and
create one-month-ahead out-of-sample forecasts as $\widehat{\bY}_{t+1}:=\widehat{\bmu}+\sum_{j=1}^{q}\widehat{\bTheta}_{j}\bY_{t+1-j}$, where $\widehat{\bmu}$ and $\{\widehat{\bTheta}_{j}\}_{j=1}^q$ have not seen $\bY_{t+1}$.
The three methods (weighted Lasso, post-Lasso and sqrt-Lasso) are implemented as in the simulations in Section~\ref{sec:Simulations}, but we now include an unpenalized intercept in each equation to account for non-zero means of the variables.\footnote{See Section \ref{subsec:IncludingIntercepts} for a modification of Algorithm \ref{alg:Data-Driven-Penalization}, which includes unpenalized intercepts.} 
For each method, each $q$,
and each month $t\in\left\{ 639,640,\dotsc,758\right\} $ to be forecast,
we calculate the forecast errors $\widehat{\bY}_{t}-\bY_{t}.$
Due to the different scaling of the variables, we then calculate the inverse-variance-weighted squared 
forecast error (IVWSFE)
\[
\mathrm{IVWSFE}_{t}:=\sum_{i=1}^{p}(\widehat{Y}_{t,i}-Y_{t,i})^{2}/\widehat{\sigma}_{Y_{i}}^{2},
\]
where $\widehat{\sigma}_{Y_{i}}^{2}$ denotes the sample variance
of the $i$\textsuperscript{th} variable (over the entire pre-processed series).

Figure \ref{fig:FRED-MD:-Forecasting-Performance} shows both the average and 95\textsuperscript{th} percentile of the 120 IVWSFEs of the weighted Lasso, post-Lasso and sqrt-Lasso, respectively. To facilitate comparison, these measures are put relative to that of the weighted Lasso with $q=1$. In terms of the \emph{average} IVWSFE error, the sqrt-Lasso does slightly better than both the weighted Lasso and post-Lasso, no matter the choice of lag length $q$, hovering between 96 and 97 pct.~of VAR(1) Lasso. However, when looking at the \emph{95\textsuperscript{th} percentile} IVWSFE, the picture is somewhat reversed, in that the post-Lasso here outperforms the sqrt-Lasso (and weighted Lasso). Thus, these methods cannot be ranked in terms of the quality of their forecasts, in general.\footnote{We also experimented with equation-by-equation ordinary least squares (OLS). With a lag order of one, OLS led to an average IVWSFE 40 percent higher than our weighted Lasso benchmark (for $q=1$). For higher lag orders, the implied design matrix repeatedly fell short of full rank, thus preventing a meaningful comparison.}

\begin{figure}[htb]
\caption{FRED-MD: Forecasting Performance\label{fig:FRED-MD:-Forecasting-Performance}}

\centering{}\includegraphics[width=0.495\textwidth]{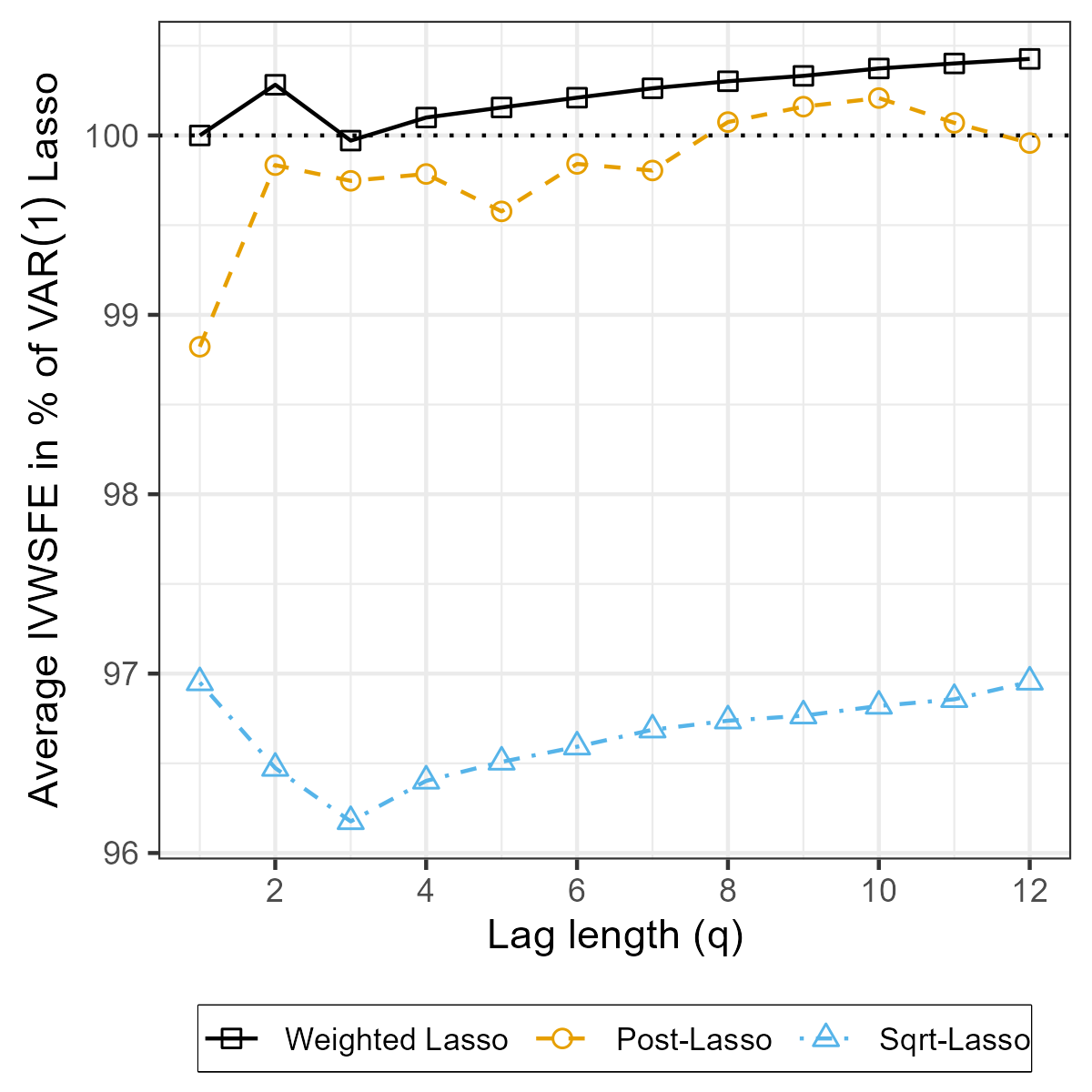}
\includegraphics[width=0.495\textwidth]{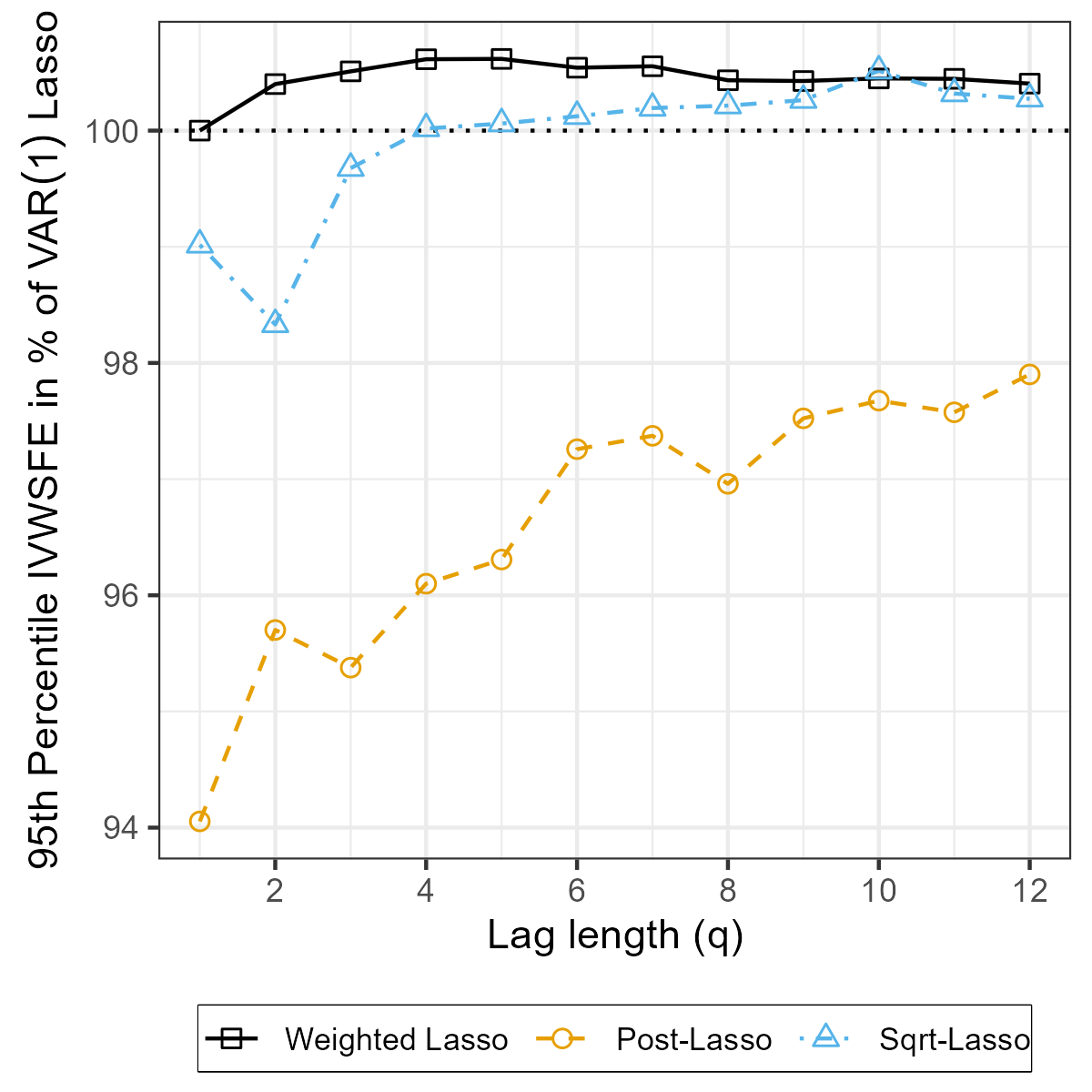}
\end{figure}


\bibliographystyle{apalike}
\bibliography{My_Library} 

\appendix


\newpage{}
\setcounter{page}{1}
\doparttoc 
\faketableofcontents 
\part{Supplementary Appendices} 
\singlespacing 
\parttoc 
\doublespacing 

\newpage

\section{Including Unpenalized Intercepts}\label{subsec:IncludingIntercepts}
Incorporating model intercepts $\bmu_0\in\R^p$, (\ref{eq:VARvector}) becomes
\[
\bY_{t}=\bmu_0+\sum_{j=1}^{q}\bTheta_{0j}\bY_{t-j}+\bepsilon_{t},\quad t\in\Z.
\]
If one is confident that intercepts belong in the model, then they should not be penalized. The resulting Lasso estimator therefore takes the form
\begin{equation}
\big(\widehat{\mu}_i(\lambda,\widehat{\bUpsilon}_{i}),\widehat{\bbeta}_{i}(\lambda,\widehat{\bUpsilon}_{i})\big)\in\argmin_{(\mu,\bbeta)\in\R^{1+pq}}\left\{ \frac{1}{n}\sum_{t=1}^{n}\left(Y_{t,i}-\mu-\bZ_{t-1}^{\top}\bbeta\right)^{2}+\frac{\lambda}{n}\Vert\widehat{\bUpsilon}_{i}\bbeta\Vert_{\ell_{1}}\right\}. \label{eq:LASSOVectorWithIntercept}
\end{equation}
The first-order condition for the intercept reveals the relation
\begin{equation}\label{eq:InterceptSlopeRelation}
\widehat{\mu}_i(\lambda,\widehat{\bUpsilon}_{i})=\overline{Y}_{\! i}-\overline{\bZ}_{\!-1}\widehat{\bbeta}_{i}(\lambda,\widehat{\bUpsilon}_{i}),
\end{equation}
where
\begin{equation}\label{eq:MeanYandZ}
\overline{Y}_{\!i}:=\frac{1}{n}\sum_{t=1}^{n}Y_{t,i}\quad\mathrm{and}\quad\overline{\bZ}_{\!-1}:=\frac{1}{n}\sum_{t=0}^{n-1}{\bZ}_{t}.     
\end{equation}
Concentrating out the intercept in (\ref{eq:LASSOVectorWithIntercept}), we therefore see that
\begin{equation}
\widehat{\bbeta}_{i}(\lambda,\widehat{\bUpsilon}_{i})\in\argmin_{\bbeta\in\R^{pq}}\left\{ \frac{1}{n}\sum_{t=1}^{n}(\ddot{Y}_{t,i}-\ddot{\bZ}_{t-1}^{\top}\bbeta)^{2}+\frac{\lambda}{n}\Vert\widehat{\bUpsilon}_{i}\bbeta\Vert_{\ell_{1}}\right\},
\label{eq:LASSOVectorDemeaned}
\end{equation}
where
\begin{equation}\label{eq:DemeanedYandZ}
\ddot{Y}_{t,i}:=Y_{t,i}-\overline{Y}_{\!i}\quad\mathrm{and}\quad\ddot{\bZ}_{t-1}:=\bZ_{t-1}-\overline{\bZ}_{\!-1}.     
\end{equation}
The problem in (\ref{eq:LASSOVectorDemeaned}) is of the form in (\ref{eq:LASSOVector}) except that we have here replaced the outcome variable and regressors with their demeaned counterparts.
The data-driven penalization algorithm now runs parallel to Algorithm \ref{alg:Data-Driven-Penalization}.
\begin{mdframed}
\begin{lyxalgorithm}
[\textbf{Data-Driven Penalization with Unpenalized Intercepts}]\label{alg:Data-Driven-Penalization-with-Unpenalized-Intercepts}\ \\
\textbf{Initialize}: Specify the penalty level in (\ref{eq:LASSOVectorDemeaned}) as
\begin{equation*}
\lambda^{\ast}_n :=2c\sqrt{n}\Phi^{-1}\big(1-\gamma_n/(2p^{2}q)\big),
\end{equation*}
and specify the initial penalty loadings as 
\begin{equation*}
\widehat{\upsilon}_{i,j}^{\left(0\right)}:=\sqrt{\frac{1}{n}\sum_{t=1}^{n}\ddot{Y}_{t,i}^{2}\ddot Z_{t-1,j}^{2}},\quad i\in[p],\quad j\in[pq].
\end{equation*}
Use $\lambda^{\ast}_n$ and $\widehat{\bUpsilon}_{i}^{\left(0\right)}:=\mathrm{diag}(\widehat{\upsilon}_{i,1}^{\left(0\right)},\dotsc,\widehat{\upsilon}_{i,pq}^{\left(0\right)})$
to compute a Lasso estimate $\widehat{\bbeta}_{i}^{\left(0\right)}:=\widehat{\bbeta}_{i}(\lambda^{\ast}_n,\widehat{\bUpsilon}_{i}^{\left(0\right)})$
via (\ref{eq:LASSOVectorDemeaned}) for each $i\in[p]$. Store the residuals $\widehat{\ddot\varepsilon}_{t,i}^{\left(0\right)}:=\ddot Y_{t,i}-\ddot\bZ_{t-1}^{\top}\widehat{\bbeta}_{i}^{\left(0\right)},t\in[n],i\in[p]$, and set $k=1$.

\noindent\textbf{Update:}
While $k\leqslant K$, specify the penalty loadings as
\begin{equation*}
\widehat{\upsilon}_{i,j}^{\left(k\right)}:=\sqrt{\frac{1}{n}\sum_{t=1}^{n}\widehat{\ddot\varepsilon}_{t,i}^{\left[k-1\right]2}\ddot Z_{t-1,j}^{2}},\quad i\in[p],\quad j\in[pq].
\end{equation*}
Use $\lambda_n^{\ast}$ and $\widehat{\bUpsilon}_{i}^{\left(k\right)}:=\mathrm{diag}(\widehat{\upsilon}_{i,1}^{\left(k\right)},\dotsc,\widehat{\upsilon}_{i,pq}^{\left(k\right)})$
to compute a Lasso estimate as $\widehat{\bbeta}_{i}^{\left(k\right)}:=\widehat{\bbeta}_{i}(\lambda_n^{\ast},\widehat{\bUpsilon}_{i}^{\left(k\right)})$
via (\ref{eq:LASSOVectorDemeaned}) for each $i\in[p]$. Store the residuals $\widehat{\ddot\varepsilon}_{t,i}^{\left(k\right)}:=\ddot Y_{t,i}-\ddot\bZ_{t-1}^{\top}\widehat{\bbeta}_{i}^{\left(k\right)},t\in[n],i\in[p]$,
and increment $k\leftarrow k+1.$
\end{lyxalgorithm}
\end{mdframed}

\noindent Algorithm \ref{alg:Data-Driven-Penalization-with-Unpenalized-Intercepts} yields the (slope) estimates $\widehat\bbeta_i(\lambda_n^\ast,\widehat\bUpsilon_i^{(K)}),i\in[p]$. Using (\ref{eq:InterceptSlopeRelation}) and (\ref{eq:MeanYandZ}), the implied intercept estimates are $\widehat\mu_i(\lambda_n^\ast,\widehat\bUpsilon_i^{(K)})=\overline Y_{\!i}-\overline{\bZ}_{\!-1}^{\top}\widehat\bbeta_i(\lambda_n^\ast,\widehat\bUpsilon_i^{(K)}),i\in[p]$.
Each step in Algorithm \ref{alg:Data-Driven-Penalization-with-Unpenalized-Intercepts} could involve refitting after Lasso selection, thus replacing the Lasso residuals with their post-Lasso counterparts---see Section \ref{sec:OLSforweights} for details.

\section{Auxiliary Results}

\subsection{Functional Dependence Measures for Strictly Stationary Processes\label{subsec:Functional-Dependence-Measures}}

Let $\left\{ \bseta_{t}\right\} _{t\in\Z}$ be independent
and identically distributed $\R^{p_{1}}$-valued random variables $(p_1\in\N)$,
and let $\{\bX_{t}\}_{t\in\Z}$ be $\R^{p_2}$-valued random variables $(p_2\in\N)$ having the
causal representation
\[
\bX_{t}=\bG\del[0]{\dotsc,\bseta_{t-1},\bseta_{t}},\quad t\in\Z,
\]
where $\bG$ is a measurable function of the present and past $\bseta_{t}$.
Strict stationarity of the stochastic process $\{\bX_{t}\}_{t\in\Z}$ follows from
$\bG$ not depending on $t$. Let $\{\bseta_{t}^{\ast}\}_{t\in\Z}$
be an independent copy of $\{\bseta_{t}\}_{t\in\Z}$.
Assuming that the $i$\textsuperscript{th} coordinate $X_{t,i}=G_{i}(\dotsc,\bseta_{t-1},\bseta_{t})$
of $\bX_{t}$ satisfies $\E[|X_{t,i}|^{r}]<\infty$ for some
$r\in[1,\infty)$, we define the functional (or physical) dependence
measure for the coordinate process $\{X_{t,i}\}_{t\in\Z}$ as
\[
\Delta_{r}^{X_{\cdot i}}\left(h\right):=\enVert[1]{X_{t,i}-G_{i}\del[0]{\dotsc,\bseta_{t-h-1}^{\ast},\bseta_{t-h}^{\ast},\bseta_{t-h+1},\dotsc,\bseta_{t}}}_{r}\quad \text{for}\quad h\in\N.
\]
The corresponding expression $\Delta_{r}^{X_{\cdot i}}(0)$ for $h=0$ is at most $2\enVert[0]{X_{0,i}}_r$. It will be convenient to further define $\Delta_{r}^{X_{\cdot i}}(h):=0$ for integers $h<0$. The coordinate process $\{X_{t,i}\}_{t\in\mathbb{Z}}$
is said to be a \emph{geometric moment contraction }[\citet{wu_limit_2004}]
if there are constants $b_{1},b_{2}\in(0,\infty)$ and $\tau\in(0,1]$
such that
\begin{equation}
\Delta_{r}^{X_{\cdot i}}\left(h\right)\leqslant b_{1}\mathrm{e}^{-b_{2}h^{\tau}}\text{ for all }h\in\mathbb{N}.\label{eq:GMCDefn}
\end{equation}

\subsection{Functional Dependence Measure Calculus}
\begin{lem}
\label{lem:Delta_r} Let $\{\left(X_{t,1},X_{t,2}\right)\}_{t\in\Z},$
be a strictly stationary causal process with values in~$\mathbb{R}^2$ and $\mathrm{E}[|X_{0,1}|^{2r}]<\infty$ as well as $\mathrm{E}[|X_{0,2}|^{2r}]<\infty$
for some $r\in[1,\infty)$. Then the following holds:
 \begin{enumerate}
 \item \label{enu:Delta_rFromPairToProduct} The product process $\{X_{t,1}X_{t,2}\}_{t\in\Z}$ is strictly stationary and causal, and its functional dependence measure satisfies
\[
\Delta_{r}^{X_{\cdot1}X_{\cdot2}}\left(h\right)\leqslant\enVert[0] {X_{0,2}}_{2r}\Delta_{2r}^{X_{\cdot1}}\left(h\right)+\enVert[0] {X_{0,1}}_{2r}\Delta_{2r}^{X_{\cdot2}}\left(h\right),\quad h\in\N.
\]
\item \label{enu:Delta_r_sum} With $Y_t := X_{t,1} + \cdots + X_{t-m,1}$ for some $m\in\N$, it holds that $\{Y_t\}_{t\in\Z}$ is strictly stationary and causal, and for $h> m$, we have
\[
\Delta^{Y_{\cdot}}_{r}(h) \leqslant \sum_{i=0}^{m}\Delta^{X_{\cdot 1}}_{r}(h-i).
\]
\item \label{enu:Delta_r_skeleton} With $k\in \N$ let $Z_{t} := X_{tk,1}$. It holds that $\{Z_t\}_{t\in\Z}$ is strictly stationary and causal with
\[
\Delta^{Z_{\cdot}}_{r}(h) = \Delta^{X_{\cdot 1}}_{r}(hk),\quad h\in\N.
\]
\end{enumerate}
\end{lem}
\begin{proof}
Since $\{\left(X_{t,1},X_{t,2}\right)\}_{t\in\Z}$ is a strictly
stationary causal process, it has the representation 
\[
\left(X_{t,1},X_{t,2}\right)=\bG\left(\dotsc,\bseta_{t-1},\bseta_{t}\right),
\]
for some $p_1\in\N$, some i.i.d.~process $\{\bseta_{t}\}_{t\in\Z}$---each $\bseta_t$ being a random element of $\R^{p_1}$---and
some measurable function $\bG$. To show Claim \ref{enu:Delta_rFromPairToProduct}, note that the product process $\{X_{t,1}X_{t,2}\}_{t\in\mathbb{Z}}$
is strictly stationary, causal, and has finite $r$\textsuperscript{th} absolute moment by H{\"o}lder's inequality.
Let $\{\bseta_{t}^{\ast}\}_{t\in\Z}$ be an independent
copy of $\{\bseta_{t}\}_{t\in\Z}$, $h\in\N$
and denote 
\begin{align*}
\del[0]{X_{t,1}',X_{t,2}'} & :=\bG\del[0]{\dotsc,\bseta_{t-h-1}^{\ast},\bseta_{t-h}^{\ast},\bseta_{t-h+1},\dotsc,\bseta_{t}}.
\end{align*}
Then
\begin{align*}
\Delta_{r}^{X_{\cdot1}X_{\cdot2}}\left(h\right) & =\enVert[1]{ X_{t,1}X_{t,2}-X_{t,1}'X_{t,2}+X_{t,1}'X_{t,2}-X_{t,1}'X_{t,2}'}_{r}\\
 & \leqslant\enVert[1]{\del[0]{X_{t,1}-X_{t,1}'}X_{t,2}}_{r}+\enVert[1]{X_{t,1}'\del[0]{X_{t,2}-X_{t,2}'}}_{r}\tag{sub-additivity}\\
 & \leqslant\enVert[0]{X_{t,2}}_{2r}\enVert[0]{X_{t,1}-X_{t,1}'}_{2r}+\enVert[0]{X_{t,1}'}_{2r}\enVert[0]{X_{t,2}-X_{t,2}'}_{2r}\tag{H{\"o}lder}\\
 & =\enVert[0]{X_{0,2}}_{2r}\Delta_{2r}^{X_{\cdot1}}\left(h\right)+\enVert[0]{X_{0,1}}_{2r}\Delta_{2r}^{X_{\cdot2}}\left(h\right).\tag{stationarity}
\end{align*}
To show Claim \ref{enu:Delta_r_sum}, note that 
\[
Y_{t} = G_{1}\left(\dotsc,\bseta_{t-1},\bseta_{t}\right) + \dots + G_{1}\left(\dotsc,\bseta_{t-m-1},\bseta_{t-m}\right), 
\]
which clearly forms a strictly stationary and causal process. Moreover, for $h > m$,
\begin{align*}
    \Delta^{Y_{\cdot}}_{r}(h)  &=  \Big\Vert \Big[ G_{1}\left(\dotsc,\bseta_{t-1},\bseta_{t}\right) + \dots + G_{1}\left(\dotsc,\bseta_{t-m-1},\bseta_{t-m}\right) \Big] \\
    & \qquad -  \Big[ G_1\left(\dotsc,\bseta_{t-h-1}^{\ast},\bseta_{t-h}^{\ast},\bseta_{t-h+1},\dotsc,\bseta_{t}\right) \\
    & \qquad\qquad + \cdots + G_1\left(\dotsc,\bseta_{t-h-1}^{\ast},\bseta_{t-h}^{\ast},\bseta_{t-h+1},\dotsc,\bseta_{t-m}\right) \Big]  \Big\Vert_r \\
    &\leqslant \left\Vert G_{1}\left(\dotsc,\bseta_{t-1},\bseta_{t}\right) - G_1\left(\dotsc,\bseta_{t-h-1}^{\ast},\bseta_{t-h}^{\ast},\bseta_{t-h+1},\dotsc,\bseta_{t}\right)\right\Vert_r \\ & \qquad + \cdots + \left\Vert G_{1}\left(\dotsc,\bseta_{t-m-1},\bseta_{t-m}\right) - G_1\left(\dotsc,\bseta_{t-h-1}^{\ast},\bseta_{t-h}^{\ast},\bseta_{t-h+1},\dotsc,\bseta_{t-m}\right) \right\Vert_r \\
    &= \Delta^{X_{\cdot 1}}_{r}(h) + \Delta^{X_{\cdot 1}}_{r}(h-1) + \dots + \Delta^{X_{\cdot 1}}_{r}(h-m).
\end{align*}
Turning to Claim \ref{enu:Delta_r_skeleton}, denote $\overline{\bseta}_{t} := (\bseta_{(t-1)k+1}^{\top},\dots,\bseta_{tk}^{\top})^{\top}$ (a random element of $\R^{p_1k}$), and define (sub-vector selector) functions $\pi^{(k)}_i : \R^{p_1k} \to \R^{p_1},i\in[k]$, by $\pi^{(k)}_i(\overline{\bseta}_{t}) = \bseta_{(t-1)k+i}$. Then we can represent $Z_t$ as
\begin{align*}
Z_t &=  G_1\left( \dotsc, \bseta_{(t-1)k},\bseta_{(t-1)k+1},\bseta_{(t-1)k+2},\dotsc,\bseta_{tk} \right) \\
    &= G_1\left(\dots, \pi^{(k)}_k(\overline{\bseta}_{t-1}),\pi^{(k)}_1(\overline{\bseta}_{t}),\pi^{(k)}_2(\overline{\bseta}_{t}),\dots,\pi^{(k)}_k(\overline{\bseta}_{t})\right)     =: \overline{G}_1\left( \dotsc, \overline{\bseta}_{t-1}, \overline{\bseta}_{t} \right).
\end{align*}
From the previous display we see that $\{Z_t\}_{t\in\Z}$ is a strictly stationary and causal process given in terms of the i.i.d.~process $\{\overline{\bseta}_{t}\}_{t\in\Z}$---each $\overline{\bseta}_{t}$ being a random element of $\R^{p_1k}$---and the implicitly defined measurable function $\overline{G}_1$. Moreover, denoting $\overline{\bseta}^{\ast}_{t} := ((\bseta^{\ast}_{(t-1)k+1})^{\top},\dots,(\bseta^{\ast}_{tk})^{\top})^{\top}$, for $h \in \N$ we have
\begin{align*}
   \Delta^{Z_{\cdot}}_{r}(h) 
   &= \enVert[2]{\overline{G}_1\left( \dotsc, \overline{\bseta}_{t-1}, \overline{\bseta}_{t} \right)  -\overline{G}_1\left(\dotsc, \overline{\bseta}_{t-h-1}^{\ast},\overline{\bseta}_{t-h}^{\ast},\overline{\bseta}_{t-h+1},\dotsc,\overline{\bseta}_{t}\right)}_r\\
   &=   \Big\Vert G_1\left( \dotsc, \bseta_{(t-1)k-1},\bseta_{(t-1)k},\dotsc,\bseta_{tk-1},\bseta_{tk} \right) \\
   &\qquad - G_1\del[1]{\dotsc, \bseta_{(t-h)k-1}^{\ast} ,\bseta_{(t-h)k}^{\ast} ,\bseta_{(t-h)k+1} ,\dotsc, \bseta_{(t-1)k-1},\bseta_{(t-1)k},\dotsc,\bseta_{tk-1},\bseta_{tk}}   \Big\Vert_r  \\
   &=  \Delta^{X_{\cdot 1}}_{r}(hk).
\end{align*}

\end{proof}

\subsection{Maximal Inequality with Functional Dependence}
We here establish a maximal inequality for an array of $T\in\N$ possibly dependent random vectors of dimension $d_T$ possibly growing with $T$.
\begin{lem}
\label{lem:MaximalInequalityFunctionalDependence} Let $\{\bX_{t}\}_{t\in\mathbb{Z}}$
be a strictly stationary causal process with each $\bX_t$ a random element of $\R^{d_T}$ for some $d_T\in\{2,3,\dotsc\}$,
and $\max_{i\in[d_T]}\Vert X_{0,i}\Vert_{\psi_{\xi}}\leqslant K$
for some constants $\xi,K\in(0,\infty)$. Then for any $r\in[1,\infty)$ and
any $m_T\in[T]$,
\[
\max_{i\in[d_T]}\envert[4]{\sum_{t=1}^{T}\del[1]{X_{t,i}-\E\left[X_{0,i}\right]}} \lesssim_{\P}T d_T^{1/r}\Delta_{r,T}+m_T\sqrt{l_{T}\ln\left(d_T T\right)}+m_T\del[1]{\ln(d_T T)}^{\overline{\xi}}
\]
as $T\to\infty$, where
\[
l_T:=\lceil T/m_T\rceil,\quad\Delta_{r,T}:=\max_{i\in[d_T]}\max_{u\in\{m_T + 1,\dotsc,2m_T\}}\Delta_{r}^{X_{\cdot,i}}(u)
\quad\text{and}\quad \overline{\xi}:=\frac{1}{\xi}+\frac{1}{\min\{\xi,1\}}.
\]

\end{lem}
\begin{rem}\label{rem:DecreasingDependence}
If $u\mapsto \Delta_{r}^{X_{\cdot,i}}(u)$ is decreasing for all~$i\in[d_T]$, one
can upper bound $\Delta_{r,T}$ by $\max_{i\in[d_T]}\Delta_{r}^{X_{\cdot,i}}(m_T)$.\hfill$\diamondsuit$
\end{rem}
\begin{proof}
[\sc{Proof of Lemma \ref{lem:MaximalInequalityFunctionalDependence}}]
The strategy for proving Lemma \ref{lem:MaximalInequalityFunctionalDependence} is two-fold. First, we approximate the original (functionally) dependent process~$\{\bX_{t}\}_{t\in\mathbb{Z}}$ with an $m$-dependent process. This idea is classic [see, e.g., \cite{zhang2018gaussian}]. The $m$-dependent process can then be handled using tools designed for independent observations.

First, write~$\bX_{t}=\bG(\dotsc,\bseta_{t-1},\bseta_{t})$, which is possible by strict stationarity and causality.
Fix $m_T\in[T]$ and, thus, $l_T=\lceil T/m_T\rceil$ and abbreviate $m:=m_T$, $l:=l_T$, and~$d:=d_T$. Define
\[
\bY_{k}:=\sum_{t=(k-1)m+1}^{\min\{km,T\}}\bX_{t}\quad\text{and}\quad\bY_{k}^{(m)}:=\E\sbr[1]{\bY_{k}\mid\left\{\bseta_{t}\right\}_{t=(k-2)m+1}^{km}},\quad k\in[l].
\]
Then, per iterated expectations and the triangle inequality,
\begin{align*}
 & \max_{i\in[d]}\envert[4]{\sum_{t=1}^{T}\del[1]{X_{t,i}-\E[X_{0,i}]}}=\max_{i\in[d]}\envert[4]{\sum_{k=1}^{l}\del[1]{Y_{k,i}-\E[Y_{k,i}^{(m)}]}}\\
 & \leqslant\max_{i\in[d]}\envert[4]{\sum_{k=1}^{l}\del[1]{Y_{k,i}-Y_{k,i}^{\left(m\right)}}}+\max_{i\in[d]}\envert[4]{\sum_{k=1}^{l}\del[1]{Y_{k,i}^{\left(m\right)}-\E[Y_{k,i}^{\left(m\right)}]}}=:\mathrm{I}_T+\mathrm{II}_T.
\end{align*}
We bound each of the right-hand side terms, starting with $\mathrm{I}_T.$
Since $\max_{i\in[d]}\Vert X_{0,i}\Vert_{\psi_{\xi}}\leqslant K$, all absolute moments of the $X_{0,i}$ exist. Hence, for any~$r\in[1,\infty)$,
\begin{align*}
\E\left[\mathrm{I}_T\right]
 &\leqslant \del[4]{\E\sbr[4]{\max_{i\in[d]}\envert[4]{\sum_{k=1}^{l}\del[1]{Y_{k,i}-Y_{k,i}^{\left(m\right)}}}^r}}^{1/r}\tag{Jensen}\\
 & \leqslant d^{1/r}\max_{i\in[d]}\enVert[4]{\sum_{k=1}^{l}\del[1]{Y_{k,i}-Y_{k,i}^{\left(m\right)}}}_{r}\tag{max $\leqslant$ sum $\leqslant$ repeated max}\\
 & \leqslant d^{1/r}\max_{i\in[d]}\sum_{k=1}^{l}\enVert[2]{Y_{k,i}-Y_{k,i}^{\left(m\right)}}_{r}\tag{sub-additivity}\\
 & \leqslant d^{1/r}\max_{i\in[d]}\sum_{k=1}^{l}\sum_{t=\left(k-1\right)m+1}^{\min\{km,T\}}\enVert[2]{X_{t,i}-\E\left[X_{t,i}\middle|\left\{ \bseta_{u}\right\} _{u=(k-2)m+1}^{km}\right]}_{r}\tag{sub-additivity}\\
 & \leqslant d^{1/r}\max_{i\in[d]}\sum_{k=1}^{l}\sum_{t=\left(k-1\right)m+1}^{km}\enVert[2]{X_{t,i}-\E\left[X_{t,i}\middle|\left\{ \bseta_{u}\right\} _{u=(k-2)m+1}^{km}\right]}_{r}.\tag{non-negativity}
\end{align*}
Fix $i\in[d],k\in[l]$ and $t\in\{(k-1)m+1,\dotsc,km\}$. Letting $\{\bseta_t^\ast\}_{t\in\Z}$ be an independent copy of $\{\bseta_t\}_{t\in\Z}$, we have
\begin{align*}
&\enVert[2]{X_{t,i}-\mathrm{E}\left[X_{t,i}\middle|\left\{ \bseta_{u}\right\} _{u=(k-2)m+1}^{km}\right]}_{r}^r\\
&= \E\sbr[3]{\envert[2]{X_{t,i}-\mathrm{E}\left[X_{t,i}\middle|\left\{ \bseta_{u}\right\} _{u=(k-2)m+1}^{t}\right]}^r}\tag{causality}\\
&= \E\sbr[4]{\envert[3]{\E\left[X_{t,i}-G_i(\dotsc,\bseta_{(k-2)m-1}^\ast,\bseta_{(k-2)m}^\ast,\bseta_{(k-2)m+1},\dotsc,\bseta_t)\middle|\{\bseta_u\}_{u=-\infty}^{t}\right]}^r}.\tag{causality}
\end{align*}
Applying Jensen's inequality (conditionally) to the right-hand side inner expectation, we get
\begin{align*}
&\E\sbr[4]{\envert[3]{\E\left[X_{t,i}-G_i(\dotsc,\bseta_{(k-2)m-1}^\ast,\bseta_{(k-2)m}^\ast,\bseta_{(k-2)m+1},\dotsc,\bseta_t)\middle|\{\bseta_u\}_{u=-\infty}^{t}\right]}^r}\\
&\leqslant \E\sbr[3]{\envert[2]{X_{t,i}-G_i(\dotsc,\bseta_{(k-2)m-1}^\ast,\bseta_{(k-2)m}^\ast,\bseta_{(k-2)m+1},\dotsc,\bseta_t)}^r}
=\sbr[1]{\Delta_{r}^{X_{\cdot,i}}(t-(k-2)m)}^r.
\end{align*}
Combining the previous two displays and adding up, by a change of variables, we see that
\begin{align*}
\sum_{t=(k-1)m+1}^{km} \enVert[2]{X_{t,i}-\E\left[X_{t,i}\middle|\left\{ \bseta_{u}\right\} _{u=(k-2)m+1}^{km}\right]}_{r}
&\leqslant \sum_{t=(k-1)m+1}^{km}\Delta_{r}^{X_{\cdot,i}}\left(t-(k-2)m\right)\\
&= \sum_{u=m+1}^{2m} \Delta_{r}^{X_{\cdot,i}}\left(u)\right)\\
&\leqslant m \max_{u\in\{m+1,\dotsc,2m\}} \Delta_{r}^{X_{\cdot,i}}\left(u)\right).
\end{align*}
Since the right-hand bound does not depend on $k$, it follows from our calculations that
\begin{align*}
\E[\mathrm{I}_T]
&\leqslant mld^{1/r}\max_{i\in[d]}\max_{u\in\{m+1,\dotsc,2m\}} \Delta_{r}^{X_{\cdot,i}}\left(u\right)=mld^{1/r}\Delta_{r,T}.
\end{align*}
By Markov's inequality and $m l \leqslant 2T$, we get $\mathrm{I}_T\lesssim_{\P}T d^{1/r}\Delta_{r,T}$ as $T\to\infty$.

To bound $\mathrm{II}_T$, observe that
\[
\mathrm{II}_T\leqslant\underbrace{\max_{i\in[d]}\envert[3]{\sum_{k\leqslant l\,\mathrm{odd}}\del[1]{Y_{k,i}^{\left(m\right)}-\E[Y_{k,i}^{(m)}]}}}_{=:\mathrm{II}^{\text{odd}}_T}+\underbrace{\max_{i\in[d]}\envert[3]{\sum_{k\leqslant l\,\mathrm{even}}\del[1]{Y_{k,i}^{\left(m\right)}-\E[Y_{k,i}^{(m)}]}}}_{=:\mathrm{II}^{\text{even}}_T}.
\]
We set up for two applications of \citet[Theorem 3.4]{kuchibhotla_moving_2022}
for the odd and even parts, respectively. We only bound $\mathrm{II}^{\text{odd}}_T$;
the argument for $\mathrm{II}^{\text{even}}_T$ is analogous. 

First, for any $r\in[1,\infty)$, $i\in[d]$ and  $k\in[l]$,
\begin{align*}
\enVert[1]{Y_{k,i}^{\left(m\right)}-\mathrm{E}[Y_{k,i}^{(m)}]}_{r}
\leqslant
2\Vert Y_{k,i}^{\left(m\right)}\Vert_{r}
\leqslant
2\Vert Y_{k,i}\Vert_{r}
\leqslant
2m\Vert X_{0,i}\Vert_{r}
\leqslant 2Kmr^{1/\xi},
\end{align*}
which means that
\[
\enVert[2]{m^{-1}\del[1]{Y_{k,i}^{(m)}-\E[Y_{k,i}^{(m)}]}}_{r}\leqslant 2Kr^{1/\xi}.
\]
From the previous display, an application of \citet[Theorem 1]{vladimirova_sub-weibull_2020} implies the existence of a constant $C_{\xi,K}\in(0,\infty)$, depending only on $\xi$ and $K$, such that
\[
\max_{i\in[d],k\in[l]} \E\left[\exp\left\{\envert[2]{m^{-1}\del[1]{Y_{k,i}^{(m)}-\E[Y_{k,i}^{(m)}]}}^{\xi}\middle/C_{K,\xi}^{\xi}\right\}\right]\leqslant2.
\]
Let $l^{\text{odd}}:=l^{\text{odd}}_T:=|\{k\in[l]:k\,\text{odd}\}|$. The $l^{\text{odd}}$ summands $m^{-1}(\bY_{k}^{(m)}-\E[\bY_{k}^{(m)}])$, are independent random elements of $\R^{d}$. It therefore follows from 
\citet[Theorem 3.4]{kuchibhotla_moving_2022}, with $t$ there set to $\ln(dT)$, that there is a constant $C_{K,\xi}'\in(0,\infty)$, depending only on $\xi$ and $K$, such that with probability
at least $1-3/(dT)$,
\begin{align}
&\max_{i\in[d]}\envert[4]{\frac{1}{l^{\text{odd}}}\sum_{k\leqslant l\,\mathrm{odd}}\frac{Y_{k,i}^{(m)}-\E[Y_{k,i}^{(m)}]}{m}}\notag\\
&\leqslant 7 \sqrt{\frac{\Gamma_{T}\left(\ln(dT)+\ln d\right)}{l^{\text{odd}}}}
+ \frac{C_{K,\xi}'\ln\left(2l^{\text{odd}}\right)^{1/\xi}\left(\ln\left(dT\right)+\ln d\right)^{1/\min\{1,\xi\}}}{l^{\text{odd}}}\label{eq:oddterms},
\end{align}
where we define
\[
\Gamma_{T}:=\max_{i\in[d]}\frac{1}{l^{\text{odd}}}\sum_{k=1}^{l^{\text{odd}}}\E\sbr{\envert[4]{\frac{Y_{k,i}^{(m)}-\E[Y_{k,i}^{(m)}]}{m}}^{2}}.
\]
From previous calculations, setting $r=2$, we know that for all $i\in[d]$ and $k\in[l]$,
\[
\E\sbr{\envert[4]{\frac{Y_{k,i}^{(m)}-\E[Y_{k,i}^{(m)}]}{m}}^{2}}
=\enVert{\frac{Y_{k,i}^{(m)}-\E[Y_{k,i}^{(m)}]}{m}}_{2}^2\leqslant 2^2K^2\cdot 2^{2/\xi},
\]
so $\Gamma_T\leqslant 2^{2+2/\xi}K^2$.
Multiplying both sides of~\eqref{eq:oddterms} by $l^{\text{odd}}$ and $m$, we deduce that for some constant $C_{K,\xi}''\in(0,\infty)$, depending only on $\xi$ and $K$, with probability approaching one as $T\to\infty$,
\[
\max_{i\in[d]}\envert[4]{\sum_{k\leqslant l\,\mathrm{odd}}\del[1]{Y_{k,i}^{(m)}-\E[Y_{k,i}^{(m)}]}}
\leqslant C_{K,\xi}''
\left(
m \sqrt{l\ln(dT)}+ m\left(\ln(dT)\right)^{\overline\xi}
\right).
\]
It follows that $\mathrm{II}_T^{\text{odd}}\lesssim_{\P} m \sqrt{l\ln(d T)}+ m(\ln(d T))^{\overline\xi}$ as $T\to\infty$. Hence, $\mathrm{II}_T\lesssim_{\P} m \sqrt{l\ln(d T)}+ m(\ln(d T))^{\overline\xi}$ as $T\to\infty$.
\end{proof}

\section{Moderate Deviation with Geometric Moment Contraction}
Let $\{X_{t}\}_{t\in\Z}$ be a strictly stationary causal
process with elements $X_t$ taking values in $\R$. Define the sums
\[
S_{k,h}:=\sum_{t=k+1}^{k+h}X_{t},\quad k\in\N_{0},\quad h\in\N,
\]
and abbreviate $S_{n}:=S_{0,n}$. Consider the following assumptions.
\begin{assumption}
\label{assu:MomentConditions} $\E\left[X_{0}\right]=0$ and
$\E\left[X_{0}^{4}\right]<\infty$.
\end{assumption}
\begin{assumption}
\label{assu:UniformNondegeneracy} There is a constant $\omega\in(0,\infty)$, such that $\E[S_{k,h}^{2}]\geqslant\omega^{2}h$
for all $k\in\N_{0}$ and all $h\in\N$.
\end{assumption}
\begin{assumption}
\label{assu:GMC} There are constants $b_{1},b_{2}\in(0,\infty)$ and $\tau\in (0,1]$, such that $\Delta_{4}^{X_{\cdot}}(h)\leqslant b_{1}\mathrm{e}^{-b_{2}h^\tau}$
for all $h\in\N.$
\end{assumption}
Let $n\in\N$ and define the block size $m_n:=\lfloor n^{1/(1+4\tau))}\rfloor$ and number of blocks $l_n:=\lfloor n/m_n\rfloor$, and let $H_{n,k}$ be the $k$\textsuperscript{th} block of (time) indices, i.e.
\begin{equation}\label{eq:HnkDef}
    H_{n,k}:=\left\{(k-1)m_n  + 1, (k-1)m_n + 2,\dotsc, k m_n\right\},\quad k\in[l_n].
\end{equation}
Define the block-normalized sum
\[
T_{n}:=\frac{\sum_{k=1}^{l_n}\sum_{t\in H_{n,k}}X_{t}}{\sqrt{\sum_{k=1}^{l_n}\left(\sum_{t\in H_{n,k}}X_{t}\right)^{2}}}.
\]
The following lemma is a consequence of \citet[Theorem 3.3]{gao_refined_2022}.
\begin{lem}
\label{lem:ModDevBlockNormalized} Let Assumptions \ref{assu:MomentConditions},
\ref{assu:UniformNondegeneracy} and \ref{assu:GMC} hold. Then there are constants $A$ and $g_0$ both in $(0,\infty)$, 
depending only on $b_{1},b_{2},\omega$ and $\tau$, such
that
\[
x\in\left[0,g_{0}n^{\tau/(2(1+4\tau))}\right]\implies\frac{\P\left(\left|T_{n}\right|\geqslant x\right)}{2\left[1-\Phi\left(x\right)\right]}\leqslant1+A\left(\frac{1+x^{2}}{n^{1/(1+4\tau)}}+\frac{1+x^{2}}{n^{\tau/(1+4\tau)}}\right).
\]
\end{lem}
\begin{proof}
If Assumptions \ref{assu:MomentConditions}, \ref{assu:UniformNondegeneracy}
and \ref{assu:GMC} hold for the random process $\{X_{t}\}$, then
they hold for $\{-X_{t}\}$ as well and with the same constants. Applying
\citet[Theorem 3.3]{gao_refined_2022} twice [with their $\alpha=1/(1+4\tau)$], we therefore get that $x\in[0,g_{0}n^{\tau/\sbr[0]{2(1+4\tau)}}]$ implies
\[
\max\cbr[1]{\P\left(T_{n}\geqslant x\right),\P\left(-T_{n}\geqslant x\right)} \leqslant\sbr[1]{1-\Phi\left(x\right)}\sbr[3]{1+A\del[3]{\frac{1+x^{2}}{n^{1/(1+4\tau)}}+\frac{1+x^{2}}{n^{\tau/(1+4\tau)}}}},
\]
for constants $A$ and $g_{0}$ with the claimed dependencies. The
claim then follows from $\P(|T_{n}|\geqslant x)=\P(T_{n}\geqslant x)+\P(-T_{n}\geqslant x)$.
\end{proof}
Next, let $\{\bX_t\}_{t\in\Z}$ be a strictly stationary causal process with elements $\bX_t$ taking values in $\R^{d_n}$, where we take $d_n$ to be at least two. Let $T_{n,j}$
denote the block-normalized sum using the $j$\textsuperscript{th} coordinate process
$\{X_{t,j}\}_{t\in\N}$ of~$\{\bX_t\}_{t\in\Z}$, such that
\begin{align*}
T_{n,j} & =\frac{\sum_{k=1}^{l_n}\sum_{t\in H_{n,k}}X_{t,j}}{\sqrt{\sum_{k=1}^{l_n}\left(\sum_{t\in H_{n,k}}X_{t,j}\right)^{2}}},\quad j\in[d_n].
\end{align*}
We invoke the growth condition:
\begin{assumption}
\label{assu:Growth} $\{\gamma_{n}\}_{n=1}^{\infty}$ is a non-random sequence in $(0,1)$
satisfying $\ln(1/\gamma_n)\lesssim\ln(nd_n)$.
\end{assumption}
\begin{lem}
\label{lem:TailBoundGMCVector} Let all coordinate processes $\{\{ X_{t,j}\}_{t\in\Z}\}_{j=1}^{d_n}$
satisfy Assumptions \ref{assu:MomentConditions}, \ref{assu:UniformNondegeneracy}
and \ref{assu:GMC} for the same constants $b_{1},b_{2},\omega$ and $\tau$, let Assumption \ref{assu:Growth} hold, and suppose that
$\ln(d_n)=o(n^{\tau/(1+4\tau)}).$
Then there is a constant $N\in\N$ such
that
\[
n\geqslant N\implies \P\left(\max_{j\in[d_n]}\left|T_{n,j}\right|\geqslant\Phi^{-1}\del[2]{1-\frac{\gamma_n}{2d_n}}\right)\leqslant 2\gamma_n.
\]
\end{lem}
\begin{proof}
Under Assumption \ref{assu:Growth}, using $d_n\geqslant2$ and the
Gaussian quantile bound $\Phi^{-1}(1-z)\leqslant\sqrt{2\ln(1/z)},z\in(0,1),$
we get
\begin{equation}
0<\Phi^{-1}\left(1-\gamma_n/\left(2d_n\right)\right)\leqslant\sqrt{2\ln\left(2d_n/\gamma_n\right)}\leqslant C_{1}\sqrt{\ln(nd_n)}\label{eq:GrowthGammaImplication}
\end{equation}
for some constant $C_{1}\in(0,\infty)$. The growth condition $\ln(d_n)=o(n^{\tau/(1+4\tau)})$
equips us with an $N_1\in\N$ such that for the constant $g_{0}\in(0,\infty)$ provided by Lemma \ref{lem:ModDevBlockNormalized},
\[
n\geqslant N_1\implies \Phi^{-1}\left(1-\gamma_n/\left(2d_n\right)\right)\in\sbr{0, g_{0}n^{\tau/\sbr[0]{2\del[0]{1+4\tau}}}}.
\]
For such $n\geqslant N_1$,
\begin{align*}
 & \P\left(\max_{j\in[d_n]}\left|T_{n,j}\right|\geqslant\Phi^{-1}\del[2]{1-\frac{\gamma_n}{2d_n}}\right)\\
 & \leqslant\sum_{j=1}^{d_n}\P\left(\left|T_{n,j}\right|\geqslant\Phi^{-1}\del[2]{1-\frac{\gamma_n}{2d_n}}\right)\tag{union bound}\\
 & \leqslant\sum_{j=1}^{d_n}2\left[1-\Phi\left(\Phi^{-1}\del[2]{1-\frac{\gamma_n}{2d_n}}\right)\right] \\
 & \qquad \qquad \times \left[1+A \left( 1+\left[\Phi^{-1}\del[2]{1-\frac{\gamma_n}{2d_n}}\right]^{2} \right) \left(\frac{1}{n^{1/(1+4\tau)}}+ \frac{1}{n^{\tau/(1+4\tau)}}\right)\right]\tag{Lemma \ref{lem:ModDevBlockNormalized}}\\
 & \leqslant\gamma_n\left[1+A\left(1+C_{1}^{2}\ln(nd_n)\right) \left(\frac{1}{n^{1/(1+4\tau)}}+ \frac{1}{n^{\tau/(1+4\tau)}}\right)\right].
\end{align*}
Again using $\ln(d_n)=o(n^{\tau/(1+4\tau)})$, we see that there is a constant $N_2\in\N$ for which
\[
n\geqslant N_2\implies A\left(1+C_{1}^{2}\ln(nd_n)\right) \left(\frac{1}{n^{1/(1+4\tau)}}+ \frac{1}{n^{\tau/(1+4\tau)}}\right)\leqslant 1.
\]
The claim then follows from setting $N:=\max\{N_1,N_2\}$.
\end{proof}

\section{Main Steps in Proving Theorem~\ref{thm:Rates-for-LASSO-data-driven-loadings}}\label{sec:proofsteps}
To sketch the proof of Theorem~\ref{thm:Rates-for-LASSO-data-driven-loadings}, we start with a ``master
lemma'' that provides non-asymptotic error guarantees for the Lasso
under standard high-level conditions, well-known in the literature.  To state this lemma, let $C\in(0,\infty)$ and $\emptyset\neq T\subseteq[pq]$ and define the \emph{restricted set}
\begin{equation}
\mathcal{R}_{C,T}:=\left\{ \bdelta\in\mathbb{R}^{pq}:\left\Vert \bdelta_{T^{c}}\right\Vert _{\ell_{1}}\leqslant C\left\Vert \bdelta_{T}\right\Vert _{\ell_{1}}\right\},\label{eq:RestrictedSet}
\end{equation}
as well as the \emph{restricted eigenvalue} of the (scaled) Gram matrix $\bfX^\top\bfX/n$
\begin{equation}
\widehat\kappa_{C}^{2}:=\widehat\kappa_{C}^{2}\big(s,\bfX^\top\bfX/n\big):=\min_{\substack{T\subseteq[pq]:\\
|T|\in[s]}
}\inf_{\substack{\bdelta\in\mathcal{R}_{C,T}:\\ \bdelta\neq \mathbf{0}_{pq}}}\frac{\bdelta^\top(\bfX^\top\bfX/n)\bdelta}{\left\Vert \bdelta\right\Vert _{\ell_{2}}^{2}}.\label{eq:RestrictedEigenvalue}
\end{equation}

Similar to \citet{belloni_sparse_2012}, we introduce certain infeasible ideal penalty loadings. 
However, to account for the dependence in the data, the infeasible loadings are based on dividing the data into blocks of carefully chosen size $m_{n}:=n^{1/(1+4\tau)}$, the number of which is
$l_{n}:=n/m_n=n^{4\tau/(1+4\tau)}$.\footnote{For simplicity of notation, we treat $m_n$ as a divisor of $n$, such that both $m_n\in\N$ and $l_n\in\N$.} Such blocking is not necessary for independent data. We stress that the introduction of a block size $m_n$ is solely for the
purpose of the theoretical analysis, as our Algorithm \ref{alg:Data-Driven-Penalization}
does not require any user-chosen block size. Gathering the indices in the $k$\textsuperscript{th} block into $H_{n,k}$ [see \eqref{eq:HnkDef}], we define our blocking-based \emph{ideal penalty loadings}
\begin{equation}
\widehat{\bUpsilon}_{i}^{0}:=\text{diag}\big(\widehat{\upsilon}_{i,1}^{0},\dotsc,\widehat{\upsilon}_{i,pq}^{0}\big),\quad i\in[p],\label{eq:PenaltyLoadingsIdeal}
\end{equation}
with ideal loading $(i,j)$ given by
\begin{equation}
\widehat{\upsilon}_{i,j}^{0}:=\sqrt{\frac{1}{n}\sum_{k=1}^{l_n}\del[4]{\sum_{t\in H_{n,k}}\varepsilon_{t,i}Z_{t-1,j}}^{2}},\quad i\in[p],\quad j\in[pq].\label{eq:PenaltyLoadingsIdeal_p2}
\end{equation}
We define the \textit{score vector} $\bS_{n,i}$ as the (negative) gradient of the loss $\widehat Q_i$ in (\ref{eq:Qihat}) at $\bbeta_{0i}$ upon normalizing by the ideal penalty loadings $\widehat{\bUpsilon}_{i}^{0}$,
\begin{equation}
\bS_{n,i}:=(\widehat{\bUpsilon}_{i}^{0})^{-1}\del[1]{-\nabla \widehat Q_i(\bbeta_{0i})}=\frac{2}{n}\sum_{t=1}^{n}(\widehat{\bUpsilon}_{i}^{0})^{-1}\bZ_{t-1}\varepsilon_{t,i},\quad i\in[p].\label{eq:def_score_ideal}
\end{equation}
The score $\bS_{n,i}$ represents the estimation noise in problem (\ref{eq:LASSOVector}), with the ideal penalty loadings yielding a form of self-normalization of the score. The self-normalization of the $\bS_{n,i}$ allows us to invoke recent moderate deviation theory for sums of dependent random variables developed in \citet{gao_refined_2022}.
Specifically, with
$c$ and $\gamma_n$ given in (\ref{eq:tuning_c_and_gamma}), and the penalty level $\lambda_n^\ast$ given in (\ref{eq:PenaltyLevelPractice}),
we show in Lemma \ref{lem:Deviation-Bound} that $\lambda^\ast_n/n\geqslant c\max_{i\in[p]}\Vert \bS_{n,i}\Vert_{\ell_{\infty}}$
with probability approaching one.

The ideal penalty loadings in (\ref{eq:PenaltyLoadingsIdeal})--(\ref{eq:PenaltyLoadingsIdeal_p2})
are infeasible, since the innovations $\{\bepsilon_{t}\}_{t\in\Z}$ are unobservable.
To overcome this problem, we prove that the loadings arising from Algorithm
\ref{alg:Data-Driven-Penalization} are ``close'' to the ideal loadings.
More precisely, akin to \citet[][p.~2387]{belloni_sparse_2012}, we refer to an
arbitrary collection of penalty loadings $\widehat{\bUpsilon}_{i}=\mathrm{diag}(\widehat{\upsilon}_{i,1},\dotsc,\widehat{\upsilon}_{i,pq}),i\in[p]$,
as \textit{asymptotically valid}, if there are non-negative random scalars $\widehat\ell:=\widehat\ell_n$ and $\widehat{u}:=\widehat u_n,n\in\N$, and a constant $u\in[1,\infty)$, such that
\begin{equation}\label{eq:AsymptoticallyValidPenaltyLoadings}
\widehat{\ell}\widehat{\upsilon}_{i,j}^{0}\leqslant \widehat{\upsilon}_{i,j} \leqslant\widehat{u}\widehat{\upsilon}_{i,j}^{0}\qquad\text{for all}\; (i,j)\in[p]\times[pq],
\end{equation}
with probability approaching one as well as
\begin{equation}\label{eq:AsymptoticallyValidPenaltyLoadings2}
\widehat{\ell} \overset{\mathrm{P}}{\to}1\quad\text{and}\quad\widehat{u}\overset{\mathrm{P}}{\to}u.
\end{equation}

We now present a master lemma that provides prediction and estimation
error guarantees. The proof is adapted from 
\citet[Lemma 6]{belloni_sparse_2012} to a form suitable for our purposes.
Similar statements can be found in \citet{bickel_simultaneous_2009}, \citet{loh_high-dimensional_2012} and \citet{negahban_unified_2012}
among others. 
To state the lemma, denote the smallest and largest ideal penalty loading $\widehat{\upsilon}_{\min}^{0}:=\min_{(i,j)\in[p]\times[pq]}\widehat{\upsilon}_{i,j}^{0}$ and $\widehat{\upsilon}_{\max}^{0}:=\max_{(i,j)\in[p]\times[pq]}\widehat{\upsilon}_{i,j}^{0}$, respectively, and their ratio ~$\widehat{\mu}_{0}:=\widehat{\upsilon}_{\max}^{0}/\widehat{\upsilon}_{\min}^{0}$.
\begin{lem}
[\textbf{Master Lemma}] \label{lem:MasterLemma} With the constant
$c$ given in (\ref{eq:tuning_c_and_gamma}) and the score vectors
$\{\bS_{n,i}\}_{i\in[p]}$ given by (\ref{eq:def_score_ideal}), suppose that the penalty level satisfies $\lambda/n\geqslant c\max_{i\in[p]}\Vert \bS_{n,i}\Vert_{\ell_{\infty}}$,
and that for some random scalars $\widehat\ell$ and $\widehat u$, the penalty loadings $\{\widehat{\bUpsilon}_{i}\}_{i=1}^{p}$ satisfy (\ref{eq:AsymptoticallyValidPenaltyLoadings})
with $1/c<\widehat\ell\leqslant\widehat u$. Then the Lasso estimates $\widehat{\bbeta}_i:=\widehat{\bbeta}_i(\lambda,\widehat\bUpsilon_i),i\in[p]$,
arising from $\lambda$ and $\{\widehat{\bUpsilon}_{i}\}_{i=1}^p$ via (\ref{eq:LASSOVector}) satisfy the error bounds
\begin{align*}
\max_{i\in[p]}\Vert\widehat{\bbeta}_{i}-\bbeta_{0i}\Vert_{2,n} & \leqslant\left(\widehat u+\frac{1}{c}\right)\frac{\lambda\widehat{\upsilon}_{\max}^{0}\sqrt{s}}{n\widehat\kappa_{\widehat c_{0}\widehat{\mu}_{0}}},\\
\max_{i\in[p]}\Vert\widehat{\bbeta}_{i}-\bbeta_{0i}\Vert_{\ell_{1}} & \leqslant\left(\widehat u+\frac{1}{c}\right)\left(1+\widehat c_{0}\widehat{\mu}_{0}\right)\frac{\lambda s\widehat{\upsilon}_{\max}^{0}}{n\widehat \kappa_{\widehat c_{0}\widehat{\mu}_{0}}^{2}}\quad\text{and}\\
\max_{i\in[p]}\Vert\widehat{\bbeta}_{i}-\bbeta_{0i}\Vert_{\ell_{2}} & \leqslant\left(\widehat u+\frac{1}{c}\right)\frac{\lambda\widehat{\upsilon}_{\max}^{0}\sqrt{s}}{n\widehat\kappa_{\widehat c_{0}\widehat{\mu}_{0}}^{2}},
\end{align*}
where $\widehat c_{0}:=(\widehat uc+1)/(\widehat\ell c-1)$, and $\widehat \kappa_{\widehat c_{0}\widehat{\mu}_{0}}$ is defined via (\ref{eq:RestrictedEigenvalue}) with $C=\widehat c_{0}\widehat{\mu}_{0}$.
\end{lem}
In order to convert the above error guarantees into convergence rate results of the form (\ref{eq:lasso_rates_1})--(\ref{eq:lasso_rates_3}) for Lasso estimators implied by the Algorithm~\ref{alg:Data-Driven-Penalization} penalization, we proceed by establishing:

\begin{enumerate}
\item \label{enu:DeviationBnd} \emph{Deviation bound}: $\P(\lambda^\ast_n/n\geqslant c\max_{i\in[p]}\Vert \bS_{n,i}\Vert _{\ell_{\infty}})\to1$
holds for $\lambda^{\ast}_n$ given in (\ref{eq:PenaltyLevelPractice}) and $\{\bS_{n,i}\}_{i\in[p]}$ in (\ref{eq:def_score_ideal}).
\item\label{enu:IdealPenaltyLoadingCtrl} \emph{Ideal penalty loadings control}: There are constants $\underline{\upsilon}^{0},\overline{\upsilon}^{0}\in(0,\infty)$ such that \newline  $\mathrm{P}\left(\underline{\upsilon}^{0}\leqslant\widehat{\upsilon}_{\min}^{0}\leqslant\widehat{\upsilon}_{\max}^{0}\leqslant\overline{\upsilon}^{0}\right)\to1$. This finding also yields~$\P(\widehat{\mu}_{0}\leqslant\overline\upsilon^0/\underline\upsilon^0)\to1$.
\item\label{enu:RestrictedEigenvalueCtrl} \emph{Restricted eigenvalue bound}: $\P(\widehat\kappa_{\widehat c_{0}\widehat{\mu}_{0}}^{2}\geqslant d/2)\to1$.
\item\label{enu:AsymptoticValidityDataDrivenLoadings} \emph{Asymptotic validity of data-driven penalty loadings:} The (final) data-driven penalty loadings $\{\widehat{\upsilon}^{(K)}_{i,j}\}_{(i,j)\in[p]\times[pq]}$ obtained from Algorithm \ref{alg:Data-Driven-Penalization}
are asymptotically valid in the sense of (\ref{eq:AsymptoticallyValidPenaltyLoadings})--(\ref{eq:AsymptoticallyValidPenaltyLoadings2}). This implies control of the random~$\widehat{\ell}$,~$\widehat{u}$, and thus~$\widehat{c}_0$.
\end{enumerate}

In the following four subsections we formalize Items \ref{enu:DeviationBnd}--\ref{enu:AsymptoticValidityDataDrivenLoadings} in the above enumeration, from which Theorem~\ref{thm:Rates-for-LASSO-data-driven-loadings} will follow.

\subsection*{Item \ref{enu:DeviationBnd}: Deviation Bound\label{subsec:Deviation-Bound}}
Recall~$c$ and~$\gamma_n$ given in (\ref{eq:tuning_c_and_gamma}), the penalty level $\lambda^\ast_n$ in (\ref{eq:PenaltyLevelPractice}) and the scores $\{\bS_{n,i}\}_{i\in[p]}$ based on ideal penalty loadings in (\ref{eq:def_score_ideal}).
\begin{lem}
[\textbf{Deviation Bound}]\label{lem:Deviation-Bound} Let
Assumptions \ref{assu:Innovations}, \ref{assu:Companion} and \ref{assu:CovarianceZ} hold, and let $(\ln p)^{(1+4\tau)/\tau}=o(n)$ for $\tau\in(\textstyle{\frac{1}{2}},1]$ provided by Assumption \ref{assu:Companion}.\ref{enu:RowNormDecay}. Then
\[
\P\del[2]{\lambda^\ast_n/n\geqslant c\max_{i\in[p]}\left\Vert \bS_{n,i}\right\Vert _{\ell_{\infty}}}\to1.
\]
\end{lem}

 \subsection*{Item \ref{enu:IdealPenaltyLoadingCtrl}: Ideal Penalty Loadings Control \label{subsec:Ideal-Penalty-Loading}}
The following lemma states that the ideal penalty loadings in (\ref{eq:PenaltyLoadingsIdeal_p2}) are bounded from
above and away from zero with probability approaching one. We obtain this control over the ideal penalty
loadings (in Lemma \ref{lem:Ideal-Penalty-Loading-Control}) as well as over the
restricted eigenvalue (in Lemma \ref{lem:Restricted-Eigenvalue-Bound}) 
by establishing a new maximal inequality (Lemma \ref{lem:MaximalInequalityFunctionalDependence})
for sums of dependent sub-Weibull random variables, which may be
of independent interest.
\begin{lem}
[\textbf{Ideal Penalty Loadings Control}]\label{lem:Ideal-Penalty-Loading-Control}
Let Assumptions \ref{assu:Innovations}, \ref{assu:Companion}.\ref{enu:SpectralRadius},  
\ref{assu:Companion}.\ref{enu:RowNormDecay} and \ref{assu:CovarianceZ}
hold, and suppose that $(\ln p)^{C(\alpha,\tau)}=o(n)$ with $C(\alpha,\tau)$ given in Assumption \ref{assu:RowSparsity} for $\alpha$ provided by Assumption \ref{assu:Innovations}.\ref{enu:EpsSubWeibullBeta} and $\tau\in (\textstyle{\frac{1}{2}},1]$ provided by Assumption \ref{assu:Companion}.\ref{enu:RowNormDecay}.
Then there are constants $\underline{\upsilon}^{0},\overline{\upsilon}^{0}\in(0,\infty)$,
depending only on $a_1,a_2,\alpha,b_1, b_2, d$ and $\tau$, such that
\begin{equation}\label{eq:IdealPenaltyLoadingsUnderCtrl}
\mathrm{P}\left(\underline{\upsilon}^{0}\leqslant\widehat{\upsilon}_{\min}^{0}\leqslant\widehat{\upsilon}_{\max}^{0}\leqslant\overline{\upsilon}^{0}\right)\to1,
\end{equation}
and, thus, $\P(\widehat{\mu}_{0}\leqslant\overline\upsilon^0/\underline\upsilon^0)\to1$.
\end{lem}

\subsection*{Item \ref{enu:RestrictedEigenvalueCtrl}: Restricted Eigenvalue Bound \label{subsec:Restricted-Eigenvalue-Bound}}
A key ingredient in bounding the restricted eigenvalue~$\widehat\kappa_{\widehat c_{0}\widehat{\mu}_{0}}^2$ away from zero
is to prove that the empirical covariance matrix
\[
\widehat{\bSigma}_{\bZ}:=\frac{1}{n}\bfX^\top\bfX=\frac{1}{n}\sum_{t=0}^{n-1}\bZ_{t}\bZ_{t}^\top
\] 
of the regressors $\{\bZ_{t}\}_{t=0}^{n-1}$ is ``close'' to $\bSigma_{\bZ}=\E[\bZ_0 \bZ_0^\top]$. The formal statement is as follows.
\begin{lem}
[\textbf{Covariance Consistency}]\label{lem:Covariance-Consistency}
Let Assumptions \ref{assu:Innovations}.\ref{enu:EpsIID}, \ref{assu:Innovations}.\ref{enu:EpsSubWeibullBeta} and
\ref{assu:Companion} hold,
and suppose that $\left(\ln p\right)^{1/\tau+\widetilde{C}(2,\alpha)}=o(n)$ with $\widetilde{C}(2,\alpha)$ given in Assumption \ref{assu:RowSparsity} for $\alpha$ and $\tau$ provided respectively by Assumptions \ref{assu:Innovations}.\ref{enu:EpsSubWeibullBeta} and \ref{assu:Companion}.\ref{enu:RowNormDecay}.
Then
\[
\max_{(i,j)\in[pq]^2}\envert[1]{(\widehat{\bSigma}_{\bZ})_{i,j}-(\bSigma_{\bZ})_{i,j}}
\lesssim_{\P} \sqrt{\frac{(\ln(pn))^{1/\tau+\widetilde{C}(2,\alpha)}}{n}}=o(1).
\]
\end{lem}
%
Combining the ideal penalty loading control of Lemma \ref{lem:Ideal-Penalty-Loading-Control}
with Lemma \ref{lem:Covariance-Consistency},
we can bound the restricted eigenvalue away from zero.

\begin{lem}
[\textbf{Restricted Eigenvalue Bound}]\label{lem:Restricted-Eigenvalue-Bound}
Let Assumptions \ref{assu:Innovations}.\ref{enu:EpsIID}, \ref{assu:Innovations}.\ref{enu:EpsSubWeibullBeta},
\ref{assu:Companion} and \ref{assu:CovarianceZ} hold and suppose that $s^2(\ln(pn))^{1/\tau+\widetilde{C}(2,\alpha)}=o(n)$ with
$\widetilde{C}(2,\alpha)$ given in Assumption \ref{assu:RowSparsity} for $\alpha$ and $\tau$ provided respectively by Assumptions \ref{assu:Innovations}.\ref{enu:EpsSubWeibullBeta} and \ref{assu:Companion}.\ref{enu:RowNormDecay}. Consider any collection $\{\widehat\bUpsilon_{i}\}_{i\in[p]}$ of asymptotically valid penalty loadings in the sense of~(\ref{eq:AsymptoticallyValidPenaltyLoadings}) holding with probability approaching one and (\ref{eq:AsymptoticallyValidPenaltyLoadings2}).
Then
\[
\P\del[1]{\widehat\kappa_{\widehat c_{0}\widehat{\mu}_{0}}^{2}\geqslant d/2}\to1,
\]
where~$\widehat c_{0}:=(\widehat uc+1)/(\widehat\ell c-1)$, and $\widehat \kappa_{\widehat c_{0}\widehat{\mu}_{0}}$ is defined via (\ref{eq:RestrictedEigenvalue}) with $C=\widehat c_{0}\widehat{\mu}_{0}$.
\end{lem}

\subsection*{Item \ref{enu:AsymptoticValidityDataDrivenLoadings}: Asymptotic Validity of Data-Driven Penalty Loadings \label{subsec:Validity-of-Data-Driven}}

Lemma \ref{lem:Restricted-Eigenvalue-Bound} applies to any asymptotically valid penalty loadings. The following result states that the data-driven penalty loadings obtained from
Algorithm \ref{alg:Data-Driven-Penalization} satisfy this requirement.

\begin{lem}
[\textbf{Asymptotic Validity of Data-Driven Penalty Loadings}]\label{lem:AsymptoticValidityDataDrivenPenaltyLoadings}
Let Assumptions \ref{assu:Innovations}, \ref{assu:Companion}, \ref{assu:CovarianceZ} and \ref{assu:RowSparsity} hold. Then, for any fixed number of updates $K\in\N_{0}$, the penalty loadings $\{\widehat{\upsilon}_{i,j}^{(K)}\}_{(i,j)\in[p]\times[pq]}$ constructed by Algorithm \ref{alg:Data-Driven-Penalization} are asymptotically valid.
\end{lem}

The proof of Lemma \ref{lem:AsymptoticValidityDataDrivenPenaltyLoadings}
is inspired by that of \citet[Lemma 11]{belloni_sparse_2012}. The main idea behind the proof is to show that the penalty loadings of Algorithm~\ref{alg:Data-Driven-Penalization} are close to the blocking-based ideal ones. To this end, we use our maximal inequality for sums of dependent sub-Weibull random variables (Lemma \ref{lem:MaximalInequalityFunctionalDependence}).

\section{Proofs of Lemmas \ref{lem:MasterLemma}--\ref{lem:AsymptoticValidityDataDrivenPenaltyLoadings} and Theorem \ref{thm:Rates-for-LASSO-data-driven-loadings}}
\subsection{Proof of Master Lemma (Lemma \ref{lem:MasterLemma})}
\begin{proof}
[\sc{Proof of Lemma \ref{lem:MasterLemma}}] The proof of the lemma
is similar to that of \citet[Lemma 6]{belloni_sparse_2012},
which, in turn, builds on a strategy attributed to \citet{bickel_simultaneous_2009}.
Let $T_{i}:=\{j\in[pq]:\beta_{0i,j}\neq0\}$. Recall that $\widehat{\bUpsilon}_{i}^{0}$
denotes the diagonal matrix $\mathrm{diag}(\widehat{\upsilon}_{i,1}^{0},\dotsc,\widehat{\upsilon}_{i,pq}^{0})$
of ideal penalty loadings, and let $\widehat{\bdelta}_{i}:=\widehat{\bbeta}_{i}-\bbeta_{0i}$. \emph{In the course of this proof only}, for any~$T\subseteq[pq]$ and~$\bm{v}\in\R^{pq}$, we denote by~$\bm{v}_{T}\in\R^{pq}$ the vector with~$\bm{v}_{T,i}=\bm{v}_i$ if~$i\in T$ and~$\bm{v}_{T,i}=0$ otherwise (not to be confused with the \emph{subvector} picked out by $T$, which lies in $\R^{|T|}$).

Expanding the quadratic part (\ref{eq:Qihat}) of the objective function in (\ref{eq:LASSOVector}), when $\lambda/n\geqslant c\Vert\bS_{i}\Vert_{\ell_{\infty}}$ we get
\begin{align*}
\widehat{Q}_i(\widehat\bbeta_i)-\widehat{Q}_i(\bbeta_{0i})
& =\frac{1}{n}\sum_{t=1}(\bZ_{t-1}^\top\widehat{\bdelta}_{i})^{2}-\frac{2}{n}\sum_{t=1}^{n}\varepsilon_{t,i}\bZ_{t-1}^\top\widehat{\bdelta}_{i}\\
 & =\Vert\widehat{\bdelta}_{i}\Vert_{2,n}^{2}-\underbrace{\left[\frac{2}{n}\sum_{t=1}^{n}(\widehat{\bUpsilon}_{i}^{0})^{-1}\bZ_{t-1}\varepsilon_{t,i}\right]^\top}_{=\bS_{n,i}^\top}(\widehat{\bUpsilon}_{i}^{0}\widehat{\bdelta}_{i})\\
 & \geqslant\Vert\widehat{\bdelta}_{i}\Vert_{2,n}^{2}-\Vert \bS_{n,i}\Vert_{\ell_{\infty}}\Vert\widehat{\bUpsilon}_{i}^{0}\widehat{\bdelta}_{i}\Vert_{\ell_{1}}\tag{H{\"o}lder}\\
 & \geqslant\Vert\widehat{\bdelta}_{i}\Vert_{2,n}^{2}-\frac{\lambda}{cn}\Vert\widehat{\bUpsilon}_{i}^{0}\widehat{\bdelta}_{i}\Vert_{\ell_{1}},
\end{align*}
which we can write as
\begin{equation}\label{eq:QdiffLowerBnd}
\widehat{Q}_i(\widehat\bbeta_i)-\widehat{Q}_i(\bbeta_{0i})\geqslant\Vert\widehat{\bdelta}_{i}\Vert_{2,n}^{2}-\frac{\lambda}{cn}\left(\Vert\widehat{\bUpsilon}_{i}^{0}\widehat{\bdelta}_{iT_{i}}\Vert_{\ell_{1}}+\Vert\widehat{\bUpsilon}_{i}^{0}\widehat{\bdelta}_{iT_{i}^{c}}\Vert_{\ell_{1}}\right).   
\end{equation}
Since $\widehat{\bbeta}_{i}$ is a minimizer of $\R^{pq}\ni\bbeta\mapsto\widehat Q_i(\bbeta)+(\lambda/n)\normn{\widehat\bUpsilon_i\bbeta}_{\ell_1}$, the triangle inequality shows
\begin{align*}
\widehat{Q}_i(\widehat\bbeta_i)-\widehat{Q}_i(\bbeta_{0i}) & \leqslant\frac{\lambda}{n}\left(\Vert\widehat{\bUpsilon}_{i}\bbeta_{0i}\Vert_{\ell_{1}}-\Vert\widehat{\bUpsilon}_{i}\widehat{\bbeta}_{i}\Vert_{\ell_{1}}\right)
  \leqslant\frac{\lambda}{n}\left(\Vert\widehat{\bUpsilon}_{i}\widehat{\bdelta}_{iT_{i}}\Vert_{\ell_{1}}-\Vert\widehat{\bUpsilon}_{i}\widehat{\bdelta}_{iT_{i}^{c}}\Vert_{\ell_{1}}\right),
\end{align*}
so from the penalty loadings $\widehat{\bUpsilon}_{i}$
satisfying (\ref{eq:AsymptoticallyValidPenaltyLoadings}), we arrive at
\begin{equation}\label{eq:QdiffUpperBnd}
\widehat{Q}_i(\widehat\bbeta_i)-\widehat{Q}_i(\bbeta_{0i})\leqslant\frac{\lambda}{n}\left(\widehat u\Vert\widehat{\bUpsilon}_{i}^{0}\widehat{\bdelta}_{iT_{i}}\Vert_{\ell_{1}}-\widehat\ell\Vert\widehat{\bUpsilon}_{i}^{0}\widehat{\bdelta}_{iT_{i}^{c}}\Vert_{\ell_{1}}\right).    
\end{equation}
The two inequalities (\ref{eq:QdiffLowerBnd}) and (\ref{eq:QdiffUpperBnd}) combine to yield
\begin{equation}
\Vert\widehat{\bdelta}_{i}\Vert_{2,n}^{2}+\frac{\lambda}{n}\left(\widehat\ell-\frac{1}{c}\right)\Vert\widehat{\bUpsilon}_{i}^{0}\widehat{\bdelta}_{iT_{i}^{c}}\Vert_{\ell_{1}}\leqslant\frac{\lambda}{n}\left(\widehat u+\frac{1}{c}\right)\Vert\widehat{\bUpsilon}_{i}^{0}\widehat{\bdelta}_{iT_{i}}\Vert_{\ell_{1}}.\label{eq:LASSOProofIntermediateIneq}
\end{equation}
Since $1/c<\widehat\ell\leqslant\widehat u$, we deduce that 
\begin{equation}
\widehat{\upsilon}_{\min}^{0}\Vert\widehat{\bdelta}_{iT_{i}^{c}}\Vert_{\ell_{1}}\leqslant\Vert\widehat{\bUpsilon}_{i}^{0}\widehat{\bdelta}_{iT_{i}^{c}}\Vert_{\ell_{1}}\leqslant\frac{\widehat u+1/c}{\widehat\ell-1/c}\Vert\widehat{\bUpsilon}_{i}^{0}\widehat{\bdelta}_{iT_{i}}\Vert_{\ell_{1}}\leqslant\frac{\widehat u+1/c}{\widehat\ell-1/c}\widehat{\upsilon}_{\max}^{0}\Vert\widehat{\bdelta}_{iT_{i}}\Vert_{\ell_{1}}.\label{eq:deltahat_i_in_ideally_weighted_cone}
\end{equation}
It follows that
\[
\Vert\widehat{\bdelta}_{iT_{i}^{c}}\Vert_{\ell_{1}}\leqslant\frac{\widehat u+1/c}{\widehat\ell-1/c}\cdot\frac{\widehat{\upsilon}_{\max}^{0}}{\widehat{\upsilon}_{\min}^{0}}\Vert\widehat{\bdelta}_{iT_{i}}\Vert_{\ell_{1}}=\widehat c_{0}\widehat{\mu}_{0}\Vert\widehat{\bdelta}_{iT_{i}}\Vert_{\ell_{1}},
\]
i.e.~$\widehat{\bdelta}_{i}$ lies in the restricted set $\mathcal{R}_{\widehat c_{0}\widehat{\mu}_{0},T_{i}}$ defined via (\ref{eq:RestrictedSet}) with $C=\widehat c_0\widehat\mu_0$ and $T=T_i$. From the definition (\ref{eq:RestrictedEigenvalue})
of the restricted eigenvalue we see that $\Vert\widehat{\bdelta}_{i}\Vert_{\ell_{2}}\leqslant\Vert\widehat{\bdelta}_{i}\Vert_{2,n}/\widehat\kappa_{\widehat c_{0}\widehat{\mu}_{0}},$
so by the Cauchy-Schwarz inequality and recalling that
$s=\max_{i\in[p]}|T_{i}|$ [cf.~(\ref{eq:def_s_sparsity})]
\begin{equation}\label{eq:IdeallyWeightedEll1IntermediateBnd}
\Vert\widehat{\bUpsilon}_{i}^{0}\widehat{\bdelta}_{iT_{i}}\Vert_{\ell_{1}} \leqslant\widehat{\upsilon}_{\max}^{0}\Vert\widehat{\bdelta}_{iT_{i}}\Vert_{\ell_{1}}\leqslant\widehat{\upsilon}_{\max}^{0}\sqrt{s}\Vert\widehat{\bdelta}_{iT_{i}}\Vert_{\ell_{2}}\leqslant\widehat{\upsilon}_{\max}^{0}\sqrt{s}\Vert\widehat{\bdelta}_{i}\Vert_{\ell_{2}}\leqslant\frac{\widehat{\upsilon}_{\max}^{0}\sqrt{s}\Vert\widehat{\bdelta}_{i}\Vert_{2,n}}{\widehat\kappa_{\widehat c_{0}\widehat{\mu}_{0}}}.
\end{equation}
Hence, continuing our string of inequalities from (\ref{eq:LASSOProofIntermediateIneq}), we see that
\[
\Vert\widehat{\bdelta}_{i}\Vert_{2,n}^{2}+\underbrace{\frac{\lambda}{n}\left(\widehat\ell-\frac{1}{c}\right)\Vert\widehat{\bUpsilon}_{i}^{0}\widehat{\bdelta}_{iT_{i}^{c}}\Vert_{\ell_{1}}}_{\geqslant0}\leqslant\frac{\lambda}{n}\left(\widehat u+\frac{1}{c}\right)\frac{\widehat{\upsilon}_{\max}^{0}\sqrt{s}\Vert\widehat{\bdelta}_{i}\Vert_{2,n}}{\widehat\kappa_{\widehat c_{0}\widehat{\mu}_{0}}},
\]
from which we arrive at the prediction error bound
\[
\Vert\widehat{\bdelta}_{i}\Vert_{2,n}\leqslant\left(\widehat u+\frac{1}{c}\right)\frac{\lambda\widehat{\upsilon}_{\max}^{0}\sqrt{s}}{n\widehat\kappa_{\widehat c_{0}\widehat{\mu}_{0}}}.
\]
As the right-hand bound is independent of $i$, the uniformity in $i\in[p]$ follows. The $\ell_{2}$ error
bound then follows from the already established $\Vert\widehat{\bdelta}_{i}\Vert_{\ell_{2}}\leqslant\Vert\widehat{\bdelta}_{i}\Vert_{2,n}/\widehat\kappa_{\widehat c_{0}\widehat{\mu}_{0}},$
and the $\ell_{1}$ bound subsequently follows from $\Vert\widehat{\bdelta}_{i}\Vert_{\ell_{1}}\leqslant\left(1+\widehat c_{0}\widehat{\mu}_{0}\right)\Vert\widehat{\bdelta}_{iT_{i}}\Vert_{\ell_{1}}$
and $\Vert\widehat{\bdelta}_{iT_{i}}\Vert_{\ell_{1}}\leqslant\sqrt{s}\Vert\widehat{\bdelta}_{i}\Vert_{\ell_{2}}.$
\end{proof}

\subsection{Proof of Deviation Bound (Lemma \ref{lem:Deviation-Bound})}
We first state and prove some supporting lemmas.
\begin{lem}\label{lem:ZandEpsZNormBnds} 
If Assumptions
\ref{assu:Innovations}.\ref{enu:EpsIID}, \ref{assu:Innovations}.\ref{enu:EpsSubWeibullBeta}, \ref{assu:Companion}.\ref{enu:SpectralRadius} and \ref{assu:Companion}.\ref{enu:RowNormDecay}
hold, then, for $\alpha$ provided by Assumption \ref{assu:Innovations}.\ref{enu:EpsSubWeibullBeta}, the following statements hold.
\begin{enumerate}[(1)]
\item\label{enu:ZiMarginalSubWeibullNormUnifBdd} For all $i\in[pq]$, 
\[
\left\Vert Z_{0,i}\right\Vert_{\psi_{\alpha}}\leqslant \underbrace{a_1b_1 \left(1+ \frac{\Gamma(1/\tau)}{b_2^{1/\tau}\tau} \right)}_{=:C(a_1,b_1,b_2,\tau)},
\]
where $\Gamma$ denotes the gamma function, $\Gamma(z):=\int_{0}^{\infty}u^{z-1}\mathrm{e}^{-u}\dif u$.
\item\label{enu:ZiUnifMomentBnd} For all $r\in[1,\infty)$ and all $i\in[pq]$,
\[
\left\Vert Z_{0,i}\right\Vert _{r}\leqslant C(a_1,b_1,b_2,\tau)r^{1/\alpha}.
\]
\item\label{enu:EpsiZjUnifMomentBnd} For all $r\in[1,\infty)$, all $i\in[pq]$, and all $j\in[pq]$
\[
\left\Vert\varepsilon_{0,i}Z_{-1,j}\right\Vert_r\leqslant C(a_1,b_1,b_2,\tau)a_1(2r)^{2/\alpha}.
\]
\item\label{enu:Zbeta_iSubWeibullNormUnifBdd} For all $i\in[p]$,
\[
\enVert[1]{\bZ_0^\top\bbeta_{0i}}_{\psi_\alpha}\leqslant  \frac{a_1b_1\Gamma(1/\tau)}{b_2^{1/\tau}\tau}.
\]
\end{enumerate}
\end{lem}

\begin{proof}[\sc{Proof of Lemma \ref{lem:ZandEpsZNormBnds}}]
Strict stationarity (which follows from Assumptions \ref{assu:Innovations}.\ref{enu:EpsIID}
and \ref{assu:Companion}.\ref{enu:SpectralRadius})\textcolor{red}{{}
}allows us to represent the companion VAR(1) in (\ref{eq:VARcompanion}) as an infinite vector
moving average [VMA$(\infty)$],
\begin{equation}
\bZ_{t}=\sum_{h=0}^{\infty}\widetilde{\bTheta}_{0}^{h}\widetilde{\bepsilon}_{t-h},\quad t\in\Z,\label{eq:VMAInfRepresentation}
\end{equation}
where the series converges absolutely almost surely. Moreover, under Assumption \ref{assu:Companion}.\ref{enu:RowNormDecay} the series converges in $\left\Vert \cdot \right\Vert_{\psi_\alpha}$ norm. To establish Claim \ref{enu:ZiMarginalSubWeibullNormUnifBdd}, since for
$i\in[pq]$, $Z_{t-1,i}=Y_{t-j,k}$ for some $(j,k)\in[q]\times[p]$, for any $\alpha$ provided by Assumption \ref{assu:Innovations}.\ref{enu:EpsSubWeibullBeta}, we have
\begin{align*}
\left\Vert Z_{0,i}\right\Vert _{\psi_{\alpha}}
&\leqslant\max_{j\in[p]}\left\Vert Y_{0,j}\right\Vert _{\psi_{\alpha}}\tag{strict stationarity}\\
&=\max_{j\in[p]}\enVert[4]{\sum_{h=0}^{\infty}(\widetilde{\bTheta}_{0}^{h})_{j,\cdot}\widetilde{\bepsilon}_{0-h}}_{\psi_{\alpha}}\tag{VMA$(\infty)$ representation}\\
&=\max_{j\in[p]}\enVert[4]{\sum_{h=0}^{\infty}(\widetilde{\bTheta}_{0}^{h})_{j,1:p}\bepsilon_{0-h}}_{\psi_{\alpha}}\tag{$\widetilde{\bepsilon}_{p+1:pq}\equiv \mathbf{0}_{p(q-1)\times1}$}\\
&\leqslant\max_{j\in[p]}\sum_{h=0}^{\infty}\big\Vert (\widetilde{\bTheta}_{0}^{h})_{j,1:p}\bepsilon_{0-h}\big\Vert _{\psi_{\alpha}}\tag{countable sub-additivity}\\
&\leqslant\left\Vert \bepsilon_{0}\right\Vert _{\psi_{\alpha}}\max_{j\in[p]}\sum_{h=0}^{\infty}\big\Vert (\widetilde{\bTheta}_{0}^{h})_{j,1:p}\big\Vert_{\ell_{2}}\tag{sub-Weibullness and stationarity}\\
&\leqslant a_1 b_1 \sum_{h=0}^{\infty} \mathrm{e}^{-b_2h^\tau}.\tag{Assumptions \ref{assu:Innovations}.\ref{enu:EpsSubWeibullBeta} and \ref{assu:Companion}.\ref{enu:RowNormDecay}}
\end{align*}
The result now follows by noting that $\tau\in(0,1]$ implies both $\sum_{h=0}^{\infty} \mathrm{e}^{-b_2h^\tau}=1+\sum_{h=1}^{\infty} \mathrm{e}^{-b_2h^\tau}$
\begin{equation}
 \sum_{h=1}^{\infty}\mathrm{e}^{-b_2h^\tau}\leqslant  \int_0^\infty \mathrm{e}^{-b_2x^\tau}\dif x = \frac{\Gamma(1/\tau)}{b_2^{1/\tau}\tau},   \label{eq:bound_integral},
\end{equation}
where the equality follows from a change of variables.
Claim \ref{enu:ZiUnifMomentBnd} is immediate from Claim \ref{enu:ZiMarginalSubWeibullNormUnifBdd} and the definition of a sub-Weibull norm. To establish Claim \ref{enu:EpsiZjUnifMomentBnd}, note that for any $r\in[1,\infty)$, $i\in[p]$ and $j\in[pq]$, using H{\"o}lder's inequality, we get
\begin{align*}
\left\Vert
\varepsilon_{0,i}Z_{-1,j}\right\Vert _{r}
&\leqslant\left\Vert \varepsilon_{0,i}\right\Vert _{2r}\left\Vert Z_{0,j}\right\Vert _{2r}
\leqslant (2r)^{1/\alpha}\left\Vert \varepsilon_{0,i}\right\Vert _{\psi_\alpha}\left\Vert Z_{0,j}\right\Vert _{2r} \\
&\leqslant (2r)^{1/\alpha}\left\Vert \bepsilon_{0}\right\Vert _{\psi_\alpha}\left\Vert Z_{0,j}\right\Vert _{2r} 
\leqslant a_1(2r)^{1/\alpha}\left\Vert Z_{0,j}\right\Vert _{2r}.
\end{align*}
Claim \ref{enu:EpsiZjUnifMomentBnd} now follows from Claim \ref{enu:ZiUnifMomentBnd}.
For Claim \ref{enu:Zbeta_iSubWeibullNormUnifBdd}, since $\bbeta_{0i}^{\top}$ is the $i$\textsuperscript{th} row of $\widetilde\bTheta_0$, letting $\be_i\in\R^{pq}$ be the $i$\textsuperscript{th} unit vector, from (\ref{eq:VMAInfRepresentation}) we see that
\[
\bZ_0^\top\bbeta_{0i}
=\sum_{h=0}^{\infty}\bbeta_{0i}^{\top}\widetilde{\bTheta}_{0}^{h}\widetilde{\bepsilon}_{0-h}=\sum_{h=0}^{\infty}\be_i^\top\widetilde{\bTheta}_{0}^{h+1}\widetilde{\bepsilon}_{0-h}=\sum_{h=0}^{\infty}(\widetilde{\bTheta}_{0}^{h+1})_{i,1:p}\bepsilon_{0-h}.
\]
It follows that
\begin{align*}
\enVert[1]{\bZ_0^\top\bbeta_{0i}}_{\psi_\alpha}
&\leqslant\sum_{h=0}^{\infty}\big\|(\widetilde{\bTheta}_{0}^{h+1})_{i,1:p}\bepsilon_{0}\big\|_{\psi_\alpha}\tag{countable sub-additivity and stationarity}\\
&\leqslant\|\bepsilon_{0}\|_{\psi_\alpha}\sum_{h=0}^{\infty}\big\|(\widetilde{\bTheta}_{0}^{h+1})_{i,1:p}\big\|_{\ell_2}\tag{sub-Weibullness}\\
&\leqslant a_1 b_1\sum_{h=0}^{\infty}\mathrm{e}^{-b_2(h+1)^{\tau}}\tag{Assumptions \ref{assu:Innovations}.\ref{enu:EpsSubWeibullBeta} and \ref{assu:Companion}.\ref{enu:RowNormDecay}},
\end{align*}
which holds uniformly in $i\in[p]$. The result now follows by noting that $\sum_{h=1}^{\infty}\mathrm{e}^{-b_2h^\tau}\leqslant  \Gamma(1/\tau)/(b_2^{1/\tau}\tau)$ for $\tau\in(0,1]$ [using \eqref{eq:bound_integral}].
\end{proof}

\begin{lem}\label{lem:norm_decay}
If Assumption \ref{assu:Companion}.\ref{enu:RowNormDecay} holds, then
\[
\max_{i\in[pq]}\big\Vert(\widetilde{\bTheta}_{0}^{h})_{i,1:p}\big\Vert_{\ell_{2}}\leqslant C_1\mathrm{e}^{-C_2h^\tau}\;\text{for all}\;h\in\N,
\]
with constants $C_1 :=(1+b_1)\mathrm{e}^{b_2q} \in(0,\infty)$ and $C_2 := b_2/q^\tau \in(0,\infty)$. 
\end{lem}    

\begin{proof}[\sc{Proof of Lemma \ref{lem:norm_decay}}]
For $q=1$ the result is immediate from Assumption \ref{assu:Companion}.\ref{enu:RowNormDecay}. For $q>1$, by construction of the companion matrix~$\widetilde{\bm{\Theta}}_0$, we have $(\widetilde\bTheta_0)_{p+1:pq,1:pq}=[\bI_{q-1}\otimes\bI_{p}:\mathbf{1}_{q-1}\otimes \mathbf{0}_{p\times p}]$ ($\otimes$ being the Kronecker product and $\mathbf{1}_{q-1}\in\R^{q-1}$ a column vector of ones). In particular,
\[
(\widetilde{\bTheta}_{0})_{p+1:2p,1:p} =\mathbf{I}_{p}\quad\text{and}\quad(\widetilde{\bTheta}_{0})_{kp+1:(k+1)p,1:p} =\mathbf{0}_{p\times p},\quad k\in \{2,\dotsc, q-1 \}.
\]
From this structure, we also deduce the recursion
\begin{align*}
(\widetilde{\bTheta}_{0}^{h})_{kp+1:(k+1)p,1:p} & =(\widetilde{\bTheta}_{0}^{h-1})_{(k-1)p+1:kp,1:p},\quad k\in\{1,\dotsc,q-1\},\quad h\in\{2,3,\dotsc\}.
\end{align*}
Hence, for $h\geqslant 1$, it follows that
\[
(\widetilde{\bTheta}_{0}^{h})_{\cdot,1:p}=
\begin{pmatrix}(\widetilde{\bTheta}_{0}^{h})_{1:p,1:p}\\
(\widetilde{\bTheta}_{0}^{h-1})_{1:p,1:p}\\
(\widetilde{\bTheta}_{0}^{h-2})_{1:p,1:p}\\
\vdots\\
(\widetilde{\bTheta}_{0}^{h-(q-2)})_{1:p,1:p}\\
(\widetilde{\bTheta}_{0}^{h-(q-1)})_{1:p,1:p}
\end{pmatrix},
\]
with the conventions that $(\widetilde{\bTheta}_{0}^{j})_{1:p,1:p}=\mathbf{0}_{p\times p}$ for integers $j<0$ and $(\widetilde{\bTheta}_{0}^{0})_{1:p,1:p}=\mathbf{I}_{p}$.
Using Assumption \ref{assu:Companion}.\ref{enu:RowNormDecay}, for $h\geqslant 1$ we therefore get
\begin{align*}
\max_{i\in[pq]}\big\Vert(\widetilde{\bTheta}_{0}^{h})_{i,1:p}\big\Vert_{\ell_{2}} & =\max_{j\in \{ h-(q-1),\dotsc,h \}}\max_{i\in[p]}\big\Vert(\widetilde{\bTheta}_{0}^{j})_{i,1:p}\big\Vert_{\ell_{2}}\\
 & \leqslant \max_{j\in \{ h-(q-1),\dotsc,h \}} \left\{ b_{1}\mathrm{e}^{-b_{2}j^\tau} \mathbf{1}(j>0) +\mathbf{1}(j=0)  \right\}.
\end{align*}
We next go case by case. On one hand, if $h \leqslant q$, then
\begin{align*}
    \max_{j\in \{ h-(q-1),\dotsc,h \}} \left\{ b_{1}\mathrm{e}^{-b_{2}j^\tau} \mathbf{1}(j>0) +\mathbf{1}(j=0)  \right\} \leqslant b_1 +1 \leqslant (1+b_1)\mathrm{e}^{b_2(q-h)}\leqslant (1+b_1)\mathrm{e}^{b_2q}\mathrm{e}^{-b_2h^{\tau}}.
\end{align*}
On the other hand, if $h > q$, then $h-(q-1)>0$ and, thus,
\begin{align*}
     \max_{j\in \{ h-(q-1),\dotsc,h \}} \left\{ b_{1}\mathrm{e}^{-b_{2}j^\tau} \mathbf{1}(j>0) +\mathbf{1}(j=0)  \right\} =  b_{1}\mathrm{e}^{-b_{2}\sbr[0]{h-(q-1)}^\tau}.
\end{align*}
For $\tau\in(0,1]$, using that $x\geqslant y\geqslant1$ implies $x-y+1\geqslant x/y$,\footnote{To establish this inequality, note that for $x\geqslant y\geqslant 1$, the desired $x-y+1\geqslant x/y$ is equivalent to $-y^2+(1+x)y-y\geqslant 0$. The left-hand side quadratic in $y$ is inverse U-shaped with roots equal to one and $x$.} we get $[h-(q-1)]^\tau \geqslant (h/q)^\tau$ and, thus, 
\begin{align*}
   \max_{j\in \{ h-(q-1),\dotsc,h \}} \left\{ b_{1}\mathrm{e}^{-b_{2}j^\tau} \mathbf{1}(j>0) +\mathbf{1}(j=0)  \right\} & \leqslant (1+b_1)\mathrm{e}^{-(b_2/q^\tau)h^\tau}. 
\end{align*}
Combining the bounds in the cases $h\leqslant q$ and $h > q$, we see that for any $\tau\in(0,1]$
\begin{align*}
\max_{i\in[pq]}\big\Vert(\widetilde{\bTheta}_{0}^{h})_{i,1:p}\big\Vert_{\ell_{2}}
&\leqslant
\mathbf{1}(h\leqslant q)(1+b_1)\mathrm{e}^{b_2q}\mathrm{e}^{-b_2h^{\tau}}+
\mathbf{1}(h>q)(1+b_1)\mathrm{e}^{-(b_2/q^\tau)h^\tau}\\
&\leqslant (1+b_1)\mathrm{e}^{b_2q} \mathrm{e}^{-(b_2/q^\tau)h^\tau}.
\end{align*}

\end{proof}

\begin{lem}\label{lem:OutcomeAndScoreProcessesUniformlyGMC} 
Let Assumptions
\ref{assu:Innovations}.\ref{enu:EpsIID}, \ref{assu:Innovations}.\ref{enu:EpsSubWeibullBeta} and
\ref{assu:Companion}
hold. Then for any $r\in[1,\infty)$, there are constants $C_1,C_2,C_3,C_4\in(0,\infty)$,
depending only on $a_1,\alpha,b_1,b_2,q,\overline{q},\tau$ and $r$, such that
\begin{enumerate}[(1)]
\item\label{enu:ZiUnifGMC} $\max_{i\in[pq]}\Delta_{r}^{Z_{\cdot,i}}(h)\leqslant C_1\mathrm{e}^{-C_3h^\tau}$ for all $h\in\N$, 
and
\item\label{enu:EpsiZjUnifGMC} $\max_{(i,j)\in[p] \times [pq]}\Delta_{r}^{\varepsilon_{\cdot,i}Z_{\cdot-1,j}}(h)\leqslant C_2\mathrm{e}^{-C_4h^\tau}$ for all $h\in\N$.
\end{enumerate}
\end{lem}
\begin{proof}[\sc{Proof of Lemma \ref{lem:OutcomeAndScoreProcessesUniformlyGMC}}]
Fix $r\in[1,\infty)$ and $i\in[p]$. Under Assumptions
\ref{assu:Innovations}.\ref{enu:EpsIID}, \ref{assu:Innovations}.\ref{enu:EpsSubWeibullBeta},
\ref{assu:Companion}.\ref{enu:SpectralRadius} and
\ref{assu:Companion}.\ref{enu:RowNormDecay}
the $r$\textsuperscript{th} absolute moments of both $Z_{t-1,i}$ and $\varepsilon_{t,i}Z_{t-1,j}$ are finite (cf. Lemma \ref{lem:ZandEpsZNormBnds}, Parts \ref{enu:ZiUnifMomentBnd} and \ref{enu:EpsiZjUnifMomentBnd}). Thus, $\Delta_{r}^{Z_{\cdot,i}}(h)$ and $\Delta_{r}^{\varepsilon_{\cdot,i}Z_{\cdot-1,j}}(h)$
are well defined for all $h\in\N$. Defining
\[
\bG\left(\dotsc,\bseta_{t-1},\bseta_{t}\right):=\sum_{\ell=0}^{\infty}(\widetilde{\bTheta}_{0}^{\ell})_{\cdot,1:p}\bepsilon_{t-\ell}=\sum_{\ell=0}^{\infty}(\widetilde{\bTheta}_{0}^{\ell})_{\cdot,1:p}\bF\left(\bseta_{t-\overline{q}-\ell},\dotsc,\bseta_{t-1-\ell},\bseta_{t-\ell}\right),
\]
with $\bF:\R^{\overline{p}\cdot\overline{q}}\to\R^{p}$ and $\{\bseta_t\}_{t\in\Z}$ provided by Assumption \ref{assu:Innovations}.\ref{enu:EpsIID}, one has the causal representation $Z_{t,i}=G_{i}(\dotsc,\bseta_{t-1},\bseta_{t})$ for~$i\in[pq]$, cf. (\ref{eq:VMAInfRepresentation}). To establish Claim \ref{enu:ZiUnifGMC}, let $\{\bseta_t^{\ast}\}_{t\in\Z}$ be an independent copy of $\{\bseta_t\}_{t\in\Z}$ and let $\bepsilon_{t}^{\ast}:=\bF\del[1]{\bseta_{t-\overline{q}}^{\ast},\dotsc,\bseta_{t-1}^{\ast},\bseta_{t}^{\ast}}$. First, by Lemma \ref{lem:norm_decay} there are constants $\overline{C}_1,\overline{C}_2\in(0,\infty)$, depending only on~$b_1,b_2,q$ and~$\tau$, such that for any $m\in\N_0$
\begin{align}
& \enVert[4]{\sum_{\ell=m}^{\infty}(\widetilde{\bTheta}_{0}^{\ell})_{i,1:p}\left(\bepsilon_{t-\ell}-\bepsilon_{t-\ell}^{\ast}\right)}_r \notag\\
& \leqslant\sum_{\ell=m}^{\infty}\big\Vert (\widetilde{\bTheta}_{0}^{\ell})_{i,1:p}\left(\bepsilon_{t-\ell}-\bepsilon_{t-\ell}^{\ast}\right)\big\Vert_r\tag{countable sub-additivity}\notag\\
& \leqslant\sum_{\ell=m}^{\infty}r^{1/\alpha}\big\Vert (\widetilde{\bTheta}_{0}^{\ell})_{i,1:p}\left(\bepsilon_{t-\ell}-\bepsilon_{t-\ell}^{\ast}\right)\big\Vert _{\psi_{\alpha}}\tag{definition of sub-Weibull norm}\notag\\
& \leqslant r^{1/\alpha}\left\Vert \bepsilon_{0}-\bepsilon_{0}^{\ast}\right\Vert _{\psi_{\alpha}}\sum_{\ell=m}^{\infty}\big\Vert (\widetilde{\bTheta}_{0}^{\ell})_{i,1:p}\big\Vert _{\ell_{2}}\tag{joint sub-Weibullness and stationarity}\notag\\
& \leqslant 2a_{1}r^{1/\alpha}\overline{C}_1\sum_{\ell=m}^{\infty} \mathrm{e}^{-\overline{C}_2\ell^\tau}\tag{sub-additivity of~$\enVert[0]{\cdot}_{\psi_\alpha}$ and Lemma \ref{lem:norm_decay}}\notag\\
& \leqslant 2a_{1}r^{1/\alpha}\overline{C}_1 \sum_{\ell=0}^{\infty}\mathrm{e}^{-\overline{C}_2(\ell+m)^\tau} \notag  \\
& \leqslant 2a_{1}r^{1/\alpha}\overline{C}_1 \sum_{\ell=0}^{\infty}\mathrm{e}^{-(\overline{C}_2/2)(\ell^\tau+m^\tau)} \tag{$(a+b)^\tau \geqslant (a^\tau + b^\tau)/2, \quad a,b\geqslant0,\quad \tau\in(0,1]$}\\
&\leqslant\underbrace{2a_{1}r^{1/\alpha}\overline{C}_1\del[3]{1 + \frac{\Gamma(1/\tau)}{(\overline{C}_2/2)^{1/\tau}\tau}}  }_{=:\overline{C}_3}\mathrm{e}^{-(\overline{C}_2/2)m^\tau}\notag \\
&\leqslant \overline{C}_3 \mathrm{e}^{-C_3 m^\tau}\label{eq:auxrnorm}, 
\end{align}
where the penultimate inequality follows as in \eqref{eq:bound_integral},
and we abbreviate $C_3 := \overline{C}_2/[2+(2\overline q)^\tau]$. (The additional $(2\overline q)^\tau$ in the denominator will be convenient for later consolidation of cases.) Fix $h\in\N$. We consider the cases (i) $\overline q=0$ and (ii) $\overline q\in\N$ in turn.

\textit{Case (i):} If $\overline q=0$, then the innovations [$\bepsilon_t=\bF(\bseta_t)$] are i.i.d., and \eqref{eq:auxrnorm} with $m=h$ shows that
\[
\Delta_{r}^{Z_{\cdot,i}}(h)
=\enVert[4]{\sum_{\ell=h}^{\infty}(\widetilde{\bTheta}_{0}^{\ell})_{i,1:p}\left(\bepsilon_{t-\ell}-\bepsilon_{t-\ell}^{\ast}\right)}_r
\leqslant \overline{C}_3\mathrm{e}^{-C_3h^\tau}.
\]
This bound has the desired form, is uniform over $i\in[pq]$, and involves constants with the claimed dependencies.

\textit{Case (ii):} If $\overline q\in\N$, then we further split into the cases (ii.a) $h\leqslant \overline q$ and (ii.b) $h>\overline q$.

\textit{Case (ii.a):} If $h\leqslant\overline{q}$, then
\[
    \Delta_{r}^{Z_{\cdot,i}}\left(h\right) \leqslant \enVert[4]{\sum_{\ell=\overline{q}}^{\infty}(\widetilde{\bTheta}_{0}^{\ell})_{i,1:p}\del[1]{\bepsilon_{t-\ell}-\bepsilon_{t-\ell}^{\ast}}}_r + \enVert[4]{\sum_{\ell=0}^{\overline{q}-1}(\widetilde{\bTheta}_{0}^{\ell})_{i,1:p}\del[1]{\bepsilon_{t-\ell}-\bepsilon_{t-\ell}^{\dagger}}}_r,
\]
where we use the shorthand
\[
\bepsilon_{t-\ell}^{\dagger}:=
\begin{cases}
\bF\del[1]{\bseta_{t-\overline{q}-\ell}^{\ast},\dotsc,\bseta_{t-\ell}^{\ast}},\quad &\text{if}\;h\leqslant \ell\\
\bF\del[1]{\bseta_{t-\overline{q}-\ell}^{\ast}, \dotsc,\bseta_{t-h}^{\ast},\bseta_{t-h+1}\dotsc,\bseta_{t-\ell}}, &\text{if}\;h> \ell.
\end{cases}
\]
By arguments similar to those yielding \eqref{eq:auxrnorm}, we arrive at 
\begin{align*}
\enVert[4]{\sum_{\ell=0}^{\overline{q}-1}(\widetilde{\bTheta}_{0}^{\ell})_{i,1:p}\left(\bepsilon_{t-\ell}-\bepsilon_{t-\ell}^{\dagger}\right)}_r
&\leqslant 2a_{1}r^{1/\alpha}\overline{C}_1\sum_{\ell=0}^{\overline{q}-1} \mathrm{e}^{-\overline{C}_2\ell^\tau}\\
&\leqslant \del[4]{\underbrace{2a_{1}r^{1/\alpha}\overline{C}_1\mathrm{e}^{C_3\overline{q}^\tau}\sum_{\ell=0}^{\overline{q}-1} \mathrm{e}^{-\overline{C}_2\ell^\tau} }_{=:\overline{C}_4}}\mathrm{e}^{-C_3h^\tau}.\tag{$h\leqslant \overline q$}
\end{align*}
In addition, from~\eqref{eq:auxrnorm} with~$m=\overline{q}$, and again using~$h\leqslant \overline{q}$, we get
\begin{align*}
    \enVert[4]{\sum_{\ell=\overline{q}}^{\infty}(\widetilde{\bTheta}_{0}^{\ell})_{i,1:p}\left(\bepsilon_{t-\ell}-\bepsilon_{t-\ell}^{\ast}\right)}_r \leqslant \overline{C}_3\mathrm{e}^{-C_3\overline{q}^\tau} \leqslant \overline{C}_3\mathrm{e}^{-C_3h^\tau}. 
\end{align*}
We conclude that for $h\leqslant\overline{q}$,
$
    \Delta_{r}^{Z_{\cdot,i}}(h) \leqslant  (\overline{C}_3+\overline{C}_4) \mathrm{e}^{-C_3h^\tau}.  
$

\textit{Case (ii.b):} If $h > \overline{q}$, we bound the physical dependence measure as follows
\begin{align*}
    \Delta_{r}^{Z_{\cdot,i}}(h)  
    &\leqslant\enVert[4]{\sum_{\ell=h}^{\infty}(\widetilde{\bTheta}_{0}^{\ell})_{i,1:p}\del[1]{\bepsilon_{t-\ell}-\bepsilon_{t-\ell}^{\ast}}}_r  
     + \enVert[4]{\sum_{\ell=h - \overline{q}}^{h-1}(\widetilde{\bTheta}_{0}^{\ell})_{i,1:p}\del[1]{\bepsilon_{t-\ell}-\bepsilon_{t-\ell}^{\dagger}}}_r
\end{align*}
and consider each right-hand side term in turn. From \eqref{eq:auxrnorm} with $m=h$, we see that
\[
\enVert[4]{\sum_{\ell=h}^{\infty}(\widetilde{\bTheta}_{0}^{\ell})_{i,1:p}\del[1]{\bepsilon_{t-\ell}-\bepsilon_{t-\ell}^{\ast}}}_r
\leqslant \overline{C}_3\mathrm{e}^{-C_3 h^\tau}.
\]
And by arguments similar to those yielding \eqref{eq:auxrnorm}, we arrive at
\begin{align*}
\enVert[4]{\sum_{\ell=h - \overline{q}}^{h-1}(\widetilde{\bTheta}_{0}^{\ell})_{i,1:p}\del[1]{\bepsilon_{t-\ell}-\bepsilon_{t-\ell}^{\dagger}}}_r
& \leqslant 2a_{1}r^{1/\alpha}\overline{C}_1\sum_{\ell=h-\overline{q}}^{h-1} \mathrm{e}^{-\overline{C}_2\ell^\tau} \\
& \leqslant 2a_{1}r^{1/\alpha}\overline{C}_1\sum_{\ell=0}^{\overline{q}-1} \mathrm{e}^{-\overline{C}_2(\ell+h-\overline{q})^\tau}\tag{change of variable}\\
& \leqslant2a_{1}\overline{q}r^{1/\alpha}\overline{C}_1 \mathrm{e}^{-\overline{C}_2(h-\overline{q})^\tau}\tag{monotonicity}\\
& \leqslant 2a_1\overline{q}r^{1/\alpha}\overline{C}_1\mathrm{e}^{-[\overline{C}_2/(\overline{q}+1)^\tau]h^\tau} \\
& \leqslant 2a_1\overline{q}r^{1/\alpha}\overline{C}_1\mathrm{e}^{-[\overline{C}_2/(2\overline{q})^\tau]h^\tau}\\
 & \leqslant 2a_1\overline{q}r^{1/\alpha}\overline{C}_1 \mathrm{e}^{-C_3 h^\tau},\tag{$C_3= \overline{C}_2/[2+(2\overline{q})^\tau]$}
\end{align*}
where we use that $x\geqslant y\geqslant 1$ implies $x-(y-1)\geqslant x/y$ (established as part of the proof of Lemma \ref{lem:norm_decay}) for $x=h$ and $y=\overline q + 1$ to deduce that $h-\overline q\geqslant h/(\overline q + 1)$. Since $\overline q\geqslant 1$, the latter term is further bounded from below by $h/(2\overline q)$. Gathering terms, we see that
\begin{align*}
    \Delta_{r}^{Z_{\cdot,i}}(h)  
    &\leqslant \underbrace{\left(\overline{C}_3 + 2a_1\overline{q}r^{1/\alpha}\overline{C}_1 \right)}_{=:\overline{C}_5} \mathrm{e}^{-C_3h^\tau},
\end{align*}
Consolidating the different cases, we conclude that Claim \ref{enu:ZiUnifGMC} holds with $C_1 := \max\{\overline{C}_3+\overline{C}_4,\overline{C}_5\}$, which has the claimed dependencies.

To establish Claim \ref{enu:EpsiZjUnifGMC}, fix~$i\in[p]$ and $j\in[pq]$. Noting that that $\varepsilon_{t,i}=F_i(\bseta_{t-\overline{q}},\dotsc,\bseta_{t})$ and $Z_{t-1,j}=G_j(\dotsc,\bseta_{t-2},\bseta_{t-1})$, we have that the joint process $\{(\varepsilon_{t,i},Z_{t-1,j})\}$ is strictly stationary and causal. Since each $\{\varepsilon_{t,i}\}_{t\in\Z},i\in[p]$, is $\overline{q}$-dependent, we have $\Delta_{r}^{\varepsilon_{\cdot,i}}(h)=0$ for $h>\overline q$. For $h\leqslant \overline q$, by now familiar calculations, we see that
\begin{align*}
\Delta_{r}^{\varepsilon_{\cdot,i}}(h)=\enVert[1]{\varepsilon_{t,i}-\varepsilon_{t,i}^\dagger}_r
&\leqslant 2\enVert[0]{\varepsilon_{0,i}}_r
\leqslant 2 r^{1/\alpha} \enVert[0]{\bepsilon_{0}}_{\psi_\alpha}
\leqslant 2 r^{1/\alpha} a_1 \\
&\leqslant \underbrace{2 r^{1/\alpha} a_1 \mathrm{e}^{C_3 \overline q^\tau}}_{=:\overline{C}_6} \mathrm{e}^{-C_3 h^\tau}.\tag{$h\leqslant \overline q$}
\end{align*}
Conclude that $\max_{i\in[p]}\Delta_{r}^{\varepsilon_{\cdot,i}}(h)\leqslant \overline{C}_6\mathrm{e}^{-C_3 h^\tau}$ for all $h\in\N$.
Turning to the product processes, for any $h\in\N$ we have
\begin{align*}
    \Delta_{r}^{\varepsilon_{\cdot,i}Z_{\cdot-1,j}}\left(h\right) & \leqslant \left\Vert \varepsilon_{0,i} \right\Vert_{2r} \Delta_{2r}^{Z_{\cdot-1,j}}(h) + \left\Vert Z_{0,j}  \right\Vert_{2r}\Delta_{2r}^{\varepsilon_{\cdot,i}}(h)  \tag{Lemma \ref{lem:Delta_r}, Part \ref{enu:Delta_rFromPairToProduct}}\\
    & \leqslant \left\Vert \varepsilon_{0,i} \right\Vert_{2r}\Delta_{2r}^{Z_{\cdot,j}}(h-1) + \left\Vert Z_{0,j}  \right\Vert_{2r}\Delta_{2r}^{\varepsilon_{\cdot,i}}(h) \\
    & \leqslant \left\Vert \varepsilon_{0,i} \right\Vert_{2r}C_1'\mathrm{e}^{-C_3(h-1)^\tau}+\left\Vert Z_{0,j}  \right\Vert_{2r}\overline{C}_6'\mathrm{e}^{-C_3h^{\tau}}\tag{Claim  \ref{enu:ZiUnifGMC} } \\
    &\leqslant \del[2]{\underbrace{ (2r)^{1/\alpha}a_1 C_1' + C(a_1,b_1,b_2,\tau)(2r)^{1/\alpha}\overline{C}_6'  }_{=:\overline{C}_7}} \mathrm{e}^{-C_3(h-1)^\tau},
\end{align*}
where the final inequality follows from Assumption~\ref{assu:Innovations}.\ref{enu:EpsSubWeibullBeta} and Part \ref{enu:ZiUnifMomentBnd} of Lemma \ref{lem:ZandEpsZNormBnds}.
In the previous display, the constant $C(a_1,b_1,b_2,\tau)$ stems from Part \ref{enu:ZiUnifMomentBnd} of Lemma \ref{lem:ZandEpsZNormBnds}, and $C_1'$ and $\overline{C}_6'$ refer to the constants $C_1$ and $\overline{C}_6$, respectively, with $r$ replaced by $2r$. (Note that $C_3=\overline{C}_2/[2+(2\overline q)^\tau]$ does not actually depend on $r$.)
For $h=1$, the final bound in the previous display shows $\Delta_{r}^{\varepsilon_{\cdot,i}Z_{\cdot-1,j}}(h)\leqslant \overline{C}_7\mathrm{e}^{C_3}\mathrm{e}^{-C_3h^{\tau}}$, while for $h\geqslant 2$, we use $h-1\geqslant h/2$ to further bound from above as follows $\Delta_{r}^{\varepsilon_{\cdot,i}Z_{\cdot-1,j}}(h) \leqslant \overline{C}_7  \mathrm{e}^{-(C_3/2^{\tau})h^\tau}$. Choosing $C_2:= \overline{C}_7\mathrm{e}^{C_3}$ and $C_4:= C_3/2^\tau$, we conclude that $\max_{(i,j)\in[p]\times[pq]}\Delta_{r}^{\varepsilon_{\cdot,i}Z_{\cdot-1,j}}(h)\leqslant C_2\mathrm{e}^{-C_4 h^\tau}$ for all $h\in\N$. As these constants have the claimed dependencies, Claim \ref{enu:EpsiZjUnifGMC} follows.
\end{proof}

Based on these supporting lemmas, we are now able to give the proof of Lemma \ref{lem:Deviation-Bound}.
\begin{proof}
[\textsc{Proof of Lemma \ref{lem:Deviation-Bound}}] We set up for
an application of Lemma \ref{lem:TailBoundGMCVector}. Consider the
stochastic processes $\{\{\varepsilon_{t,i}Z_{t-1,j}\}_{t\in\Z}\},(i,j)\in[p]\times[pq]$, of which there are a total of $p^2q(=d_n)$. Since $\{\bZ_t\}_{t\in\Z}$ admits the VMA$(\infty)$ representation in (\ref{eq:VMAInfRepresentation}), each process $\{\{\varepsilon_{t,i}Z_{t-1,j}\}_{t\in\Z}\},(i,j)\in[p]\times[pq]$, is both strictly stationary and causal with elements taking values in $\R$.
Setting $r=4$, Part \ref{enu:EpsiZjUnifMomentBnd} of Lemma \ref{lem:ZandEpsZNormBnds} and Part \ref{enu:EpsiZjUnifGMC} of Lemma \ref{lem:OutcomeAndScoreProcessesUniformlyGMC} respectively show that there are constants $C,C',C''\in(0,\infty)$, depending only on $a_1,\alpha,b_1,b_2,q,\overline{q}$ and $\tau$, such that
both $\max_{(i,j)\in[p]\times[pq]}\Vert \varepsilon_{0,i}Z_{-1,j}\Vert_{4}\leqslant C$
and $\max_{i\in[p],j\in[pq]}\Delta_{4}^{\varepsilon_{\cdot,i}Z_{\cdot-1,j}}(h)\leqslant C'\mathrm{e}^{-C''h^\tau}$ for all $h\in\N$. By Assumption \ref{assu:Innovations}.\ref{enu:EigValsOfSigmaEpsBddAwayFromZero} the processes $\{\{\varepsilon_{t,i}Z_{t-1,j}\}_{t\in\Z}\},(i,j)\in[p]\times[pq]$ have mean zero.
It follows that Assumptions \ref{assu:MomentConditions} and
\ref{assu:GMC} hold (with $b_{1}$ and $b_2$ there equal to the constants $C'$ and $C''$ respectively). For
Assumption \ref{assu:UniformNondegeneracy}, note that for any $i\in[p]$, $j\in[pq]$, $k\in\N_0$ and $h\in\N$,
\begin{align*}
\E\left[\del[4]{\sum_{t=k+1}^{k+h}\varepsilon_{t,i}Z_{t-1,j}}^{2}\right] & =\sum_{t=k+1}^{k+h}\E\sbr[1]{\varepsilon_{t,i}^{2}Z_{t-1,j}^{2}}\tag{$\E[\varepsilon_{t}\mid\mathcal{F}_{t-1}^{\bseta}]=0$ a.s.~and $Z_{t-1} \in \mathcal{F}_{t-1}^{\bseta}$}\\
 & = h \E\sbr[1]{\varepsilon_{0,i}^{2}Z_{-1,j}^{2}}\tag{stationarity}\\
 & = h\E\sbr[2]{\E\sbr[1]{\varepsilon_{0,i}^{2}\mid\mathcal{F}_{-1}^{\bseta}}Z_{-1,j}^{2}}\tag{$Z_{-1} \in \mathcal{F}_{-1}^{\bseta}$}\\
& \geqslant h\E\sbr[1]{a_2 Z_{-1,j}^{2}} \tag{Assumption \ref{assu:Innovations}.\ref{enu:EigValsOfSigmaEpsBddAwayFromZero}}\\
 & \geqslant a_{2} h \Lambda_{\min}(\bSigma_{\bZ})\\
 & \geqslant a_{2} d h, \tag{Assumption \ref{assu:CovarianceZ}}
\end{align*}
implying that Assumption \ref{assu:UniformNondegeneracy} is satisfied for the constant $\omega= (a_{2}d)^{1/2}\in(0,\infty)$.
The derived constants depend on neither $i$ nor $j$. Choosing the sequence $\gamma_n$ in Assumption \ref{assu:Growth} to be the non-random sequence $\gamma_n$ in (\ref{eq:tuning_c_and_gamma}), we have both $\ln(1/\gamma_n)\lesssim\ln(pqn)$ and $\gamma_n\to0$. Since $\gamma_n$ lies in $(0,1)$ for large enough $n$, Lemma \ref{lem:TailBoundGMCVector} now shows that there is a constant $N\in\N$ such that
\[
n\geqslant N\implies \P\del{\max_{(i,j)\in[p]\times[pq]}\frac{\envert{\sum_{t=1}^n\varepsilon_{t,i}Z_{t-1,j}}}{\sqrt{\sum_{k=1}^l\del{\sum_{t\in H_{n,k}}\varepsilon_{t,i}Z_{t-1,j}}^2}}\geqslant \Phi^{-1}\del{1-\frac{\gamma_n}{2p^2q}}}\leqslant 2\gamma_n,
\]
Scaling both sides of the inequality inside the probability by $2c\sqrt n=2c(1/n)/(1/\sqrt n)$, the previous display becomes
\[
n\geqslant N\implies \P\del{c\max_{i\in[p]}\enVert[0]{\bS_{n,i}}_{\ell_\infty}\geqslant \frac{\lambda_n^\ast}{n}}\leqslant 2\gamma_n.
\]
from which the asserted claim follows.
\end{proof}

\subsection{Proof of Ideal Penalty Loadings Control (Lemma \ref{lem:Ideal-Penalty-Loading-Control})}
The proof of Lemma \ref{lem:Ideal-Penalty-Loading-Control} essentially follows from the following result.
\begin{lem}\label{lem:DeltaIdealVanishes}
Let Assumptions \ref{assu:Innovations}.\ref{enu:EpsIID}, \ref{assu:Innovations}.\ref{enu:EpsSubWeibullBeta}, \ref{assu:Companion}.\ref{enu:SpectralRadius} and
\ref{assu:Companion}.\ref{enu:RowNormDecay} hold and suppose that $(\ln p)^{C(\alpha,\tau)}=o(n)$ with $C(\alpha,\tau)$ given in Assumption \ref{assu:RowSparsity} for $\alpha$ provided by Assumption \ref{assu:Innovations}.\ref{enu:EpsSubWeibullBeta} and $\tau\in(\textstyle{\frac{1}{2}},1]$ provided by Assumption \ref{assu:Companion}.\ref{enu:RowNormDecay}. Let $\upsilon_{i,j}^{0}:=(\E[\varepsilon_{0,i}^{2}Z_{-1,j}^{2}])^{1/2}$ for $i\in[p]$ and $j\in[pq]$, and define
\begin{equation}\label{eq:DeltaIdealDefn}
\widehat\Delta^0:=\max_{(i,j)\in[p]\times[pq]}\envert[1]{(\widehat{\upsilon}_{i,j}^{0})^2-(\upsilon_{i,j}^{0})^2},
\end{equation}
where $\widehat{\upsilon}_{i,j}^{0}$ is given in \eqref{eq:PenaltyLoadingsIdeal_p2}. Then $\widehat\Delta^0\to_\P0$.
\end{lem}
\begin{proof}[\sc{Proof of Lemma \ref{lem:DeltaIdealVanishes}}]
We set up for an application of Lemma \ref{lem:MaximalInequalityFunctionalDependence}. Observe first that
\begin{align*}
\widehat\Delta^0 &= \max_{(i,j)\in[p]\times[pq]}\envert[4]{\frac{1}{n}\sum_{k=1}^{l_n}\del[4]{\sum_{t\in H_{n,k}}\varepsilon_{t,i}Z_{t-1,j}}^{2}-\E\sbr[1]{\varepsilon_{0,i}^{2}Z_{-1,j}^{2}}}\\
 & =m_n\max_{(i,j)\in[p]\times[pq]}\envert[4]{\frac{1}{l_n}\sum_{k=1}^{l_n}\sbr[4]{\frac{1}{m_n^{2}}\del[4]{\sum_{t\in H_{n,k}}\varepsilon_{t,i}Z_{t-1,j}}^{2}-\frac{1}{m_n}\E\sbr[1]{\varepsilon_{0,i}^{2}Z_{-1,j}^{2}}}}\\
 & =\frac{m_n}{l_n}\max_{(i,j)\in[p]\times[pq]}\envert[4]{\sum_{k=1}^{l_n}\left(m_n^{-2}(V_{k,i,j}^{(n)})^{2}-\E\sbr[1]{m_n^{-2}(V_{0,i,j}^{(n)})^{2}}\right)},
\end{align*}
where we have defined
\begin{equation}
V_{k,i,j}^{(n)}:=\sum_{t\in H_{n,k}}\varepsilon_{t,i}Z_{t-1,j}\label{eq:def_V_k_i_j} ,\quad i\in[p],\quad j\in[pq],\quad k\in[l_n].
\end{equation}
Setting up for an application of Lemma \ref{lem:Delta_r}, denote
\begin{equation}
 Y_{t,i,j}:= \sum_{\ell=0}^{m_n-1}\varepsilon_{t-\ell,i}Z_{t-1-\ell,j}, ,\quad i\in[p],\quad j\in[pq].  \label{eq:def_Y_t_i_j}
\end{equation}
Arguing as in the proof of Lemma \ref{lem:OutcomeAndScoreProcessesUniformlyGMC}, we get the causal representations $\bepsilon_{t}=\bF(\bseta_{t-\overline{q}},\dotsc,\bseta_{t})$ and $\bZ_{t-1}=\bG(\dotsc,\bseta_{t-2},\bseta_{t-1})$, respectively. Hence, the joint process $\{(\bepsilon_t^\top,\bZ_{t-1}^\top)^\top\}_{t\in\Z}$ is strictly stationary and causal. The right-hand side of \eqref{eq:def_Y_t_i_j} forms an $\R$-valued mapping $\overline{G}_{i,j}^{(n)}$ for which $Y_{t,i,j}=\overline{G}_{i,j}^{(n)}(\dotsc,\bseta_{t-1},\bseta_{t} )$, so the process $\{Y_{t,i,j} \}_{t\in\Z}$ is strictly stationary and causal. Gathering the mappings $\{\overline{G}_{i,j}^{(n)} : i\in [p], j\in [pq] \}$ into an $\R^{p\times pq}$-valued mapping $\overline{\bG}^{(n)}$, such that $\{Y_{t,i,j} : i\in [p],j\in [pq] \} = \overline{\bG}^{(n)}(\dotsc,\bseta_{t-1},\bseta_{t})$, we see that $\{Y_{t,i,j} : i\in [p],j\in [pq] \}_{t\in\Z}$ is strictly stationary and causal. Noting that 
\begin{equation}\label{eq:def_Y_n_t_i_j}    
V_{k,i,j}^{(n)}=Y_{km_n,i,j}, \quad k \in \Z,
\end{equation}
we have that each skeleton process $\{V_{k,i,j}^{(n)}\}_{k\in\Z},i\in[p],j\in[pq]$, is strictly stationary and causal by Part \ref{enu:Delta_r_skeleton} of Lemma \ref{lem:Delta_r}. As each of these processes arise from the same i.i.d.~sequence $\{\bseta_t\}_{t\in\Z}$, we conclude that their joint process $\{V_{k,i,j}^{(n)}:i\in[p],j\in[pq]\}_{k\in\Z}$ is strictly stationary and causal as well.

Next, for any $i\in[p],j\in[pq],r\in[1,\infty)$ and $k\in[l_n]$,
\begin{align*}
\enVert[1]{m_n^{-2}(V^{(n)}_{k,i,j})^{2}}_{r} &=m_n^{-2}\enVert[4]{\del[4]{\sum_{t\in H_{n,k}}\varepsilon_{t,i}Z_{t-1,j}}^{2}}_{r}
 =m_n^{-2}\left\Vert \sum_{t=1}^{m_n}\varepsilon_{t,i}Z_{t-1,j}\right\Vert _{2r}^{2}\\
 & \leqslant\left\Vert \varepsilon_{0,i}Z_{-1,j}\right\Vert _{2r}^{2}\tag{sub-additivity and stationarity}\\
 & \leqslant C r^{4/\alpha}\tag{Lemma \ref{lem:ZandEpsZNormBnds}, Part \ref{enu:EpsiZjUnifMomentBnd}}
\end{align*}
for some constant $C\in(0,\infty)$ depending only on $a_1,\alpha,b_1,b_2$ and $\tau$. The previous display implies that
\[
\max_{(i,j)\in[p]\times[pq]}\left\Vert m_n^{-2}(V^{(n)}_{k,i,j})^{2}\right\Vert_{\psi_{\alpha/4}}\leqslant C,
\]
showing that the sub-Weibullness requirement of Lemma \ref{lem:MaximalInequalityFunctionalDependence} is satisfied with the constants $(K,\xi)$ there equal to $(C,\alpha/4)$. 
Invoking Lemma \ref{lem:MaximalInequalityFunctionalDependence} for the vectorization of $\{m_n^{-2}(V^{(n)}_{k,i,j})^2:i\in[p],j\in[pq]\}_{k\in\Z}$ with $r=2$,~$T=l_n$,~$d_T=p^2q$ and $m_T=1$ (so that $l_T=l_n$), we see that, as $n\to\infty$, 
\begin{align*}
 \widehat\Delta^0&=\frac{m_n}{l_n}\max_{(i,j)\in[p]\times[pq]}\envert[4]{\sum_{k=1}^{l_n}\left(m_n^{-2}(V^{(n)}_{k,i,j})^{2}-\E\sbr[1]{m_n^{-2}(V_{0,i,j}^{(n)})^{2}}\right)}\\
 & \lesssim_{\P}\frac{m_n}{l_n}\left(l_n\left(p^{2}q\right)^{1/2}\Delta_{2,n}+\sqrt{l_n\ln\left(p^{2}ql_n\right)}+\left(\ln\left(p^{2}ql_n\right)\right)^{\widetilde{C}(4,\alpha)}\right)\\
 & =m_n\left(p^{2}q\right)^{1/2}\Delta_{2,n}+\frac{m_n}{l_n}\sqrt{l_n\ln\left(p^{2}ql_n\right)}+\frac{m_n}{l_n}\left(\ln\left(p^{2}ql_n\right)\right)^{\widetilde{C}(4,\alpha)}\\
 & =:\mathrm{I}_n+\mathrm{II}_n+\mathrm{III}_n,
\end{align*}
with $\widetilde{C}(4,\alpha) = \max\{8/\alpha, (4+\alpha)/\alpha\} = 1/(\alpha/4) + 1/\min\{\alpha/4,1\}$ and 
\[
\Delta_{2,n}:=\max_{(i,j)\in[p]\times[pq]}\Delta_{2}^{m_n^{-2}(V^{(n)}_{\cdot,i,j})^{2}}\left(2\right).
\]
We handle each of the right-hand side terms in turn, starting with $\mathrm{II}_n$ and $\mathrm{III}_n$.
Recalling that $m_n=n^{1/(1+4\tau)}$, $l_n=n^{4\tau/(1+4\tau)}$, and $(\ln p)^{(1+4\tau)/(4\tau-2)}=o(n)$ with $\tau\in(\textstyle{\frac{1}{2}},1]$, for some constant $C\in(0,\infty)$, depending only on $q$, we get 
\[
\mathrm{II}_n\leqslant Cn^{\frac{1-4\tau}{1+4\tau}}\sbr{n^{\frac{4\tau}{1+4\tau}}(\ln p+\ln n)}^{1/2}=C\sqrt{\frac{\ln p+\ln n}{n^{(4\tau-2)/(1+4\tau)}}}=o(1).
\]
Recalling also that $(\ln p)^{\widetilde{C}(4,\alpha)(1+4\tau)/(4\tau-1)}=o(n)$, for some constant $C'\in(0,\infty)$ depending only on $\alpha$ and $q$, we see that
\[
\mathrm{III}_n\leqslant C'n^{\frac{1-4\tau}{1+4\tau}}\left(\ln p + \ln n\right)^{\widetilde{C}(4,\alpha)}=C'\left(\frac{\ln p+\ln n}{n^{(4\tau-1)/\sbr[0]{\widetilde{C}(4,\alpha)(1+4\tau)}}}\right)^{\widetilde{C}(4,\alpha)}=o(1).
\]
Since the lag-order $q$ is constant (Assumption \ref{assu:Companion}.\ref{enu:lag-order}),  $\mathrm{I}_n=o(1)$ is equivalent to $m_n p\Delta_{2,n}=o(1)$. To this end, with $Y_{t,i,k} $ given in \eqref{eq:def_Y_t_i_j}, note that for $h > m_n$ and any $r\in[1,\infty)$, we have
\begin{align*}
    \Delta_{r}^{Y_{\cdot,i,j}}(h) & \leqslant \sum_{\ell =0}^{m_n-1}\Delta_{r}^{\varepsilon_{\cdot,i}Z_{\cdot-1,j}}\left(h - \ell\right) \tag{Lemma \ref{lem:Delta_r}, Part \ref{enu:Delta_r_sum}} \\
    & \leqslant \sum_{\ell =0}^{m_n-1} C_2\mathrm{e}^{-C_4(h-\ell)^{\tau}} \tag{Lemma \ref{lem:OutcomeAndScoreProcessesUniformlyGMC}, Part \ref{enu:EpsiZjUnifGMC}}\\
    & \leqslant C_2 m_n \mathrm{e}^{-C_4(h-m_n)^{\tau}},\tag{monotonicity}
\end{align*}
for constants $C_2,C_4\in(0,\infty)$,
depending only on $a_1,\alpha,b_1,b_2,q,\overline{q},\tau$ and $r$.
Hence, using \eqref{eq:def_Y_n_t_i_j}, Part \ref{enu:Delta_r_skeleton} of Lemma \ref{lem:Delta_r} shows that for any $r\in[1,\infty)$ both
\begin{equation}
 \Delta_{r}^{V^{(n)}_{\cdot,i,j}}(2) = \Delta_{r}^{Y_{\cdot,i,j}}(2m_n) \leqslant C_2 m_n \mathrm{e}^{-C_4m_n^{\tau}}\label{eq:Delta_Y_n_t_i_j_bound}
\end{equation}
and, for some constant $C\in(0,\infty)$ depending only on $a_1,\alpha,b_1,b_2$ and $\tau$,
\begin{align*}
 \Delta_{r}^{(V^{(n)}_{\cdot,i,j})^2}(2) & \leqslant 2\enVert[1]{V_{0,i,j}^{(n)}}_{2r}\Delta_{2r}^{V^{(n)}_{\cdot,i,j}}(2) \tag{Lemma \ref{lem:Delta_r}, Part \ref{enu:Delta_rFromPairToProduct}} \\
 & \leqslant  2m_n \left\Vert \varepsilon_{0,i}Z_{-1,j} \right\Vert_{2r} \Delta_{2r}^{V^{(n)}_{\cdot,i,j}}(2) \tag{sub-additivity and stationarity} \\
 & \leqslant 2 m_n Cr^{2/\alpha} \Delta_{2r}^{V^{(n)}_{\cdot,i,j}}(2)  \tag{Lemma \ref{lem:ZandEpsZNormBnds}, Part \ref{enu:EpsiZjUnifMomentBnd}} \\
 & \leqslant 2 Cr^{2/\alpha} C_2 m_n^2 \mathrm{e}^{-C_4m_n^{\tau}} \tag{Equation \eqref{eq:Delta_Y_n_t_i_j_bound}}.
\end{align*}
As the latter bound holds uniformly over $(i,j)\in[p]\times[pq]$, setting $r=2$, we see that for some constant $C''\in(0,\infty)$, depending only on $a_1,\alpha,b_1,b_2,q,\overline{q}$ and $\tau$,
\[
\Delta_{2,n}
= \max_{(i,j)\in[p]\times[pq]}\Delta_{2}^{m_n^{-2}(V^{(n)}_{\cdot,i,j})^{2}}\left(2\right)
= m_n^{-2} \max_{(i,j)\in[p]\times[pq]}\Delta_{2}^{(V^{(n)}_{\cdot,i,j})^{2}}\left(2\right)
\leqslant C''\mathrm{e}^{-C_4m_n^{\tau}}.
\]
The previous display shows that $m_n p\Delta_{2,n}\lesssim m_n p \mathrm{e}^{-C_4m_n^\tau}=\mathrm{e}^{-C_4m_n^\tau +\ln(p)+\ln(m_n)}$. As $m_n=n^{1/(1+4\tau)}$, the right-hand side of this inequality is $o(1)$, provided that $(\ln p)^{(1+4\tau)/\tau}=o(n)$, which holds by hypothesis. We conclude that also $\mathrm{I}_n=o(1)$.
\end{proof}

\begin{proof}
[\sc{Proof of Lemma \ref{lem:Ideal-Penalty-Loading-Control}}] Recall that $\upsilon_{i,j}^{0}=(\E[\varepsilon_{0,i}^{2}Z_{-1,j}^{2}])^{1/2}$ for $i\in[p]$ and $j\in[pq]$. Using Assumptions  \ref{assu:Innovations} and \ref{assu:CovarianceZ}, as in the proof of Lemma  \ref{lem:Deviation-Bound} we have $\min_{i\in[p],j\in[pq]}(\upsilon_{i,j}^{0})^2\geqslant a_2d > 0$, so the $\{(\upsilon_{i,j}^{0})^2\}_{i\in[p],j\in[pq]}$ are bounded away from zero by the constant $a_2d=:C$. Part \ref{enu:EpsiZjUnifMomentBnd} of Lemma \ref{lem:ZandEpsZNormBnds} (with $r=2$) shows that the $\{(\upsilon_{i,j}^{0})^2\}_{i\in[p],j\in[pq]}$ are bounded from above by a constant $C'\in(0,\infty)$ depending only on $a_1,\alpha,b_1,b_2$ and $\tau$. Recalling $\widehat\Delta^0$ in (\ref{eq:DeltaIdealDefn}), it follows that
\begin{align*}
    \min_{i\in[p],j\in[pq]}(\widehat{\upsilon}_{i,j}^{0})^2
    \geqslant 
    \min_{i\in[p],j\in[pq]}(\upsilon_{i,j}^{0})^2-\widehat\Delta^0
    \quad\text{and}\quad
 \max_{i\in[p],j\in[pq]}(\widehat{\upsilon}_{i,j}^{0})^2
    \leqslant 
    \max_{i\in[p],j\in[pq]}(\upsilon_{i,j}^{0})^2+\widehat\Delta^0.
\end{align*}
Lemma \ref{lem:DeltaIdealVanishes} shows that $\widehat\Delta\to_\P0$. Hence, (\ref{eq:IdealPenaltyLoadingsUnderCtrl}) is satisfied with, e.g., $\underline{\upsilon}^0:=\sqrt{C/2}$ and $\overline{\upsilon}^0:=\sqrt{2C'}$.
\end{proof}

\subsection{Proof of Restricted Eigenvalue Bound (Lemma \ref{lem:Restricted-Eigenvalue-Bound})}
We first prove Lemma \ref{lem:Covariance-Consistency} (covariance consistency), and then use it to prove Lemma \ref{lem:Restricted-Eigenvalue-Bound}.
\begin{proof}
[\sc{Proof of Lemma \ref{lem:Covariance-Consistency}}] We set up
for an application of Lemma \ref{lem:MaximalInequalityFunctionalDependence}.
Assumptions \ref{assu:Innovations}.\ref{enu:EpsIID} and \ref{assu:Companion}.\ref{enu:SpectralRadius}
yield the VMA$\ensuremath{(\infty)}$ representation (\ref{eq:VMAInfRepresentation}),
so $\{\bZ_{t}\}_{t\in\mathbb{Z}}$ is strictly stationary and causal. These properties
are inherited by the process $\{\bZ_{t}\bZ_{t}^\top\}_{t\in\Z}$
with elements taking values in $\R^{pq\times pq}$. Part \ref{enu:ZiUnifGMC} of Lemma \ref{lem:ZandEpsZNormBnds} shows that, for $\alpha$ provided by Assumption \ref{assu:Innovations}.\ref{enu:EpsSubWeibullBeta},
\[
\max_{i\in[pq]}\left\Vert Z_{0,i}\right\Vert _{\psi_{\alpha}}\leqslant C_1,
\]
for a constant $C_1$ depending only on $a_1,b_1,b_2$ and $\tau$. Combined with the equivalent definitions of sub-Weibullness in \citet[Theorem 1]{vladimirova_sub-weibull_2020} and the bound in \citet[Proposition D.2]{kuchibhotla_moving_2022}, the previous display shows that for some constant $C_2\in(0,\infty)$, depending only on $a_1,\alpha,b_1,b_2$ and $\tau$,
\[
\max_{(i,j)\in[pq]^2}\left\Vert Z_{0,i}Z_{0,j}\right\Vert_{\psi_{\alpha/2}}\leqslant C_{2}.
\]
Part \ref{enu:ZiUnifMomentBnd} of Lemma \ref{lem:ZandEpsZNormBnds} (with $r=4$) shows that
\[
\max_{i\in[pq]}\left\Vert Z_{0,i}\right\Vert _{4}\leqslant 4^{1/\alpha}C_1 =:C_3,
\]
and Part \ref{enu:ZiUnifGMC} of Lemma \ref{lem:OutcomeAndScoreProcessesUniformlyGMC} (with $r=4$) shows that there is are constants $C_4,C_5\in(0,\infty)$, depending only on $a_1,\alpha,b_1,b_2,q,\overline{q}$ and $\tau$, such that
\[
\max_{i\in[pq]}\Delta_{4}^{Z_{\cdot,i}}\left(h\right)\leqslant C_4\mathrm{e}^{-C_{5}h^\tau},\quad h\in\N.
\]
It follows from Part \ref{enu:Delta_rFromPairToProduct} of Lemma \ref{lem:Delta_r} (with $r=2$) that
\[
\max_{(i,j)\in[pq]^2}\Delta_{2}^{Z_{\cdot,i}Z_{\cdot,j}}\left(h\right)\leqslant2C_{3}C_{4}\mathrm{e}^{-C_{5}h^\tau},\quad h\in\N.
\]
From the previous display and monotonicity, we see that for any choice of $m_n\in[n]$,
\[
\max_{(i,j)\in[pq]^2}\max_{u\in\{m_n+1,\dotsc,2m_n\}}\Delta_{2}^{Z_{\cdot,i}Z_{\cdot,j}}\left(u\right)\leqslant 2C_3C_4\mathrm{e}^{-C_{5}m_n^\tau}.
\]
By hypothesis we have $(\ln p)^{1/\tau+\widetilde{C}(2,\alpha)}=o(n)$, which implies that $(\ln p)^{1/\tau}=o(n)$, so that we  eventually have $\lceil(\ln(pn))^{1/\tau}/C_5^{1/\tau}\rceil\in[n]$.
Applying Lemma \ref{lem:MaximalInequalityFunctionalDependence} for the vectorization of $\{\bZ_t\bZ_t^\top\}_{t\in\Z}$
with $\xi=\alpha/2$, $K=C_{2}$, $r=2$,~$T=n$,~$d_T=p^2q^2$ and $m_T=\lceil(\ln(pn))^{1/\tau}/C_5^{1/\tau}\rceil$, we see that, as $n\to\infty$,
\begin{align*}
&\max_{(i,j)\in[pq]^2}\envert[1]{(\widehat{\bSigma}_{\bZ})_{i,j}-(\bSigma_{\bZ})_{i,j}} \\
& \overset{d}{=}\frac{1}{n}\max_{(i,j)\in[pq]^2}\envert[4]{\sum_{t=1}^{n}\left(Z_{t,i}Z_{t,j}-\E\left[Z_{0,i}Z_{0,j}\right]\right)}\tag{strict stationarity}\\
&\lesssim_\P \frac{1}{n}\Big(n[(pq)^2]^{1/2}\mathrm{e}^{-C_{5}m_T^\tau}+m_T\sqrt{l_T\ln((pq)^2 n)}+m_T(\ln((pq)^2 n))^{\widetilde{C}(2,\alpha)}\Big)\tag{Lemma \ref{lem:MaximalInequalityFunctionalDependence}}\\
&\lesssim \frac{1}{n} + \sqrt{\frac{(\ln(pn))^{1/\tau+1}}{n}} + \frac{(\ln(pn))^{1/\tau+\widetilde{C}(2,\alpha)}}{n}\tag{Assumption \ref{assu:Companion}.\ref{enu:lag-order} and choice of $m_T$}\\
&\lesssim \sqrt{\frac{(\ln(pn))^{1/\tau+\widetilde{C}(2,\alpha)}}{n}} + \frac{(\ln(pn))^{1/\tau+\widetilde{C}(2,\alpha)}}{n}=o(1),
\end{align*}
where the final $\lesssim$ stems from $\widetilde{C}(2,\alpha)\geqslant 1$.
\end{proof}

\begin{proof}
[\sc{Proof of Lemma \ref{lem:Restricted-Eigenvalue-Bound}}] Let
$T\subseteq[pq]$ and $\bdelta\in\R^{pq}$
satisfy $|T|\in[s]$ and $\Vert\bdelta_{T^{c}}\Vert_{\ell_{1}}\leqslant \widehat c_{0}\widehat{\mu}_{0}\Vert\bdelta_{T}\Vert_{\ell_{1}}$,
where we define $\widehat c_0=(\widehat u c + 1)/(\widehat\ell c -1)$ for some random sequences $\widehat\ell$ and $\widehat u$ provided by asymptotic validity (\ref{eq:AsymptoticallyValidPenaltyLoadings})--(\ref{eq:AsymptoticallyValidPenaltyLoadings2}) of the penalty loadings $\{\widehat\bUpsilon_i\}_{i=1}^p$. Thus,
\[
\Vert\bdelta\Vert_{\ell_{1}}=\Vert\bdelta_T\Vert_{\ell_{1}}+\Vert\bdelta_{T^c}\Vert_{\ell_{1}}\leqslant(1+\widehat c_{0}\widehat{\mu}_{0})\Vert\bdelta_{T}\Vert_{\ell_{1}}\leqslant(1+\widehat c_{0}\widehat{\mu}_{0})\sqrt{s}\Vert\bdelta\Vert_{\ell_{2}},
\]
where the final inequality stems from the Cauchy-Schwarz inequality. Abbreviating $\widehat{\Delta}:=\max_{(i,j)\in[pq]^2}|(\widehat{\bSigma}_{\bZ})_{i,j}-(\bSigma_{\bZ})_{i,j}|$, per the previous display
we get
\begin{align*}
\bdelta^\top\bSigma_{\bZ}\bdelta-\bdelta^\top\widehat{\bSigma}_{\bZ}\bdelta=\bdelta^\top(\bSigma_{\bZ}-\widehat{\bSigma}_{\bZ})\bdelta & \leqslant|\bdelta^\top(\widehat{\bSigma}_{\bZ}-\bSigma_{Z})\bdelta|\\
 & \leqslant\big\Vert(\widehat{\bSigma}_{\bZ}-\bSigma_{\bZ})\bdelta\big\Vert_{\ell_{\infty}}\left\Vert \bdelta\right\Vert _{\ell_{1}}\tag{H{\"o}lder}\\
 & \leqslant\widehat{\Delta}\left\Vert \bdelta\right\Vert _{\ell_{1}}^{2}\tag{H{\"o}lder}\\
 & \leqslant(1+\widehat c_{0}\widehat{\mu}_{0})^{2}s\Vert\bdelta\Vert_{\ell_{2}}^{2}\widehat{\Delta}.
\end{align*}
For~$\bdelta\neq\mathbf{0}_{pq}$ the previous display rearranges to
\[
\frac{\bdelta^\top\widehat{\bSigma}_{\bZ}\bdelta}{\left\Vert \bdelta\right\Vert _{\ell_{2}}^{2}}\geqslant\frac{\bdelta^\top\bSigma_{\bZ}\bdelta}{\left\Vert \bdelta\right\Vert _{\ell_{2}}^{2}}-(1+\widehat c_{0}\widehat{\mu}_{0})^{2}s\widehat{\Delta}.
\]
Since $\bdelta$ lies in the restricted set $\mathcal{R}_{\widehat c_0\widehat \mu_0,T}$, it follows from the definition \eqref{eq:RestrictedEigenvalue} that 
\begin{align*}
\widehat\kappa_{\widehat c_{0}\widehat{\mu}_{0}}^{2} & \geqslant\min_{\substack{T\subseteq[pq]:\\ |T|\in[s]}}\inf_{\substack{\bdelta\in\mathcal{R}_{\widehat c_{0}\widehat{\mu}_{0},T}:\\\bdelta\neq \mathbf{0}_{pq}}}\frac{\bdelta^\top\bSigma_{\bZ}\bdelta}{\left\Vert \bdelta\right\Vert _{\ell_{2}}^{2}}-(1+\widehat c_{0}\widehat{\mu}_{0})^{2}s\widehat{\Delta}
\geqslant 
\Lambda_{\min}(\bSigma_{\bZ})-(1+\widehat c_{0}\widehat{\mu}_{0})^{2}s\widehat{\Delta}.
\end{align*}
As the (unrestricted) eigenvalue $\Lambda_{\min}(\bSigma_{\bZ})$ is bounded away from zero by $d$ (Assumption \ref{assu:CovarianceZ}), to establish the claim it suffices to show that~$(1+\widehat c_{0}\widehat{\mu}_{0})^{2}s\widehat{\Delta}\to_{\P}0$.
To this end, note that asymptotic validity of the penalty loadings $\{\widehat{\bUpsilon}_{i}\}_{i\in[p]}$
implies that $\widehat\ell\to_{\mathrm{P}}1$ and $\widehat u\to_{\P}u$ for some constant $u\in[1,\infty)$.
It follows that $\widehat c_{0}=(\widehat u c+1)/(\widehat\ell c-1)\to_{\P}(uc+1)/(c-1)\in[1,\infty)$. Lemma \ref{lem:Ideal-Penalty-Loading-Control}
establishes $\widehat{\mu}_{0}\lesssim_{\P}1$, which implies that
$\widehat c_{0}\widehat{\mu}_{0}\lesssim_{\P}1$ as well. The growth condition $s^2(\ln(pn))^{1/\tau+\widetilde{C}(2,\alpha)}=o(n)$ implies $(\ln(pn))^{1/\tau+\widetilde{C}(2,\alpha)}=o(n)$, so that by Lemma \ref{lem:Covariance-Consistency} we obtain
\[
\widehat{\Delta}\lesssim_{\P}\sqrt{\frac{(\ln(pn))^{1/\tau+\widetilde{C}(2,\alpha)}}{n}}.
\]
Now $s^2(\ln(pn))^{1/\tau+\widetilde{C}(2,\alpha)}=o(n)$ yields $(1+\widehat c_{0}\widehat{\mu}_{0})^{2}s\widehat{\Delta}\to_{\P}0$, as desired.
\end{proof}

\subsection{Proof of Asymptotic Validity of Penalty Loadings from Algorithm \ref{alg:Data-Driven-Penalization} (Lemma \ref{lem:AsymptoticValidityDataDrivenPenaltyLoadings})}
We first state and proof a lemma giving convergence rates for the Lasso based on any asymptotically valid penalty loadings. We subsequently leverage this intermediate result to prove the asymptotic validity of the data-driven penalty loadings coming out of Algorithm \ref{alg:Data-Driven-Penalization}.
\begin{lem}[\textbf{Convergence Rates for Lasso with Asymptotically Valid Loadings}]\label{lem:Rates-for-LASSO-any-asymptotically-valid-loadings} Let Assumptions \ref{assu:Innovations}, \ref{assu:Companion}, \ref{assu:CovarianceZ} and \ref{assu:RowSparsity} hold, and let $\{\widehat\bUpsilon_i\}_{i\in[p]}$ be any collection of asymptotically valid penalty loadings. Then Lasso estimates $\widehat\bbeta_i:=\widehat\bbeta_i(\lambda^\ast_n,\widehat\bUpsilon_i),i\in[p]$, based on the penalty level~$\lambda_n^\ast$ in (\ref{eq:PenaltyLevelPractice}) and the penalty loadings $\{\widehat\bUpsilon_i\}_{i\in[p]}$ satisfy (\ref{eq:lasso_rates_1})--(\ref{eq:lasso_rates_3}).
\end{lem}
\begin{proof}[\sc{Proof of Lemma \ref{lem:Rates-for-LASSO-any-asymptotically-valid-loadings}}] 
Lemma \ref{lem:Deviation-Bound} shows~$\P(\lambda^\ast_n/n\geqslant c\max_{i\in[p]}\Vert \bS_{n,i}\Vert _{\ell_{\infty}})\to1$, and by Lemma \ref{lem:Ideal-Penalty-Loading-Control} there are constants $(\underline{\upsilon}^{0},\overline{\upsilon}^{0})\in(0,\infty)^2$, depending only on $a_1,a_2,\alpha,b_1, b_2, d$ and $\tau$, such that
$\P(\underline{\upsilon}^{0}\leqslant\widehat{\upsilon}_{\min}^{0}\leqslant\widehat{\upsilon}_{\max}^{0}\leqslant\overline{\upsilon}^{0})\to1$ as well as $\P(\widehat{\mu}_{0}\leqslant\overline\upsilon^0/\underline\upsilon^0)\to1$. Since the penalty loadings $\{\widehat\bUpsilon_i\}_{i\in[p]}$ are asymptotically valid, there are non-negative random variables $\widehat\ell:=\widehat\ell_n$ and $\widehat u:=\widehat u_n,n\in\N$, such that (\ref{eq:AsymptoticallyValidPenaltyLoadings}) holds with probability approaching one and (\ref{eq:AsymptoticallyValidPenaltyLoadings2}) holds. For any choice of such random variables, Lemma \ref{lem:Restricted-Eigenvalue-Bound} provides the restricted eigenvalue bound $\P(\widehat\kappa^2_{\widehat c_0\widehat\mu_0}\geqslant d/2)\to1$, where $\widehat c_{0}=(\widehat uc+1)/(\widehat\ell c-1)$, and $\widehat \kappa_{\widehat c_{0}\widehat{\mu}_{0}}^{2}$ is defined via (\ref{eq:RestrictedEigenvalue}) with $C=\widehat c_{0}\widehat{\mu}_{0}$.
Thus, it follows from the master lemma (Lemma \ref{lem:MasterLemma}) that
\begin{align*}
\max_{i\in[p]}\Vert\widehat{\bbeta}_{i}-\bbeta_{0i}\Vert_{2,n} & \lesssim_\P\frac{\lambda_n^\ast\sqrt{s}}{n},\\
\max_{i\in[p]}\Vert\widehat{\bbeta}_{i}-\bbeta_{0i}\Vert_{\ell_{1}} & \lesssim_\P\frac{\lambda_n^\ast s}{n}\quad\text{and}\\
\max_{i\in[p]}\Vert\widehat{\bbeta}_{i}-\bbeta_{0i}\Vert_{\ell_{2}} & \lesssim_\P\frac{\lambda_n^\ast\sqrt{s}}{n},
\end{align*}
and (\ref{eq:lasso_rates_1})--(\ref{eq:lasso_rates_3}) follow as $\Phi^{-1}(1-z)\leqslant\sqrt{2\ln(1/z)},z\in(0,1)$, implies $\lambda_n^\ast\lesssim \sqrt{n\ln(pn)}$.
\end{proof}

\begin{proof}
[\sc{Proof of Lemma \ref{lem:AsymptoticValidityDataDrivenPenaltyLoadings}}]
The proof is divided into three steps. Step 1 establishes asymptotic validity of the \textit{initial} penalty loadings $\{\widehat{\upsilon}_{i,j}^{(0)}\}_{(i,j)\in[p]\times[pq]}$, while Step 2 proves asymptotic validity of the \textit{once updated} penalty loadings $\{\widehat{\upsilon}_{i,j}^{(1)}\}_{(i,j)\in[p]\times[pq]}$.~Step 3 observes that the argument in Step 2 can be iterated a finite number of times to arrive at the asymptotic validity of the $K$ \textit{times updated}
loadings $\{\widehat{\upsilon}_{i,j}^{(K)}\}_{(i,j)\in[p]\times[pq]}$ for any fixed $K\in\N_0$.

\textbf{Step 1:} To establish asymptotic validity of the \textit{initial} penalty loadings, let
\[
\upsilon_{i,j}^{(0)}:=\del[1]{\E\sbr[1]{(\widehat{\upsilon}_{i,j}^{(0)})^2}}^{1/2}=\del[1]{\E\sbr[1]{\del[1]{Y_{0,i}^{2}Z_{-1,j}^{2}}}}^{1/2},\quad (i,j)\in[p]\times[pq],
\]
and 
\[
\upsilon_{i,j}^{0}:=\del[1]{\E\sbr[1]{(\widehat{\upsilon}_{i,j}^{0})^2}}^{1/2}=\del[1]{\E\sbr[1]{\del[1]{\varepsilon_{0,i}^{2}Z_{-1,j}^{2}}}}^{1/2},\quad (i,j)\in[p]\times[pq].
\]
Since $\E[\bepsilon_{0} |\mathcal{F}_{-1}^{\bseta}] = \mathbf{0}_p$ almost surely (Assumption \ref{assu:Innovations}.\ref{enu:EigValsOfSigmaEpsBddAwayFromZero}) and $\bZ_{-1} \in \mathcal{F}_{-1}^{\bseta}$, we have
\[
\E[Y_{0,i}^2Z_{-1,j}^2]=\E[(\bZ_{-1}^\top\bbeta_{0i}+\varepsilon_{0,i})^{2}Z_{-1,j}^{2}]=\E[(\bZ_{-1}^\top\bbeta_{0i})^{2}Z_{-1,j}^{2}]+\E[\varepsilon_{0,i}^{2}Z_{-1,j}^{2}],
\]
for $(i,j)\in[p]\times[pq]$, implying that
\begin{equation}\label{eq:EInitialUpsSqReltoEIdealUpsSqEquality}
(\upsilon_{i,j}^{(0)})^2=\underbrace{\E[(\bZ_{0}^\top\bbeta_{0i})^{2}Z_{0,j}^{2}]}_{\geqslant0}+(\upsilon_{i,j}^{0})^2.
\end{equation}
Parts \ref{enu:ZiUnifMomentBnd} and \ref{enu:Zbeta_iSubWeibullNormUnifBdd} of Lemma \ref{lem:ZandEpsZNormBnds} together with the Cauchy-Schwarz inequality show that, under Assumptions \ref{assu:Innovations}.\ref{enu:EpsIID}, \ref{assu:Innovations}.\ref{enu:EpsSubWeibullBeta},
\ref{assu:Companion}.\ref{enu:SpectralRadius} and \ref{assu:Companion}.\ref{enu:RowNormDecay}, there is a constant $C_1\in(0,\infty)$, depending only on $a_1,\alpha,b_1,b_2$ and $\tau$, such that
\[
\max_{i\in[p],j\in[pq]}\E[(\bZ_{0}^\top\bbeta_{0i})^{2}Z_{0,j}^{2}]
\leqslant C_1.
\]
It follows from (\ref{eq:EInitialUpsSqReltoEIdealUpsSqEquality}) that for any $(i,j)\in[p]\times[pq]$,
\begin{align}\label{eq:EInitialUpsSqReltoEIdealUpsSqInequalities}
(\upsilon_{i,j}^0)^2\leqslant(\upsilon_{i,j}^{(0)})^2\leqslant C_1 + (\upsilon_{i,j}^0)^2.
\end{align}
Lemma \ref{lem:Ideal-Penalty-Loading-Control} shows that there is a constant $\underline\upsilon^0\in(0,\infty)$, depending only on $a_1,a_2,\alpha,b_1, b_2, d$ and $\tau$, such that $\P(\widehat\upsilon_{\min}^0\geqslant\underline\upsilon^0)\to1.$ To simplify notation, without loss of generality, we henceforth argue as if $\P(\widehat\upsilon_{\min}^0\geqslant\underline\upsilon^0)=1$ for all $n$.
Lemma \ref{lem:DeltaIdealVanishes} shows that $\widehat\Delta^0$ in (\ref{eq:DeltaIdealDefn}) satisfies
\begin{align}\label{eq:DeltaIdealSqVanishes}
\widehat\Delta^0=\max_{i\in[p],j\in[pq]}\envert[1]{(\widehat{\upsilon}_{i,j}^{0})^2-(\upsilon_{i,j}^{0})^2}\overset{\P}{\to}0.
\end{align}
Suppose for now that also
\begin{align}\label{eq:DeltaInitialSqVanishes}
\widehat\Delta^{(0)}:=\max_{i\in[p],j\in[pq]}\envert[1]{(\widehat{\upsilon}_{i,j}^{(0)})^2-(\upsilon_{i,j}^{(0)})^2}\overset{\P}{\to}0.
\end{align}
(We establish (\ref{eq:DeltaInitialSqVanishes}) below.) We now construct non-negative
random scalars $\widehat\ell^{(0)}:=\widehat\ell^{(0)}_n$ and $\widehat{u}^{(0)}:=\widehat{u}^{(0)}_n,n\in\N$, and a constant $u^{(0)}\in[1,\infty)$ such that (\ref{eq:AsymptoticallyValidPenaltyLoadings}) holds with probability approaching one and (\ref{eq:AsymptoticallyValidPenaltyLoadings2}) holds. To construct $\widehat\ell^{(0)}$, observe that for any $(i,j)\in[p]\times[pq]$, by the triangle inequality and the lower bound in (\ref{eq:EInitialUpsSqReltoEIdealUpsSqInequalities}),
\[
(\widehat\upsilon_{i,j}^{(0)})^2
\geqslant(\upsilon_{i,j}^{(0)})^2-\widehat\Delta^{(0)}
\geqslant(\upsilon_{i,j}^{0})^2-\widehat\Delta^{(0)}
\geqslant(\widehat\upsilon_{i,j}^{0})^2-(\widehat\Delta^0+\widehat\Delta^{(0)}).
\]
Since $0<\underline\upsilon^0\leqslant\widehat\upsilon_{\min}^0\leqslant\widehat\upsilon_{i,j}^0$, we can continue this string of inequalities to get
\[
(\widehat\upsilon_{i,j}^{(0)})^2\geqslant\left(1-\frac{\widehat\Delta^0+\widehat\Delta^{(0)}}{(\underline\upsilon^0)^2}\right)(\widehat\upsilon_{i,j}^{0})^2.
\]
As penalty loadings are also non-negative, these observations suggest
\begin{equation}\label{eq:ellhatInitDefn}
\widehat{\ell}^{(0)}:=\sqrt{\max\left\{ 0,1-\frac{\widehat{\Delta}^{0}+\widehat{\Delta}^{(0)}}{(\underline{\upsilon}^{0})^{2}}\right\}}
\end{equation}
To construct $\widehat{u}^{(0)}$, observe that for any $(i,j)\in[p]\times[pq]$, by the triangle inequality and the upper bound in (\ref{eq:EInitialUpsSqReltoEIdealUpsSqInequalities}),
\[
(\widehat\upsilon_{i,j}^{(0)})^2\leqslant(\upsilon_{i,j}^{(0)})^2+\widehat\Delta^{(0)}\leqslant(\upsilon_{i,j}^{0})^2 + C_1 + \widehat\Delta^{(0)}\leqslant(\widehat\upsilon_{i,j}^{0})^2 + C_1 + \widehat\Delta^{0} + \widehat\Delta^{(0)}.
\]
Again, since $0<\underline\upsilon^0\leqslant\widehat\upsilon_{\min}^0\leqslant\widehat\upsilon_{i,j}^0$, we can continue this string of inequalities to get
\[
(\widehat\upsilon_{i,j}^{(0)})^2\leqslant\left(1+ \frac{C_1 + \widehat\Delta^{0} + \widehat\Delta^{(0)}}{(\underline\upsilon^0)^2}\right)(\widehat\upsilon_{i,j}^{0})^2.
\]
The previous displays suggests
\begin{equation}\label{eq:uhatInitDefn}
\widehat{u}^{(0)}:=
\sqrt{1+ \frac{C_1 + \widehat\Delta^{0} + \widehat\Delta^{(0)}}{(\underline\upsilon^0)^2}}
\end{equation}
By construction of (\ref{eq:ellhatInitDefn}) and (\ref{eq:uhatInitDefn}), from (\ref{eq:DeltaIdealSqVanishes}) and the working hypothesis (\ref{eq:DeltaInitialSqVanishes}),
we see that (\ref{eq:AsymptoticallyValidPenaltyLoadings}) holds with probability approaching one. We also see that $\widehat\ell^{(0)}\to_\P1$ and $\widehat u^{(0)}\to_\P u^{(0)}$ for the constant $u^{(0)}:=[1+C_1/(\underline\upsilon^0)^2]^{1/2}\in[1,\infty)$, which depends only on $a_1,a_2,\alpha,b_1,b_2,d$ and $\tau$, showing that also (\ref{eq:AsymptoticallyValidPenaltyLoadings2}) holds. Hence, as long as we can establish (\ref{eq:DeltaInitialSqVanishes}), the initial penalty loadings $\{\widehat{\upsilon}_{i,j}^{\left(0\right)}\}_{(i,j)\in[p]\times[pq]}$ are asymptotically valid.

To this end, we set up for an application of Lemma \ref{lem:MaximalInequalityFunctionalDependence}.
Part \ref{enu:ZiMarginalSubWeibullNormUnifBdd} of Lemma \ref{lem:ZandEpsZNormBnds} shows that for some constant $C_2\in(0,\infty)$, depending only on $a_1,b_1,b_2$ and $\tau$,
\begin{equation}\label{eq:repeatstart}
\max_{i\in[pq]}\left\Vert Z_{0,i}\right\Vert_{\psi_{\alpha}}\leqslant C_2.
\end{equation}
Combined with the equivalent definitions of sub-Weibullness in \citet[Theorem 1]{vladimirova_sub-weibull_2020} and the bound in \citet[Proposition D.2]{kuchibhotla_moving_2022}, the previous display shows that for some constant $C_3\in(0,\infty)$, depending only on $a_1,\alpha,b_1,b_2$ and $\tau$,
\[
\max_{(i,j)\in[p]\times[pq]}\enVert[1]{ Y_{0,i}^2Z_{-1,j}^2}_{\psi_{\alpha/4}}\leqslant C_3.
\]
Setting $r=4$, Part \ref{enu:ZiUnifMomentBnd} of Lemma \ref{lem:ZandEpsZNormBnds} shows that there is a constant~$C_4\in(0,\infty)$, depending only on $a_1,\alpha,b_1,b_2$ and $\tau$, such that
\[
\max_{i\in[pq]}\Vert Z_{0,i}\Vert_{4}\leqslant C_4,
\]
and Part \ref{enu:ZiUnifGMC} of Lemma \ref{lem:OutcomeAndScoreProcessesUniformlyGMC} shows that there are constants $C_5,C_6\in(0,\infty)$, depending only on $a_1,\alpha,b_1,b_2,q, \overline{q}$ and $\tau$, such that
\[
\max_{i\in[pq]}\Delta^{Z_{\cdot,i}}_{4}(h)\leqslant C_5\mathrm{e}^{-C_6h^\tau},\quad h\in\N.
\]
Repeated use of Part \ref{enu:Delta_rFromPairToProduct} of Lemma \ref{lem:Delta_r} (with $r=1$ and then $r=2$), the previous two displays, and reasoning as in the proof of Part \ref{enu:EpsiZjUnifGMC} of Lemma \ref{lem:OutcomeAndScoreProcessesUniformlyGMC} (to account for the lag in $\bZ_{t-1}$), we see that there are constants $C_7,C_8\in(0,\infty)$, depending only on $a_1,\alpha,b_1,b_2,q, \overline{q}$ and $\tau$, such that
\[
\max_{(i,j)\in[p]\times[pq]}\Delta_1^{Y_{\cdot,i}^2Z_{\cdot-1,j}^2}(h)\leqslant C_7\mathrm{e}^{-C_8h^{\tau}},\quad h\in\N.
\]
By Assumption \ref{assu:RowSparsity}, we have $(\ln p)^{C(\alpha,\tau)}=o(n)$ with
\[
C(\alpha,\tau)=\max\cbr{\frac{\widetilde{C}(4,\alpha)(1+4\tau)}{4\tau-1},\frac{1+4\tau}{4\tau-2},\frac{1+4\tau}{\tau}}\quad\text{and}\quad\widetilde{C}(4,\alpha)= \max\cbr{\frac{8}{\alpha},\frac{4+\alpha}{\alpha}},
\]
which implies that $(\ln p)^{1/\tau}=o(n)$, so that eventually $\lceil(2\ln(pn))^{1/\tau}/C_8^{1/\tau}\rceil\in[n]$.
Applying Lemma \ref{lem:MaximalInequalityFunctionalDependence} for the vectorization of the strictly stationary and causal joint process $\{Y_{t,i}^2Z_{t-1,j}^2:(i,j)\in[p]\times[pq]\}_{t\in\Z}$ with $r=1$,~$\xi=\alpha/4$,~$K=C_3$,~$T=n$,~$d_T=p^2 q$ and  $m_T=\lceil(2\ln(p n))^{1/\tau}/C_8^{1/\tau}\rceil$, and using that $4/\alpha + 1/\min\{\alpha/4,1\} = \widetilde{C}(4,\alpha)$, we see that, as $n\to\infty$, 
\begin{align}
\widehat\Delta^{(0)}
&=\max_{(i,j)\in[p]\times[pq]}\envert[1]{(\widehat{\upsilon}_{i,j}^{(0)})^2-(\upsilon_{i,j}^{(0)})^2}\notag \\
&=\frac{1}{n}\max_{(i,j)\in[p]\times[pq]}\envert[4]{\sum_{t=1}^{n}\left(Y_{t,i}^{2}Z_{t-1,j}^{2}-\E[Y_{0,i}^{2}Z_{-1,j}^{2}]\right)}\notag \\
&\lesssim_{\P}\frac{1}{n}\left(np^{2}q\mathrm{e}^{-C_{8}m_T^\tau}+m_T\sqrt{l_{T}\ln(p^2qn)}+m_T(\ln(p^2qn))^{\widetilde{C}(4,\alpha)}\right)\tag{Lemma \ref{lem:MaximalInequalityFunctionalDependence}}\\
&\lesssim \frac{1}{n} + \sqrt{\frac{(\ln(pn))^{1/\tau+1}}{n}} + \frac{(\ln(pn))^{1/\tau+\widetilde{C}(4,\alpha)}}{n}\tag{Assumption \ref{assu:Companion}.\ref{enu:lag-order} and choice of $m_T$} \\
&\lesssim \sqrt{\frac{(\ln(pn))^{1/\tau+\widetilde{C}(4,\alpha)}}{n}}\tag{$\widetilde{C}(4,\alpha)\geqslant 1$}\\
&\leqslant \sqrt{\frac{(\ln(pn))^{C(\alpha,\tau)}}{n}}=o(1),\label{eq:repeatend}
\end{align}
where the $\leqslant$ holds as long as we can show $1/\tau+\widetilde{C}(4,\alpha)\leqslant C(\alpha,\tau)$. To this end, note that $\widetilde{C}(4,\alpha) = 8/\alpha$ for $\alpha \in (0,2]$ and $\widetilde{C}(4,\alpha)\leqslant 4$ for $\alpha \geqslant 2$, such that
\[
\frac{1}{\tau} +  \widetilde{C}(4,\alpha) = \frac{1}{\tau} + \frac{8}{\alpha} \leqslant \frac{8(1+4\tau)}{\alpha(4\tau-1)} \leqslant C(\alpha,\tau)\quad \text{for $\alpha \in (0, 2]$}
\]
and 
\[
1/\tau + \widetilde{C}(4,\alpha) \leqslant 1/\tau + 4 = (1+4\tau)/\tau \leqslant C(\alpha,\tau)\quad \text{for $\alpha \geqslant 2$}.
\]

\textbf{Step 2:} If $K=0$, then there is nothing left to show. We therefore presume that the number of updates $K\in\N$, and consider the once updated penalty loadings
$\{\widehat{\upsilon}_{i,j}^{(1)}\}_{(i,j)\in[p]\times[pq]}$ from 
(\ref{eq:PenaltyLoadingsRefined}). As observed in Step 1, Lemma \ref{lem:Ideal-Penalty-Loading-Control} shows that there is a constant $\underline\upsilon^0\in(0,\infty)$, depending only on $a_1,a_2,\alpha,b_1,b_2,d$ and $\tau$, such that $\P(\widehat\upsilon_{\min}^0\geqslant\underline\upsilon^0)\to1.$ As in Step 1, we (without loss of generality) continue to argue as if $\P(\widehat\upsilon_{\min}^0\geqslant\underline\upsilon^0)=1$ for all $n$.
Recall that (\ref{eq:DeltaIdealSqVanishes}) holds. Suppose for now that also
\begin{align}\label{eq:DeltaUpdateSqVanishes}
\widehat\Delta^{(1)}:=\max_{(i,j)\in[p]\times[pq]}\envert[1]{(\widehat{\upsilon}_{i,j}^{(1)})^2-(\upsilon_{i,j}^{0})^2}\overset{\P}{\to}0.
\end{align}
(We establish (\ref{eq:DeltaUpdateSqVanishes}) below.) We now construct 
random scalars $\widehat\ell^{(1)}:=\widehat\ell^{(1)}_n$ and $\widehat{u}^{(1)}:=\widehat{u}^{(1)}_n,n\in\N$, such that (\ref{eq:AsymptoticallyValidPenaltyLoadings}) holds with probability approaching one and (\ref{eq:AsymptoticallyValidPenaltyLoadings2}) holds with $\widehat u^{(1)}\to_\P1$. To construct $\widehat\ell^{(1)}$, observe that for any $(i,j)\in[p]\times[pq]$, by the triangle inequality and $0<\underline\upsilon^0\leqslant\widehat\upsilon_{\min}^0\leqslant\widehat\upsilon_{i,j}^0$,
\[
(\widehat\upsilon_{i,j}^{(1)})^2
\geqslant(\upsilon_{i,j}^{0})^2-\widehat\Delta^{(1)}
\geqslant(\widehat\upsilon_{i,j}^{0})^2-(\widehat\Delta^0+\widehat\Delta^{(1)})
\geqslant\left(1-\frac{\Delta^0+\widehat\Delta^{(1)}}{(\underline\upsilon^0)^2}\right)(\widehat\upsilon_{i,j}^{0})^2.
\]
A penalty loadings are non-negative, these observations suggest
\begin{equation}\label{eq:ellhatUpdDefn}
\widehat{\ell}^{(1)}:=\sqrt{\max\left\{ 0,1-\frac{\widehat{\Delta}^{0}+\widehat{\Delta}^{(1)}}{(\underline{\upsilon}^{0})^{2}}\right\}}.
\end{equation}
Similar arguments lead us to
\begin{equation}\label{eq:uhatUpdDefn}
\widehat{u}^{(1)}:=\sqrt{1+\frac{\widehat{\Delta}^{0}+\widehat{\Delta}^{(1)}}{(\underline{\upsilon}^{0})^{2}}}.
\end{equation}
By construction of (\ref{eq:ellhatUpdDefn}) and (\ref{eq:uhatUpdDefn}), from (\ref{eq:DeltaIdealSqVanishes}) and the working hypothesis (\ref{eq:DeltaUpdateSqVanishes}),
we see that (\ref{eq:AsymptoticallyValidPenaltyLoadings}) holds with probability approaching one. Furthermore, $\widehat\ell^{(1)}\to_\P1$ and $\widehat u^{(1)}\to_\P 1$, showing that also (\ref{eq:AsymptoticallyValidPenaltyLoadings2}) holds, now with $u=1$. Hence, it remains to establish (\ref{eq:DeltaUpdateSqVanishes}).

To this end, define the (infeasible) penalty loadings
\[
\widetilde{\upsilon}_{i,j}^{0}:=\left(\frac{1}{n}\sum_{t=1}^{n}\varepsilon_{t,i}^{2}Z_{t-1,j}^{2}\right)^{1/2},\quad(i,j)\in[p]\times[pq],
\]
which correspond to the ideal penalty loadings without any blocking.
To show $\widehat\Delta^{(1)}\to_\P0$, by the triangle inequality, it suffices that
\[
\widetilde\Delta^0:=\max_{(i,j)\in[p]\times[pq]}\envert[1]{(\widetilde{\upsilon}_{i,j}^{0})^2-(\upsilon_{i,j}^{0})^2}\overset{\P}{\to}0\quad\text{and}\quad
\widetilde\Delta^{(1)}:=\max_{(i,j)\in[p]\times[pq]}\envert[1]{(\widehat{\upsilon}_{i,j}^{(1)})^2-(\widetilde{\upsilon}_{i,j}^{0})^2}\overset{\P}{\to}0.
\]
To show that $\widetilde\Delta^0\to_\P0$, note that
\[
\widetilde\Delta^0=\frac{1}{n}\max_{(i,j)\in[p]\times[pq]}\left|\sum_{t=1}^{n}\left(\varepsilon_{t,i}^{2}Z_{t-1,j}^{2}-\E[\varepsilon_{0,i}^{2}Z_{-1,j}^{2}]\right)\right|.
\]
By Assumption \ref{assu:Innovations}.\ref{enu:EpsSubWeibullBeta} one has~$\max_{i\in[p]}\Vert\varepsilon_{0,i}\Vert_{\psi_{\alpha}}\leqslant a_1$. In addition,~using that $\{\varepsilon_{t}\}_{t\in\Z}$ is $\overline{q}$-dependent (Assumption \ref{assu:Innovations}.\ref{enu:EpsIID}), it trivially holds that $\{\varepsilon_{t,i}\}_{t\in\Z}$ is a geometric moment contraction. Now, replacing~$Y_{t,i}$ with~$\varepsilon_{t,i}$ in the arguments starting at~\eqref{eq:repeatstart} and leading to~\eqref{eq:repeatend} yields~$\widetilde{\Delta}^0\to_\P0$.

To show that $\widetilde\Delta^{(1)}\to_\P0$, recall that $\widehat{\varepsilon}_{t,i}^{(0)}=Y_{t,i}-\bZ_{t-1}^\top\widehat{\bbeta}_{i}^{(0)}=\varepsilon_{t,i}-\bZ_{t-1}^\top\widehat{\bdelta}_{i}^{(0)},$
where $\widehat{\bdelta}_{i}^{(0)}:=\widehat{\bbeta}_{i}^{(0)}-\bbeta_{0i}$ is the estimation error arising from the initial penalty loadings. Hence,
\begin{align*}
\widetilde\Delta^{(1)}
&=\max_{(i,j)\in[p]\times[pq]}\envert[3]{\frac{1}{n}\sum_{t=1}^{n}\sbr[1]{\widehat{\varepsilon}_{ti}^{\left(0\right)})^2-\varepsilon_{ti}^{2}}Z_{t-1,j}^{2}}\\
&=\max_{(i,j)\in[p]\times[pq]}\envert[3]{\frac{1}{n}\sum_{t=1}^{n}\left[(\bZ_{t-1}^\top\widehat{\bdelta}_{i}^{(0)})^{2}-2\varepsilon_{t,i}\bZ_{t-1}^\top\widehat{\bdelta}_{i}^{(0)}\right]Z_{t-1,j}^{2}}\\
&\leqslant2\underbrace{\max_{(i,j)\in[p]\times[pq]}\envert[3]{\frac{1}{n}\sum_{t=1}^{n}\varepsilon_{t,i}Z_{t-1,j}^{2}\bZ_{t-1}^\top\widehat{\bdelta}_{i}^{(0)}}}_{=:\mathrm{I}_{n}}+
\underbrace{\max_{(i,j)\in[p]\times[pq]}\envert[3]{\frac{1}{n}\sum_{t=1}^{n}|\bZ_{t-1}'\widehat{\bdelta}_{i}^{(0)}|^{2}Z_{t-1,j}^{2}}}_{=:\mathrm{II}_{n}}.
\end{align*}
To show $\widetilde\Delta^{(1)}\to_\P0$, we argue $\mathrm{I}_{n}\to_\P0$ and $\mathrm{II}_{n}\to_\P0$, in turn. First,
\begin{align}
\mathrm{I}_{n}
&\leqslant\max_{(i,j)\in[p]\times[pq]}\sqrt{\frac{1}{n}\sum_{t=1}^{n}\varepsilon_{t,i}^{2}Z_{t-1,j}^{4}}\cdot\Vert\widehat{\bdelta}_{i}^{(0)}\Vert_{2,n}\tag{Cauchy-Schwarz}\nonumber\\
&\leqslant\underbrace{\max_{t\in[n]}\max_{j\in[pq]}\left|Z_{t-1,j}\right|}_{=:\mathrm{I}_n^{(a)}}\cdot\underbrace{\max_{(i,j)\in[p]\times[pq]}\widetilde{\upsilon}_{i,j}^{0}}_{=:\mathrm{I}_n^{(b)}}\cdot\underbrace{\max_{i\in[p]}\Vert\widehat{\bdelta}_{i}^{(0)}\Vert_{2,n}}_{=:\mathrm{I}_n^{(c)}}\label{eq:DeltaTildeUpdPartIBnd}.
\end{align}
We next handle each of the terms in (\ref{eq:DeltaTildeUpdPartIBnd}), starting with $\mathrm{I}_n^{(c)}$. From Lemma \ref{lem:Rates-for-LASSO-any-asymptotically-valid-loadings} and the asymptotic validity of the initial loadings (Step 1), we know that
\[
\mathrm{I}_n^{(c)}=\max_{i\in[p]}\Vert\widehat{\bdelta}_{i}^{(0)}\Vert_{2,n}\lesssim_\P\sqrt{s\ln(pn)/n}.
\]
Turning to $\mathrm{I}_n^{(b)}$, as
\[
\widetilde\Delta^0=\max_{(i,j)\in[p]\times[pq]}\envert[1]{(\widetilde{\upsilon}_{i,j}^{0})^2-(\upsilon_{i,j}^{0})^2}\overset{\P}{\to}0
\]
and, using Part \ref{enu:EpsiZjUnifMomentBnd} of Lemma \ref{lem:ZandEpsZNormBnds}, we see that
\[
\max_{(i,j)\in[p]\times[pq]}\upsilon_{i,j}^0 =\max_{(i,j)\in[p]\times[pq]}(\E[\varepsilon_{0,i}^2Z_{-1,j}^2])^{1/2}\leqslant a_14^{2/\alpha} C(a_1,b_1,b_2,\tau)<\infty,
\]
we must have 
\[
\mathrm{I}_n^{(b)}=\max_{(i,j)\in[p]\times[pq]}\widetilde{\upsilon}_{i,j}^{0}\lesssim_{\P}1.
\]
Concerning $\mathrm{I}_n^{(a)}$, Part \ref{enu:ZiMarginalSubWeibullNormUnifBdd} of Lemma \ref{lem:ZandEpsZNormBnds} yields
\[
\max_{j\in[pq]}\Vert Z_{0,j}\Vert_{\psi_{\alpha}}\leqslant C(a_1,b_1,b_2,\tau)<\infty.
\]
Thus, by \citet[Theorem 1]{vladimirova_sub-weibull_2020}, there is a constant $C_9\in(0,\infty)$, depending only on $a_1,\alpha,b_1,b_2$ and $\tau$, such that
\[
\P(|Z_{0,j}|\geqslant x)\leqslant2\mathrm{e}^{-(x/C_{9})^{\alpha}}
\]
holds for all $x\in(0,\infty)$ and all $j\in[pq]$. A union bound argument thus leads to
\[
\max_{t\in[n]}\max_{j\in[pq]}\left|Z_{t-1,j}\right|\lesssim_{\P}(\ln(pn))^{1/\alpha},
\]
so gathering our findings in (\ref{eq:DeltaTildeUpdPartIBnd}), we see that 
\[
\mathrm{I}_n\lesssim_\P (\ln(pn))^{1/\alpha}\cdot\sqrt{\frac{s\ln(pn)}{n}}=\sqrt{\frac{s(\ln(pn))^{1+2/\alpha}}{n}}\leqslant \sqrt{\frac{s(\ln(pn))^{1/\tau+\widetilde{C}(2,\alpha)}}{n}}=o(1).
\] 
Our previous calculations also show that
\begin{align*}
\mathrm{II}_{n}
&=\max_{(i,j)\in[p]\times[pq]}\envert[3]{\frac{1}{n}\sum_{t=1}^{n}|\bZ_{t-1}^\top\widehat{\bdelta}_{i}^{(0)}|^{2}Z_{t-1,j}^{2}}
\leqslant\max_{t\in[n]}\max_{j\in[pq]}Z_{t-1,j}^{2}\cdot\max_{i\in[p]}\Vert\widehat{\bdelta}_{i}^{(0)}\Vert_{2,n}^{2}\\
&\lesssim_{\P}(\ln(pn))^{2/\alpha}\cdot\frac{s\ln(pn)}{n}
=\frac{s(\ln(pn))^{1+2/\alpha}}{n}
=o\left(1\right).
\end{align*}
It follows that $\widetilde\Delta^{(1)}\to_\P0$ and, thus, the once updated penalty loadings $\{\widehat\upsilon_{i,j}^{(1)}\}_{(i,j)\in[p]\times[pq]}$ are asymptotically valid. 

\textbf{Step 3:} To finish the proof, observe that asymptotic validity of the once updated penalty loadings (Step 2) and Lemma \ref{lem:Rates-for-LASSO-any-asymptotically-valid-loadings} imply
\[
\max_{i\in[p]}\Vert\widehat{\bdelta}_{i}^{(1)}\Vert_{2,n}\lesssim_{\P}\sqrt{s\ln(pn)/n},
\]
with $\widehat{\bdelta}_{i}^{(1)}:=\widehat{\bbeta}_{i}^{(1)}-\bbeta_{0i},i\in[p]$,
being the estimation error arising from the once updated penalty loadings. Since the rate of convergence of the resulting $\Vert\widehat{\bdelta}_{i}^{\left(k\right)}\Vert_{2,n}$
remains $\sqrt{s\ln(pn)/n}$ for any finite $k$, we can iterate on the argument given in Step 2 to obtain asymptotic validity
of the penalty loadings constructed from $K\in\N$ updates.
\end{proof}

\subsection{Proof of Convergence Rates for Lasso with Data-Driven Penalization (Theorem \ref{thm:Rates-for-LASSO-data-driven-loadings})}
\begin{proof}[\sc{Proof of Theorem \ref{thm:Rates-for-LASSO-data-driven-loadings}}]
The assumptions of the theorem suffice for asymptotic validity of the penalty loadings $\{\widehat\upsilon_{i,j}^{(K)}\}_{(i,j)\in[p]\times[pq]}$ generated by Algorithm \ref{alg:Data-Driven-Penalization} with any fixed number $K\in\N_0$ of updates, cf.~Lemma \ref{lem:AsymptoticValidityDataDrivenPenaltyLoadings}. The rates in (\ref{eq:lasso_rates_1})--(\ref{eq:lasso_rates_3}) thus follow from Lemma \ref{lem:Rates-for-LASSO-any-asymptotically-valid-loadings}.
\end{proof}

\section{Proof of Theorem~\ref{thm:postLasso} and an Alternative Post-Lasso}
We first state a theorem for least squares estimation following any estimators which (i) satisfy the rates~\eqref{eq:lasso_rates_1}--\eqref{eq:lasso_rates_3} and (ii) do not select ``too many'' irrelevant variables. We then use that the versions of the Lasso studied (differing only in the construction of the loadings) possess these properties to prove Theorems~\ref{thm:postLasso} and \ref{thm:postgeneralestim}, the latter of which is stated in Section~\ref{sec:OLSforweights} below.

For any estimator~$\widehat{\bbeta}_i$ (e.g.~the Lasso) of~$\bbeta_{0i}$, let~$\widehat{T}_i:=\supp{(\widehat{\bbeta}_{i})}$ be the indices of its non-zero coefficients. Denote by
\begin{equation}\label{eq:postestim}
\widetilde{\bbeta}_{i}\in\argmin_{\mathclap{\substack{\bbeta\in\R^{pq}\\ \supp{(\bbeta)\subseteq \widehat{T}_i}}}} \widehat{Q}_{i}\del[0]{\bbeta},\qquad i\in[p].    
\end{equation}
the (set of) least squares estimator(s) including only those variables in~$\widehat{T}_i$ ``selected'' by~$\widehat{\bbeta}_i$. If~$\widehat{\bbeta}_i=\widehat{\bbeta}_i(\lambda,\widehat\bUpsilon_i)$ is a Lasso with penalty~$\lambda$ and loadings~$\widehat\bUpsilon_i$ from~\eqref{eq:LASSOVector}, then~$\widetilde{\bbeta}_i=\widetilde\bbeta_i(\lambda,\widehat\bUpsilon_i)$ is a post-Lasso estimator from~\eqref{eq:PostLassom}. Define~$T_{i}:=\supp{(\bbeta_{0i})}$.
\begin{thm}\label{thm:postgeneralestim}
Let Assumptions \ref{assu:Innovations}, \ref{assu:Companion}, \ref{assu:CovarianceZ} and \ref{assu:RowSparsity} hold, let~$\limsup_{n\to\infty}\phi_{\max}(s\ln(n), \bSigma_{\bZ})<\infty$, and let the estimators~$\widehat{\bbeta}_i,i\in[p]$,  satisfy the rates in \eqref{eq:lasso_rates_1}--\eqref{eq:lasso_rates_3} and~$\max_{i\in[p]}|\widehat{T}_i\setminus T_i|\lesssim_{\P} s$. Then the least squares estimators~$\widetilde{\bbeta}_{i},i\in[p]$, in~\eqref{eq:postestim} following the selections of~$\widehat{\bbeta}_i,i\in[p]$, will also satisfy the rates in~\eqref{eq:lasso_rates_1}--\eqref{eq:lasso_rates_3}.
\end{thm}

\begin{proof}
Throughout this proof, norms of vectors and matrices indexed by the empty set are interpreted as zero. Let~$i\in[p]$. Since on~$\cbr[0]{\widehat{T}_i=\emptyset}$ it holds   that~$\widetilde{\bbeta}_{i}=\widehat{\bbeta}_{i}=\bm{0}_{pq\times1}$, one has for any norm~$\left\Vert\cdot\right\Vert$ and~$i\in[p]$
\begin{align*}
\enVert[1]{\widetilde{\bbeta}_{i}-\bbeta_{0i}}
&=
\enVert[1]{\widetilde{\bbeta}_{i}-\bbeta_{0i}}\mathbf{1}(\widehat{T}_i\neq \emptyset)
+
\enVert[1]{\widehat{\bbeta}_{i}-\bbeta_{0i}}\mathbf{1}(\widehat{T}_i= \emptyset)\\
&\leqslant
\enVert[1]{\widetilde{\bbeta}_{i}-\bbeta_{0i}}\mathbf{1}(\widehat{T}_i\neq \emptyset)
+
\enVert[1]{\widehat{\bbeta}_{i}-\bbeta_{0i}}.
\end{align*}
By assumption,~$\enVert[0]{\widehat{\bbeta}_{i}-\bbeta_{0i}}$ is of the desired order of magnitude for any of the norms considered in~\eqref{eq:lasso_rates_1}--\eqref{eq:lasso_rates_3}. Thus, it remains to bound~$\enVert[1]{\widetilde{\bbeta}_{i}-\bbeta_{0i}}\mathbf{1}(\widehat{T}_i\neq \emptyset)$ and we can assume~$\widehat{T}_i\neq \emptyset$. To this end, it holds by assumption that
\begin{align}\label{eq:selectedsetcardinality}
|\widehat{T}_i|
\leqslant
|\widehat{T}_i\setminus T_i|+|T_i|
\lesssim_{\mathrm{P}} s
\end{align}
uniformly in~$i\in[p]$. The norm equivalence relations table on p.~314 in \cite{hornjohnson} shows that the~$\ell_2$ operator norm of a square matrix is bounded from above by the (row/column) dimension times its maximum absolute entry. Lemma~\ref{lem:Covariance-Consistency} and \eqref{eq:selectedsetcardinality} thus imply
\begin{align}\label{eq:sigmahat2norm}
    \enVert[1]{(\widehat{\bSigma}_{\bZ}-\bSigma_{\bZ})_{\widehat T_i,\widehat T_i}}_{\ell_2}
\leqslant
|\widehat T_i|\max_{(i,j)\in [pq]^2}\envert[1]{(\widehat{\bSigma}_{\bZ}-\bSigma_{\bZ})_{i,j}}
\lesssim_{\mathrm{P}}
s \sqrt{\frac{(\ln(pn))^{1/\tau+\widetilde{C}(2,\alpha)}}{n}},
\end{align} 
uniformly in~$i\in[p]$. Thus, by Assumptions~\ref{assu:CovarianceZ} and~\ref{assu:RowSparsity},~$(\bfX_{\widehat T_i})^\top \bfX_{\widehat T_i}=n\widehat{\bSigma}_{\bZ,\widehat T_i,\widehat T_i}$ is invertible with probability approaching one. 
In addition, since~$[pq]=\widehat T_i\cup (\widehat T_i^c\cap T_i)\cup (\widehat T_i^c\cap T_i^c)$ and~$\bbeta_{0i, \widehat T_i^c\cap T_i^c}=\bm{0}_{|\widehat T_i^c\cap T_i^c|}$ (when~$\widehat T_i^c\cap T_i^c\neq \emptyset$), one has\footnote{We leave~$\bbeta_{0i, \widehat T_i^c\cap T_i^c}$ undefined when~$\widehat T_i^c\cap T_i^c=\emptyset$ and omit this proviso in the sequel in similar situations.}
\begin{align*}
\bm{y}_{i}
=
\bfX\bbeta_{0i}+\bu_{i}
=
\bfX_{\widehat T_i}\bbeta_{0i,\widehat{T}_i}+\bfX_{\widehat T_i^c\cap T_i}\bbeta_{0i,\widehat T_i^c\cap T_i}+\bu_{i},
\end{align*}
where~$\bu_i=(\varepsilon_{1,i},\hdots,\varepsilon_{n,i})^\top$. Hence, with probability approaching one
\begin{align}
\widetilde{\bbeta}_{i,\widehat T_i}
&=
\del[2]{(\bfX_{\widehat T_i})^\top \bfX_{\widehat T_i}}^{-1}\bfX_{\widehat T_i}^\top \bm{y}_{i}\notag\\
&=
\bbeta_{0i,\widehat T_i}+\del[2]{(\bfX_{\widehat T_i})^\top \bfX_{\widehat T_i}}^{-1}(\bfX_{\widehat T_i})^\top\del[1]{\bfX_{\widehat T_i^c\cap T_i}\bbeta_{0i,\widehat T_i^c\cap T_i}+\bu_{i}}.\label{eq:generic_refitting_decomposition_i}
\end{align}
We proceed by bounding the~$\ell_2$ error terms
\begin{align*}
\enVert[2]{\del[1]{(\bfX_{\widehat T_i})^\top \bfX_{\widehat T_i}}^{-1}(\bfX_{\widehat T_i})^\top \bfX_{\widehat T_i^c\cap T_i}\bbeta_{0i,\widehat T_i^c\cap T_i}}_{\ell_2}\qquad \text{and}\qquad \enVert[2]{\del[1]{(\bfX_{\widehat T_i})^\top \bfX_{\widehat T_i}}^{-1}(\bfX_{\widehat T_i})^\top \bu_i}_{\ell_2}
\end{align*}
in turn. To bound the former, observe that by~\eqref{eq:lasso_rates_3}
\begin{align}\label{eq:normrelnonsel}
	\enVert[1]{\bbeta_{0i,\widehat T_i^c\cap T_i}}_{\ell_2}
	=
		\enVert[1]{\widehat\bbeta_{i,\widehat T_i^c\cap T_i}-\bbeta_{0i,\widehat T_i^c\cap T_i}}_{\ell_2}
		\leqslant
		\max_{i\in[p]}\enVert[1]{\widehat\bbeta_{i}-\bbeta_{0i}}_{\ell_2}
		\lesssim_{\mathrm{P}}
		\sqrt{\frac{s\ln\left(pn\right)}{n}},
\end{align}
uniformly in~$i\in[p]$. Furthermore, from the inequality in~\citet[Equation~(5.8.4)]{hornjohnson} (with their matrices $A$ and $E$ set to $\bSigma_{\bZ,\widehat{T}_i,\widehat{T}_i}$ and $(\widehat{\bSigma}_{\bZ}-\bSigma_{\bZ})_{\widehat T_i,\widehat T_i}$, respectively),~\eqref{eq:sigmahat2norm}, Assumptions \ref{assu:CovarianceZ}, \ref{assu:RowSparsity} and~$\limsup_{n\to\infty}\phi_{\max}(s\ln(n), \bSigma_{\bZ})<\infty$ combine to show that,
\[
\max_{i\in[p]}\enVert[2]{(\widehat{\bSigma}_{\bZ,\widehat{T}_i,\widehat{T}_i})^{-1}-(\bSigma_{\bZ,\widehat{T}_i,\widehat{T}_i})^{-1}}_{\ell_2}
\lesssim_\P
\max_{i\in[p]}\enVert[1]{(\widehat{\bSigma}_{\bZ}-\bSigma_{\bZ})_{\widehat T_i,\widehat T_i}}_{\ell_2}
\overset{\P}{\to}0.
\]
It then follows from the triangle inequality and Assumption \ref{assu:CovarianceZ} that, with probability approaching to one,
\begin{align*}
\max_{i\in[p]}\enVert[2]{\del[1]{(\bfX_{\widehat T_i})^\top \bfX_{\widehat T_i}/n}^{-1}}_{\ell_2}
\leqslant
2\Lambda_{\min}^{-1}(\bSigma_{\bZ})\leqslant2/d.
\end{align*}
We also see that
\begin{align*}
\enVert[1]{\widehat{\bSigma}_{\bZ,\widehat{T}_i,\widehat{T}_i}}_{\ell_2}
&\leqslant 
\enVert[1]{\bSigma_{\bZ,\widehat{T}_i,\widehat{T}_i}}_{\ell_2}
+\enVert[1]{\del[1]{\widehat{\bSigma}_{\bZ}-\bSigma_{\bZ}}_{\widehat{T}_i,\widehat{T}_i}}_{\ell_2}\tag{triangle inequality}\\
&=
\Lambda_{\max}\del[1]{\bSigma_{\bZ,\widehat{T}_i,\widehat{T}_i}}
+\enVert[1]{\del[1]{\widehat{\bSigma}_{\bZ}-\bSigma_{\bZ}}_{\widehat{T}_i,\widehat{T}_i}}_{\ell_2}\tag{symmetric p.s.d.}\\
&\lesssim_\P \phi_{\max}(s\ln(n), \bSigma_{\bZ})\tag{$\max_{i\in[p]}|\widehat{T}_i|\lesssim_\P s$ and \eqref{eq:sigmahat2norm}}
\end{align*}
uniformly over $i\in[p]$, implying $\max_{i\in[p]}\enVert[0]{(\bfX_{\widehat T_i})^\top/\sqrt{n}}_{\ell_2}\lesssim_\P1$.
Since $\max_{i\in[p]}|T_i|\leqslant s$, an argument parallel to that leading to \eqref{eq:sigmahat2norm} shows that also $\enVert[0]{\del[0]{\widehat{\bSigma}_{\bZ}-\bSigma_{\bZ}}_{T_i,T_i}}_{\ell_2}\to_\P0$. Since $\widehat{T}_{i}^{c}\cap T_i\subseteq T_i$, we have $\enVert[0]{\bfX_{\widehat T_i^c\cap T_i}/\sqrt{n}}_{\ell_2}\leqslant \enVert[0]{\bfX_{T_i}/\sqrt{n}}_{\ell_2}$. An argument similar to that yielding the previous display therefore shows $\max_{i\in[p]}\enVert[0]{\bfX_{\widehat T_i^c\cap T_i}/\sqrt{n}}_{\ell_2}\lesssim_\P1$. Combining the above findings, submultiplicativity of the operator norm~$\enVert[0]{\cdot}_{\ell_2}$ implies
\begin{align*}
\enVert[2]{\del[1]{(\bfX_{\widehat T_i})^\top \bfX_{\widehat T_i}}^{-1}(\bfX_{\widehat T_i})^\top \bfX_{\widehat T_i^c\cap T_i}\bbeta_{0i,\widehat T_i^c\cap T_i}}_{\ell_2}
\lesssim_{\mathrm{P}}
\enVert[1]{\bbeta_{0i,\widehat T_i^c\cap T_i}}_{\ell_2}
\lesssim_{\mathrm{P}}
\sqrt{\frac{s\ln\left(pn\right)}{n}},
\end{align*}
uniformly in~$i\in[p]$.

Next, to bound~$\del[1]{(\bfX_{\widehat T_i})^\top \bfX_{\widehat T_i}}^{-1}(\bfX_{\widehat T_i})^\top \bu_i$, note that by definition of an operator norm,
\begin{align*}
	\enVert[2]{\del[1]{(\bfX_{\widehat T_i})^\top \bfX_{\widehat T_i}/n}^{-1}(\bfX_{\widehat T_i})^\top \bu_i/n}_{\ell_2}
	\leqslant
	\enVert[2]{\del[1]{(\bfX_{\widehat T_i})^\top \bfX_{\widehat T_i}/n}^{-1}}_{\ell_2}\enVert[2]{(\bfX_{\widehat T_i})^\top \bu_i/n}_{\ell_2}.
\end{align*}
It has already been established that the first factor on the right-hand side of the previous display is bounded in probability. Concerning the second factor, by Lemmas~\ref{lem:Deviation-Bound} and~\ref{lem:Ideal-Penalty-Loading-Control} as well as~$\lambda_n^\ast\lesssim \sqrt{n\ln(pn)}$
\begin{align*}
\enVert[2]{(\bfX_{\widehat T_i})^\top \bu_i/n}_{\ell_2}
\leqslant
\sqrt{|\widehat T_i|}	\enVert[1]{(\bfX_{\widehat T_i})^\top \bu_i/n}_{\ell_\infty}
\leqslant
\frac{\sqrt{|\widehat T_i|}}{2}\enVert[1]{\widehat{\bm{\Upsilon}}_i^{0}}_{\ell_\infty}\enVert[1]{\bS_{n,i}}_{\ell_\infty}
\lesssim_{\mathrm{P}}
\sqrt{\frac{s\ln\left(pn\right)}{n}},
\end{align*}
uniformly in~$i\in[p]$, such that by \eqref{eq:generic_refitting_decomposition_i} we get
\begin{align*}
\enVert[1]{\widetilde\bbeta_{i,\widehat T_i}-\bbeta_{0i,\widehat T_i}}_{\ell_2}\lesssim_{\mathrm{P}}
\sqrt{\frac{s\ln\left(pn\right)}{n}}, 
\end{align*}
uniformly in~$i\in[p]$. Thus, since~$\widetilde\bbeta_{i,\widehat T_i^c\cap T_i^c}=\bbeta_{0i,\widehat T_i^c\cap T_i^c}=\bm{0}$, $\widetilde\bbeta_{i,\widehat T_i^c\cap T_i}=\bm{0}$, and $\enVert[0]{\bbeta_{0i,\widehat T_i^c\cap T_i}}_{\ell_2}\lesssim_{\mathrm{P}}
\sqrt{s\ln\left(pn\right)/n}$, cf.~\eqref{eq:normrelnonsel},
\begin{align*}
\enVert[1]{\widetilde\bbeta_{i}-\bbeta_{0i}}_{\ell_2}
\leqslant
\enVert[1]{\widetilde\bbeta_{i,\widehat T_i}-\bbeta_{0i,\widehat T_i}}_{\ell_2}
+
\enVert[1]{\widetilde\bbeta_{i,\widehat T_i^c\cap T_i}-\bbeta_{0i,\widehat T_i^c\cap T_i}}_{\ell_2}
\lesssim_{\mathrm{P}}
\sqrt{\frac{s\ln\left(pn\right)}{n}},
\end{align*}
uniformly in~$i\in[p]$. By similar reasoning and using this $\ell_2$ rate, for the $\ell_1$ norm we get
\begin{align*}
\enVert[1]{\widetilde\bbeta_{i}-\bbeta_{0i}}_{\ell_1}
&\leqslant
\enVert[1]{\widetilde\bbeta_{i,\widehat T_i}-\bbeta_{0i,\widehat T_i}}_{\ell_1}
+
\enVert[1]{\widetilde\bbeta_{i,\widehat T_i^c\cap T_i}-\bbeta_{0i,\widehat T_i^c\cap T_i}}_{\ell_1}\\
&\leqslant
\sqrt{|\widehat{T}_i|}\enVert[1]{\widetilde\bbeta_{i,\widehat T_i}-\bbeta_{0i,\widehat T_i}}_{\ell_2}
+
\sqrt{s}\enVert[1]{\bbeta_{0i,\widehat T_i^c\cap T_i}}_{\ell_2}\tag{Cauchy-Schwarz}\\
&\lesssim_{\mathrm{P}}
\sqrt{\frac{s^2\ln\left(pn\right)}{n}}.
\end{align*}
Finally to upper bound the error of~$\widetilde\bbeta_i$ the prediction norm~$\enVert[0]{\cdot}_{2,n}$, observe that by \eqref{eq:selectedsetcardinality} it holds for~$\widehat{A}_i:=\widehat{T}_i\cup T_i$ that~$|\widehat{A}_i|\leqslant |\widehat T_i| + |T_i| \lesssim_{\mathrm{P}} s$ uniformly in $i\in[p]$. Thus, since $\phi_{\max}(s\ln(n),\bSigma_{\bZ})$ is eventually bounded and by Lemma~\ref{assu:RowSparsity} as well as Assumption~\ref{assu:RowSparsity}, by the triangle inequality and symmetric positive semi-definiteness,
\begin{align*}
\Lambda_{\max}\del[1]{\widehat{\bSigma}_{\bZ,\widehat{A}_i,\widehat{A}_i}}
&\leqslant 
\Lambda_{\max}\del[1]{\bSigma_{\bZ,\widehat{A}_i,\widehat{A}_i}}
+\enVert[1]{\del[1]{\widehat{\bSigma}_{\bZ}-\bSigma_{\bZ}}_{\widehat{A}_i,\widehat{A}_i}}_{\ell_2}\\
&\lesssim_{\mathrm{P}}
\phi_{\max}(s\ln(n),\bSigma_{\bZ})+s\max_{(i,j)\in [pq]^2}\envert[1]{(\widehat{\bSigma}_{\bZ}-\bSigma_{\bZ})_{i,j}},
\end{align*}
implying that $\Lambda_{\max}\del[0]{\widehat{\bSigma}_{\bZ,\widehat{A}_i,\widehat{A}_i}}$ is bounded in probability uniformly in $i\in[p]$. Thus,
\begin{align*}
\enVert[1]{\widetilde\bbeta_{i}-\bbeta_{0i}}_{2,n}^2
&=
\frac{1}{n}\sum_{t=1}^n\sbr[1]{\bZ_{t-1}^\top\del[1]{\widetilde\bbeta_{i}-\bbeta_{0i}}}^2\\
&=
\frac{1}{n}\sum_{t=1}^n\sbr[2]{\bZ_{t-1,\widehat{A}_i}^\top\del[1]{\widetilde\bbeta_{i}-\bbeta_{0i}}_{\widehat{A}_i}}^2\tag{$\del[1]{\widetilde\bbeta_{i}-\bbeta_{0i}}_{\widehat{A}_i^c}=\mathbf{0}$}\\
&=
\del[1]{\widetilde\bbeta_{i}-\bbeta_{0i}}_{\widehat{A}_i}^\top
\del[1]{\bfX_{\widehat{A}_i}^\top\bfX_{\widehat{A}_i}/n}\del[1]{\widetilde\bbeta_{i}-\bbeta_{0i}}_{\widehat{A}_i}\\
&\leqslant
\Lambda_{\max}(\widehat{\bSigma}_{\bZ,\widehat{A}_i,\widehat{A}_i})\enVert[1]{\widetilde\bbeta_{i}-\bbeta_{0i}}_{\ell_2}^2\\
&\lesssim_{\mathrm{P}}	
	\frac{s\ln\left(pn\right)}{n},
\end{align*}
uniformly in $i\in[p]$.
\end{proof}

\subsection{Proof of Theorem~\ref{thm:postLasso}}

Abbreviate~$\widehat\phi_{\max}(m):=\phi_{\max}(m,\widehat{\bSigma}_{\bZ})$, which is the empirical analog of~$\phi_{\max}(m,\bSigma_{\bZ})$. To apply Theorem~\ref{thm:postgeneralestim}, the following lemma upper bounds the number of variables~$\widehat{m}_i\del[1]{\widehat{\bbeta}_i(\lambda,\widehat{\bUpsilon}_i)}:=|\widehat{T}_i\del[1]{\widehat{\bbeta}_i(\lambda,\widehat{\bUpsilon}_i)}\setminus T_i|$ wrongly included by the Lasso. The proof of the lemma essentially follows from Lemma 10 in~\cite{belloni_sparse_2012} which, in turn, is similar to arguments in~\cite{belloni_least_2013}.
\begin{lem}\label{lem:Lasso_overselection}
Under the assumptions of Lemma~\ref{lem:MasterLemma} the Lasso estimates $\widehat{\bbeta}_i:=\widehat{\bbeta}_i(\lambda,\widehat\bUpsilon_i),i\in[p]$,
arising from $\lambda$ and $\{\widehat{\bUpsilon}_{i}\}_{i=1}^p$ via (\ref{eq:LASSOVector}) satisfy 	
\begin{align*}
\max_{i\in[p]}\widehat{m}_i\del[1]{\widehat{\bbeta}_i(\lambda,\widehat{\bUpsilon}_i)}
\leqslant
4 s\cdot\min_{m\in\widehat{\mathcal{M}}}\widehat\phi_{\max}(m)\sbr[3]{\frac{\widehat{c}_0\widehat{\mu}_{0}}{\widehat\kappa_{\widehat c_{0}\widehat{\mu}_{0}}}}^2,
\end{align*}
where we define
\begin{align}\label{eq:Mdef}
\widehat{\mathcal{M}}:=\cbr[3]{m\in\N: m>8s\cdot\widehat\phi_{\max}(m)\sbr[3]{\frac{\widehat{c}_0\widehat{\mu}_{0}}{\widehat\kappa_{\widehat c_{0}\widehat{\mu}_{0}}}}^2}.
\end{align}
\end{lem}	
\begin{rem}
As $m\mapsto\widehat{\phi}_{\max}(m)$ is (weakly) increasing, and equal to $\Lambda_{\max}(\widehat{\bSigma}_{\bZ})$ for $m\geqslant pq$, the set $\widehat{\mathcal{M}}$ is non-empty, and $\min\{\widehat{\phi}_{\max}(m):m\in\widehat{\mathcal{M}}\}$ exists.\hfill$\diamondsuit$
\end{rem}
\begin{proof}[\sc{Proof of Lemma \ref{lem:Lasso_overselection}}]
Let~$i\in[p]$. Throughout the proof, we suppress the dependence of~$\widehat{T}_i\del[1]{\widehat{\bbeta}_i(\lambda,\widehat{\bUpsilon}_i)}$ and~$\widehat{m}_i\del[1]{\widehat{\bbeta}_i(\lambda,\widehat{\bUpsilon}_i)}$ on~$\widehat{\bbeta}_i(\lambda,\widehat{\bUpsilon}_i)$ and write~$\widehat M_i:=\widehat{T}_i\setminus T_i$. In case~$\widehat M_i=\emptyset$ we are done. Thus, assume~$\widehat M_i\neq\emptyset$. From the first-order conditions for optimality of~$\widehat\bbeta_i $ one has
\begin{align*}
2\sum_{t=1}^{n}\sbr[2]{\widehat{\bm{\Upsilon}}_i^{-1}\bZ_{t-1}\left(Y_{t,i}-\bZ_{t-1}^{\top}\widehat\bbeta_i\right)}_{\widehat{M_i}}
=
\lambda \sign\del[1]{\widehat\bbeta_{i,\widehat{
M_i}}},
\end{align*}
where the sign function is applied entrywise. Therefore, writing~$\bu_i=(\varepsilon_{1,i},\hdots,\varepsilon_{n,i})^\top$ and inserting~$Y_{t,i}=\bZ_{t-1}^\top \bbeta_{0i}+\varepsilon_{t,i}$, it follows that
\begin{align}
\sqrt{\widehat m_i}\lambda 
&=
2\enVert[1]{\del[1]{\widehat{\bm{\Upsilon}}_i^{-1}\bfX^\top \bfX(\bbeta_{0i}-\widehat{\bbeta}_i)+\widehat{\bm{\Upsilon}}_i^{-1}\bfX^\top\bu_i}_{\widehat{
M_i}}}_{\ell_2}\notag\\ 
&\leqslant 
2\enVert[1]{\del[1]{\widehat{\bm{\Upsilon}}_i^{-1}\bfX^\top \bfX(\bbeta_{0i}-\widehat{\bbeta}_i)}_{\widehat{
M_i}}}_{\ell_2}+2\enVert[1]{(\widehat{\bm{\Upsilon}}_i^{-1}\bfX^\top\bu_i)_{\widehat{
M_i}}}_{\ell_2}.\label{eq:false_pos_intermediate_upper_bnd} 	 
\end{align}
Next, since~$\enVert[1]{\widehat{\bm{\Upsilon}}_i^{-1}\widehat{\bm{\Upsilon}}_i^{0}}_{\ell_2}\leqslant 1/\widehat{\ell}$ and~$\enVert[1]{(\widehat{\bm{\Upsilon}}_i^{0})^{-1}}_{\ell_2}\leqslant 1/\widehat{\upsilon}_{\min}^0$, we see that
\begin{align*}
\enVert[1]{\del[1]{\widehat{\bm{\Upsilon}}_i^{-1}\bfX^\top \bfX(\bbeta_{0i}-\widehat{\bbeta}_i)}_{\widehat{
M_i}}}_{\ell_2}
&\leqslant
\frac{1}{\widehat{\ell}\widehat{\upsilon}_{\min}^0}\enVert[1]{\del[1]{\bfX^\top \bfX(\bbeta_{0i}-\widehat{\bbeta}_i)}_{\widehat{
M_i}}}_{\ell_2}\\
&=\frac{1}{\widehat{\ell}\widehat{\upsilon}_{\min}^0}
\sup_{\substack{\bdelta\in\R^{pq}:\\\enVert[0]{\bdelta}_{\ell_0}\leqslant\widehat{m}_i,\\\enVert[0]{\bdelta}_{\ell_2}\leqslant1}}\envert[1]{\bdelta^{\top}\bfX^\top \bfX(\bbeta_{0i}-\widehat{\bbeta}_i)}\\
&\leqslant
\frac{1}{\widehat{\ell}\widehat{\upsilon}_{\min}^0}
\sup_{\bdelta\in\mathfrak{D}(\widehat{m}_i)}\enVert[1]{\bfX\bdelta}_{\ell_2}\enVert[1]{ \bfX(\bbeta_{0i}-\widehat{\bbeta}_i)}_{\ell_2}\tag{Cauchy-Schwarz}\\
&\leqslant
\frac{1}{\widehat{\ell}\widehat{\upsilon}_{\min}^0}n\sqrt{\widehat{\phi}_{\max}\del[1]{\widehat{m}_i}}\enVert[1]{\widehat{\bbeta}_i-\bbeta_{0i}}_{2,n}\\
&\leqslant
\frac{1}{\widehat{\ell}\widehat{\upsilon}_{\min}^0}\sqrt{\widehat{\phi}_{\max}\del[1]{\widehat{m}_i}}\left(\widehat u+\frac{1}{c}\right)\frac{\lambda\widehat{\upsilon}_{\max}^{0}\sqrt{s}}{\widehat\kappa_{\widehat c_{0}\widehat{\mu}_{0}}},
\end{align*}
where the last inequality follows from the prediction error bound in Lemma~\ref{lem:MasterLemma}. In addition, because~$c\max_{i\in[p]}\Vert \bS_{n,i}\Vert_{\ell_{\infty}} \leqslant \lambda/n$, we see that
\begin{align*}
2\enVert[1]{(\widehat{\bm{\Upsilon}}_i^{-1}\bfX^\top\bu_i)_{\widehat{
M_i}}}_{\ell_2}
\leqslant
n\frac{\sqrt{\widehat{m}_i}}{\widehat{\ell}}\enVert{\bS_{n,i}}_{\ell_\infty}
\leqslant 
\frac{\sqrt{\widehat{m}_i}\lambda}{c\widehat{\ell}}.
\end{align*}
Gathering these upper bounds into \eqref{eq:false_pos_intermediate_upper_bnd} and rearranging yields
\begin{align*}
\sqrt{\widehat{m}_i}	
\leqslant
\left.\frac{2}{\widehat{\ell}\widehat{\upsilon}_{\min}^0}\sqrt{\widehat{\phi}_{\max}\del[1]{\widehat{m}_i}}\left(\widehat u+\frac{1}{c}\right)\frac{\widehat{\upsilon}_{\max}^{0}\sqrt{s}}{\widehat\kappa_{\widehat c_{0}\widehat{\mu}_{0}}}\middle/\del[3]{1-\frac{1}{c\widehat{\ell}}}\right.
=
2 \widehat{c}_0\sqrt{\widehat{\phi}_{\max}\del[1]{\widehat{m}_i}}\frac{\widehat{\mu}_{0}\sqrt{s}}{\widehat\kappa_{\widehat c_{0}\widehat{\mu}_{0}}}
\end{align*}
where we use~$\del[1]{\widehat u+\frac{1}{c}}/\del[1]{1-\frac{1}{c\widehat{\ell}}}=\widehat{c}_0\widehat{\ell}$. Squaring both sides of the inequality, we arrive at
\[
\widehat{m}_i
\leqslant
4s\cdot\widehat{\phi}_{\max}\del[1]{\widehat{m}_i}\sbr[3]{ \frac{\widehat{c}_0\widehat{\mu}_{0}}{\widehat\kappa_{\widehat c_{0}\widehat{\mu}_{0}}}}^2.
\]
Seeking a contradiction, suppose next that there is an element~$m\in\widehat{\mathcal{M}}$, such that~$\widehat{m}_i>m$. Then by monotonicity and sublinearity of maximal sparse eigenvalues of positive semi-definite matrices \citep[Lemma 9]{belloni_sparse_2012}, we get
\begin{align*}
\widehat{\phi}_{\max}\del[1]{\widehat{m}_i}
=\widehat{\phi}_{\max}\del[1]{\del[0]{\widehat{m}_i/m}m}
&\leqslant \widehat{\phi}_{\max}\del[1]{\lceil\widehat{m}_i/m\rceil m}\tag{monotonicity}\\
&\leqslant \lceil\widehat{m}_i/m\rceil\widehat{\phi}_{\max}\del[0]{m}\tag{sublinearity}\\
&\leqslant 2\del[1]{\widehat{m}_i/m}\widehat{\phi}_{\max}\del[0]{m}\tag{$a\in[1,\infty)\implies\lceil a\rceil\leqslant 2a$}.
\end{align*}
The previous two displays combine to yield
\begin{align*}
\widehat{m}_i
\leqslant 
8s\del[1]{\widehat{m}_i/m}\widehat{\phi}_{\max}\del[0]{m}\sbr[3]{ \frac{\widehat{c}_0\widehat{\mu}_{0}}{\widehat\kappa_{\widehat c_{0}\widehat{\mu}_{0}}}}^2.
\end{align*}
This inequality rearranges to
\[
m\leqslant 
8s\cdot\widehat{\phi}_{\max}\del[0]{m}\sbr[3]{ \frac{\widehat{c}_0\widehat{\mu}_{0}}{\widehat\kappa_{\widehat c_{0}\widehat{\mu}_{0}}}}^2,
\]
which contradicts $m\in\widehat{\mathcal M}$. Conclude that $\widehat{m}_i\leqslant m$ for all $m\in\widehat{\mathcal M}$, so that per monotonicity of maximum sparse eigenvalues, we get $\widehat{\phi}_{\max}\del[1]{\widehat{m}_i}\leqslant \min\{\widehat{\phi}_{\max}\del[0]{m}:m\in\widehat{\mathcal M}\}$ and, thus,
\[
\widehat{m}_i
\leqslant
4s\cdot\min_{m\in\widehat{\mathcal M}}\widehat{\phi}_{\max}\del[0]{m}\sbr[3]{ \frac{\widehat{c}_0\widehat{\mu}_{0}}{\widehat\kappa_{\widehat c_{0}\widehat{\mu}_{0}}}}^2.
\]
The asserted claim now follows from the right-hand side not depending on~$i\in[p].$
\end{proof}

Choosing~$\lambda=\lambda_n^*$ and any asymptotically valid loadings~$\widehat\bUpsilon_i$ allows one to upper bound the over-selection by the Lasso if one also imposes~$\limsup_{n\to\infty}\phi_{\max}(s\ln(n), \bSigma_{\bZ})<\infty$.
\begin{cor}\label{cor:Lasso_overselection}
Let Assumptions \ref{assu:Innovations}, \ref{assu:Companion}, \ref{assu:CovarianceZ} and \ref{assu:RowSparsity} hold, and let~$\limsup_{n\to\infty}\phi_{\max}(s\ln(n), \bSigma_{\bZ})<\infty$. Then Lasso estimators $\widehat\bbeta_i(\lambda_n^\ast,\widehat\bUpsilon_i),i\in[p]$, based on the penalty level $\lambda_n^\ast$ in \eqref{eq:PenaltyLevelPractice} and any asymptotically valid loadings~$\{\widehat\bUpsilon_i\}_{i\in[p]}$, satisfy~$\max_{i\in[p]}\widehat{m}_i\del[1]{\widehat{\bbeta}_i(\lambda_n^*,\widehat{\bUpsilon}_i)}\lesssim_{\mathrm{P}} s $.
\end{cor}

\begin{proof}
By asymptotic validity of~$\widehat\bUpsilon_i$ as well as Lemmas~\ref{lem:Ideal-Penalty-Loading-Control} and \ref{lem:Restricted-Eigenvalue-Bound} it follows that for the constant
\[
L:=2\cdot\frac{(uc+1)/(c-1)\cdot (\overline\upsilon^0/\underline\upsilon^0)}{\sqrt{d/2}}\in(0,\infty),
\]
one has
\begin{align}\label{eq:overselaux1}
\P\del[3]{\frac{\widehat{c}_0\widehat{\mu}_{0}}{\widehat\kappa_{\widehat c_{0}\widehat{\mu}_{0}}}\leqslant L}
\to 1,
\end{align}
with~$u$ defined in~\eqref{eq:AsymptoticallyValidPenaltyLoadings2} and the remaining quantities as defined in the quoted lemmas. 
Furthermore, let~$\bdelta\in\mathfrak{D}(m)$ [see~\eqref{eq:Dm}] and abbreviate~$I_{\bdelta} := \supp(\bdelta)= \cbr[0]{j\in[pq]:\delta_j\neq 0}$. Then
\begin{align*}
\bdelta ^\top \widehat{\bSigma}_{\bZ}\bdelta
\leqslant
\bdelta^\top 	\bSigma_{\bZ}\bdelta+\envert[1]{\bdelta ^\top (\widehat{\bSigma}_{\bZ}-\bSigma_{\bZ})\bdelta}
\leqslant 
\phi_{\max}(m,\bSigma_{\bZ})+\envert[1]{\bdelta_{I_\delta} ^\top (\widehat{\bSigma}_{\bZ}-\bSigma_{\bZ})_{I_\delta,I_\delta}\bdelta_{I_\delta}}.
\end{align*}
By the Cauchy-Schwarz inequality and the norm equivalence relations table on p.~314 in \cite{hornjohnson}, which shows that the~$\ell_2$ operator norm of a square matrix is bounded from above by the (row/column) dimension times its maximum absolute entry, the second term on the right-hand side is bounded by
\begin{align*}
\enVert[1]{\bdelta_{I_\delta}}_{\ell_2}	\enVert[1]{(\widehat{\bSigma}_{\bZ}-\bSigma_{\bZ})_{I_\delta,I_\delta}\bdelta_{I_\delta}}_{\ell_2}
\leqslant
\enVert[1]{(\widehat{\bSigma}_{\bZ}-\bSigma_{\bZ})_{I_\delta,I_\delta}}_{\ell_2}
\leqslant
m\max_{(i,j)\in [pq]^2}\envert[1]{(\widehat{\bSigma}_{\bZ}-\bSigma_{\bZ})_{i,j}}.
\end{align*}
Therefore, combining the two previous displays with Lemma~\ref{lem:Covariance-Consistency}, 
\begin{align*}
\widehat{\phi}_{\max}(m)
&\leqslant
\phi_{\max}(m,\bSigma_{\bZ})+m\max_{(i,j)\in [pq]^2}\envert[1]{(\widehat{\bSigma}_{\bZ}-\bSigma_{\bZ})_{i,j}}\\
&\lesssim_{\P} 
\phi_{\max}(m,\bSigma_{\bZ})+m\sqrt{\frac{(\ln(pn))^{1/\tau+\widetilde{C}(2,\alpha)}}{n}}.
\end{align*}
By assumption there is a~$K<\infty$, such that for any fixed~$C\in\N$ and defining~$m_C:=Cs$, it eventually holds that~$\phi_{\max}(m_C, \bSigma_{\bZ})\leqslant \phi_{\max}(s\ln(n), \bSigma_{\bZ})\leqslant K$. Thus, by Assumption~\ref{assu:RowSparsity} and the previous display,
\begin{align}\label{eq:auxbound}
\P\del[1]{\widehat{\phi}_{\max}(m_C)\leqslant 2K}\to 1\qquad \text{for any fixed }C\in\N.
\end{align}
Hence, in combination with~\eqref{eq:overselaux1}, specifying the constant~$\check C:=17\lceil KL^2\rceil$, the implied~$m_{\check C}$ must satisfy
\begin{align*}
\P\del[2]{8s\cdot\widehat{\phi}_{\max}(m_{\check C})\sbr[3]{\frac{\widehat{c}_0\widehat{\mu}_{0}}{\widehat\kappa_{\widehat c_{0}\widehat{\mu}_{0}}}}^2
< 	
m_{\check C}}
\geqslant
\P\del[2]{8s\cdot\widehat{\phi}_{\max}(m_{\check C})\sbr[3]{\frac{\widehat{c}_0\widehat{\mu}_{0}}{\widehat\kappa_{\widehat c_{0}\widehat{\mu}_{0}}}}^2 \leqslant 16sKL^2}\to 1,
\end{align*}
which means that~$\P\del[0]{m_{\check C}\in\widehat{\mathcal{M}}}\to 1$, where~$\widehat{\mathcal{M}}$ is defined in~\eqref{eq:Mdef} of Lemma~\ref{lem:Lasso_overselection}. Thus, by Lemma~\ref{lem:Lasso_overselection} with $\lambda=\lambda_n^\ast$ (the lemma assumptions being satisfied with probability tending to one under the imposed Assumptions~\ref{assu:Innovations}--\ref{assu:RowSparsity}) as well as using~\eqref{eq:overselaux1} and \eqref{eq:auxbound} (with~$C=\check C$), it follows that
\[
\max_{i\in[p]}\widehat{m}_i\del[1]{\widehat{\bbeta}\del[1]{\lambda_n^\ast,\widehat{\bUpsilon}_i}}
\lesssim_{\P} 
s\cdot \min_{m\in\widehat{\mathcal M}}\widehat{\phi}_{\max}\del[0]{m}
\lesssim_{\P} s\cdot \widehat{\phi}_{\max}\del[0]{m_{\check C}}
\lesssim_{\P} s.
\]
\end{proof}

\begin{proof}[\sc{Proof of Theorem~\ref{thm:postLasso}}]
Fix~$K\in\N_0$. By Theorem~\ref{thm:Rates-for-LASSO-data-driven-loadings}, Lasso estimators $\widehat\bbeta_i(\lambda_n^\ast,\widehat\bUpsilon_i^{(K)})$, $i\in[p]$, arising from the data-driven penalization in Algorithm \ref{alg:Data-Driven-Penalization} satisfy the rates in (\ref{eq:lasso_rates_1})--(\ref{eq:lasso_rates_3}). By Lemma~\ref{lem:AsymptoticValidityDataDrivenPenaltyLoadings}, the loadings~$\{\widehat\bUpsilon_i^{(K)}
\}_{i\in[p]}$ are asymptotically valid, so that Corollary \ref{cor:Lasso_overselection} yields $\max_{i\in[p]}|\widehat{T}_i(\lambda_n^*,\widehat\bUpsilon_i^{(K)})\setminus T_i|\lesssim_\P s$. Thus, it follows from Theorem~\ref{thm:postgeneralestim} that post-Lasso estimators~$\widetilde\bbeta_i(\lambda_n^*,\widehat{\bUpsilon}_i^{(K)}),i\in[p]$, satisfy the rates in~\eqref{eq:lasso_rates_1}--\eqref{eq:lasso_rates_3}.
\end{proof}

\subsection{Using least squares to construct weights for the Lasso in every step}\label{sec:OLSforweights}
The following algorithm defines an alternative version of the post-Lasso. It differs from the one studied in Theorem~\ref{thm:postLasso} by constructing the loadings of the Lasso in every step by the residuals of the post-Lasso rather than the ones of the Lasso itself.
\begin{mdframed}
\begin{lyxalgorithm}
[\textbf{Data-Driven Penalization with Refitting}]\label{alg:Data-Driven-Penalization-with-Refitting}\ \\
\textbf{Initialize}: Specify the penalty level in (\ref{eq:LASSOVector}) as
\begin{equation*}
\lambda^{\ast}_n :=2c\sqrt{n}\Phi^{-1}\big(1-\gamma_n/(2p^{2}q)\big),
\end{equation*}
and specify the initial penalty loadings as 
\begin{equation*}
\check{\upsilon}_{i,j}^{\left(0\right)} := \sqrt{\frac{1}{n}\sum_{t=1}^{n}Y_{t,i}^{2}Z_{t-1,j}^{2}},\quad i\in[p],\quad j\in[pq].
\end{equation*}
Use $\lambda^{\ast}_n$ and $\check{\bUpsilon}_{i}^{\left(0\right)}:=\mathrm{diag}(\check{\upsilon}_{i,1}^{\left(0\right)},\dotsc,\check{\upsilon}_{i,pq}^{(0)})$
to compute an estimate $\check{\bbeta}_{i}^{(0)}:=\allowbreak\check{\bbeta}_{i}(\lambda^{\ast}_n,\check{\bUpsilon}_{i}^{(0)})$
via post-Lasso (\ref{eq:postestim}) for each $i\in[p]$. Store the residuals $\check{\varepsilon}_{t,i}^{\left(0\right)}:=Y_{t,i}-\bZ_{t-1}^{\top}\check{\bbeta}_{i}^{\left(0\right)},t\in[n],i\in[p]$, and set $k=1$.

\noindent\textbf{Update:}
While $k\leqslant K$, specify the penalty loadings as
\begin{equation*}
\check{\upsilon}_{i,j}^{\left(k\right)}:=\sqrt{\frac{1}{n}\sum_{t=1}^{n}\check{\varepsilon}_{t,i}^{\left(k-1\right)2}Z_{t-1,j}^{2}},\quad i\in[p],\quad j\in[pq].
\end{equation*}
Use $\lambda_n^{\ast}$ and $\check{\bUpsilon}_{i}^{\left(k\right)}:=\mathrm{diag}(\check{\upsilon}_{i,1}^{\left(k\right)},\dotsc,\check{\upsilon}_{i,pq}^{\left(k\right)})$
to compute an estimate $\check{\bbeta}_{i}^{\left(k\right)}:=\check{\bbeta}_{i}(\lambda_n^{\ast},\check{\bUpsilon}_{i}^{\left(k\right)})$
via post-Lasso (\ref{eq:postestim}) for each $i\in[p]$. Store the residuals $\check{\varepsilon}_{t,i}^{\left(k\right)}:=Y_{t,i}-\bZ_{t-1}^{\top}\check{\bbeta}_{i}^{\left(k\right)},t\in[n],i\in[p]$,
and increment $k\leftarrow k+1.$
\end{lyxalgorithm}
\end{mdframed}

\noindent Algorithm \ref{alg:Data-Driven-Penalization-with-Refitting} could include unpenalized intercepts via an initial demeaning of variables---see Section \ref{subsec:IncludingIntercepts} for the necessary adjustments.

\begin{thm}[\textbf{Convergence Rates for Post-Lasso based on Data-Driven Penalization with Refitting}]\label{thm:OLSweights}
Let Assumptions \ref{assu:Innovations}, \ref{assu:Companion}, \ref{assu:CovarianceZ} and \ref{assu:RowSparsity} hold and fix~$K\in\N_0$. Furthermore, assume that~$\limsup_{n\to\infty}\phi_{\max}(s\ln(n), \bSigma_{\bZ})<\infty$. Then post-Lasso estimators~$\check{\bbeta}_{i}(\lambda_n^\ast,\check\bUpsilon_i^{(K)})$, $i\in[p]$, arising from the data-driven penalization in Algorithm \ref{alg:Data-Driven-Penalization-with-Refitting} satisfy the rates in~\eqref{eq:lasso_rates_1}--\eqref{eq:lasso_rates_3}.
\end{thm}
\begin{proof}		
\textbf{Step 1: Induction start.} First, let~$K=0$. For any pair~$(\lambda, \widehat{\bUpsilon}_i)$  as in~\eqref{eq:LASSOVector} and~\eqref{eq:Qihat} define
\begin{align*}
\mathcal{S}(\lambda, \widehat{\bUpsilon}_i):=\argmin_{\bbeta\in\R^{pq}}\left\{ \widehat Q_i(\bbeta) +\frac{\lambda}{n}\Vert\widehat{\bUpsilon}_{i}\bbeta\Vert_{\ell_{1}}\right\},\qquad i\in[p].
\end{align*}
Because~$\check{\bUpsilon}_i^{(0)}=\widehat{\bUpsilon}_i^{(0)}$, the  set of post-Lasso estimators in Algorithm~\ref{alg:Data-Driven-Penalization-with-Refitting} with~$K=0$ are identical to those studied in Theorem~\ref{thm:postLasso} with~$K=0$, i.e.
\begin{align*}
\bigcup_{\bbeta\in \mathcal{S}(\lambda_n^*, \check{\bUpsilon}_i^{(0)})}\argmin_{\substack{\bm{b}\in\R^{pq}:\\
\supp(\bm{b})\subseteq \supp(\bbeta)}} \widehat{Q}_{i}\del[0]{\bm{b}}
=
\bigcup_{\bbeta\in \mathcal{S}(\lambda_n^*, \widehat{\bUpsilon}_i^{(0)})}\argmin_{\substack{\bm{b}\in\R^{pq}:\\
\supp(\bm{b})\subseteq \supp(\bbeta)}} \widehat{Q}_{i}\del[0]{\bm{b}}.	
\end{align*}
Therefore,~$\check{\bbeta}_{i}(\lambda_n^\ast,\check\bUpsilon_i^{(0)}),i\in[p]$, satisfy the rates in~\eqref{eq:lasso_rates_1}--\eqref{eq:lasso_rates_3} by Theorem~\ref{thm:postLasso} with~$K=0$.

\textbf{Step 2: Induction step.} Next, let~$K\in\N$ and suppose that~$\check{\bbeta}_{i}(\lambda_n^\ast,\check\bUpsilon_i^{(K-1)}),i\in[p]$, satisfy the rates in~\eqref{eq:lasso_rates_1}--\eqref{eq:lasso_rates_3}. We first show that the loadings~$\check\bUpsilon_i^{(K)}=\mathrm{diag}(\check{\upsilon}_{i,1}^{(K)},\dots, \check{\upsilon}_{i,pq}^{(K)}),i\in[p]$, are asymptotically valid. To this end, inspection of Step 2 of the proof of Lemma~\ref{lem:AsymptoticValidityDataDrivenPenaltyLoadings} shows that it suffices to establish that
\begin{align}\label{eq:auzloadingsOLS}
\check\Delta^{(K)}:=\max_{(i,j)\in[p]\times[pq]}\big|(\check{\upsilon}_{i,j}^{(K)})^2-(\widetilde{\upsilon}_{i,j}^{0})^2\big|\overset{\P}{\to}0,
\end{align}
where, as in the proof of Step 2 of Lemma~\ref{lem:AsymptoticValidityDataDrivenPenaltyLoadings}, we denote
\[
\widetilde{\upsilon}_{i,j}^{0}:=\left(\frac{1}{n}\sum_{t=1}^{n}\varepsilon_{t,i}^{2}Z_{t-1,j}^{2}\right)^{1/2},\quad(i,j)\in[p]\times[pq].
\]	
From~\eqref{eq:DeltaTildeUpdPartIBnd} of the proof of Lemma~\ref{lem:AsymptoticValidityDataDrivenPenaltyLoadings}, it can be seen that the only property used of the specific estimators~$\widehat{\bbeta}_{i}^{(0)}:=\widehat{\bbeta}_{i}(\lambda_n^*,\widehat\bUpsilon_i^{(0)}),i\in[p]$, there in establishing the asymptotic validity of the resulting loadings is that they satisfy the prediction error rate in~\eqref{eq:lasso_rates_1}. Since any post-Lasso estimators~$\check{\bbeta}_{i}(\lambda_n^\ast,\check\bUpsilon_i^{(K-1)}),i\in[p]$, also satisfy the rate in~\eqref{eq:lasso_rates_1} by supposition, the resulting loadings~$\check{\bUpsilon}_i^{(K)}$ are asymptotically valid. Hence, any Lasso estimators~$\widehat{\bbeta}(\lambda_n^*,\check{\bUpsilon}_i^{(K)}),i\in[p]$, satisfy (\ref{eq:lasso_rates_1})--(\ref{eq:lasso_rates_3}) by Lemma~\ref{lem:Rates-for-LASSO-any-asymptotically-valid-loadings} and, thus,~$\max_{i\in[p]}\widehat{m}_i\del[1]{\widehat{\bbeta}_i(\lambda_n^*,\check{\bUpsilon}_i^{(K)})}\lesssim_{\mathrm{P}} s $ by Corollary~\ref{cor:Lasso_overselection}. It now follows that post-Lasso estimators~$\check{\bbeta}_{i}(\lambda_n^\ast,\check\bUpsilon_i^{(K)}),i\in[p]$, satisfy the rates in~\eqref{eq:lasso_rates_1}--\eqref{eq:lasso_rates_3} by Theorem~\ref{thm:postgeneralestim}.
\end{proof}

\section{Proof of Theorem~\ref{thm:Rates-for-sqrtLASSO-data-driven-loadings}}   
Lemma~\ref{lem:sqrtLasso} below extends \citet[Theorem 1]{belloni2011square} and its proof for sqrt-Lasso to allow for (data-driven) penalty loadings. This lemma is the sqrt-Lasso version of Lemma \ref{lem:MasterLemma} for Lasso. Define~$s^2_{\varepsilon_i}:=n^{-1}\sum_{t=1}^n\varepsilon_{t,i}^2=\widehat{Q}_i(\bbeta_{0i})$ for~$i\in[p]$. Unless stated otherwise, all quantities in Lemma~\ref{lem:sqrtLasso} and its proof are as defined around Lemma \ref{lem:MasterLemma}.      
\begin{lem}[\textbf{Master Lemma for Square-Root Lasso}]\label{lem:sqrtLasso}
With the constant
$c$ given in (\ref{eq:tuning_c_and_gamma}) and the score vectors
$\{\bS_{n,i}\}_{i\in[p]}$ given by (\ref{eq:def_score_ideal}), suppose that the penalty level satisfies $\lambda/n\geqslant 0.5c\max_{i\in[p]}\Vert \bS_{n,i}\Vert_{\ell_{\infty}}$,
and that for some random scalars $\widehat\ell=\widehat\ell_n$ and $\widehat u=\widehat u_n$, the penalty loadings $\{\widehat{\bUpsilon}_{i}\}_{i\in[p]}$ satisfy 
\begin{equation}\label{eq:pensqrtLasso}
\widehat{\ell}\widehat{\upsilon}_{i,j}^{0}\leqslant \widehat{\upsilon}_{i,j}s_{\varepsilon_i} \leqslant\widehat{u}\widehat{\upsilon}_{i,j}^{0}\quad\text{for all}\quad (i,j)\in[p]\times[pq]
\end{equation}
with $1/c<\widehat\ell\leqslant\widehat u$. Furthermore, suppose that
\begin{equation}\label{eq:sqrt_lasso_bound_qualifier}
\frac{\lambda\widehat u \widehat\upsilon^0_{\max}\sqrt{s}}{n\min_{i\in[p]}s_{\varepsilon_i}\widehat\kappa_{\widehat c_{0}\widehat{\mu}_{0}}}\leqslant\frac{1}{\sqrt{2}}.  
\end{equation}
Then the sqrt-Lasso estimates $\dot{\bbeta}_i:=\dot{\bbeta}_i(\lambda,\widehat\bUpsilon_i),i\in[p]$,
arising from $\lambda$ and $\{\widehat{\bUpsilon}_{i}\}_{i\in[p]}$ via (\ref{eq:srtLasso}) satisfy the error bounds
\begin{align*}
\max_{i\in[p]}\Vert\dot{\bbeta}_{i}-\bbeta_{0i}\Vert_{2,n} & \leqslant  4\del[2]{\widehat u+\frac{1}{c}}\frac{\lambda \widehat\upsilon^0_{\max}\sqrt{s}}{n\widehat\kappa_{\widehat c_{0}\widehat{\mu}_{0}}}\\
\max_{i\in[p]}\Vert\dot{\bbeta}_{i}-\bbeta_{0i}\Vert_{\ell_{1}}& \leqslant4\del[2]{\widehat u+\frac{1}{c}}\left(1+\widehat c_{0}\widehat{\mu}_{0}\right)\frac{\lambda \widehat\upsilon^0_{\max}s}{n\widehat\kappa_{\widehat c_{0}\widehat{\mu}_{0}}^2}\quad\text{and}\\
\max_{i\in[p]}\Vert\dot{\bbeta}_{i}-\bbeta_{0i}\Vert_{\ell_{2}} & \leqslant4\del[2]{\widehat u+\frac{1}{c}}\frac{\lambda \widehat\upsilon^0_{\max}\sqrt{s}}{n\widehat\kappa_{\widehat c_{0}\widehat{\mu}_{0}}^2},
\end{align*}
where $\widehat c_{0}:=(\widehat uc+1)/(\widehat\ell c-1)$, and $\widehat \kappa_{\widehat c_{0}\widehat{\mu}_{0}}$ is defined via (\ref{eq:RestrictedEigenvalue}) with $C=\widehat c_{0}\widehat{\mu}_{0}$.
\end{lem}
\begin{proof}
We first fix~$i\in[p]$ and argue the error bounds for~$\dot\bbeta_i$. Inspection of the proof will reveal that the argument holds uniformly in~$i\in[p]$. Let~$\dot\bdelta_i:=\dot\bbeta_i-\bbeta_{0i}$. By definition of~$\dot{\bbeta}_i$ as a minimizer, the triangle inequality and \eqref{eq:pensqrtLasso},
\begin{align}
\del[1]{\widehat Q_i(\dot\bbeta_i)}^{1/2}-\del[1]{\widehat Q_i(\bbeta_{0i})}^{1/2}
&\leqslant
\frac{\lambda}{n}\Vert\widehat{\bUpsilon}_{i}\bbeta_{0i}\Vert_{\ell_{1}}-\frac{\lambda}{n}\Vert\widehat{\bUpsilon}_{i}\dot\bbeta_{i}\Vert_{\ell_{1}}\notag\\	
&\leqslant
\frac{\lambda}{n}\sbr[2]{\Vert(\widehat{\bUpsilon}_{i}\dot\bdelta_{i})_{T_i}\Vert_{\ell_{1}}-\Vert(\widehat{\bUpsilon}_{i}\dot\bdelta_{i})_{T_i^c}\Vert_{\ell_{1}}}\notag\\
&\leqslant
\frac{\lambda}{ns_{\varepsilon_i}}\sbr[2]{\widehat{u}\Vert(\widehat{\bUpsilon}^0_{i}\dot\bdelta_{i})_{T_i}\Vert_{\ell_{1}}-\widehat{\ell}\Vert(\widehat{\bUpsilon}^0_{i}\dot\bdelta_{i})_{T_i^c}\Vert_{\ell_{1}}}\label{eq:srtLassoaux}.	
\end{align}
In addition, since the negative objective gradient at the true parameters equals
\begin{align*}
-\nabla_{\bbeta}\del[1]{\widehat Q_i(\bbeta)}^{1/2}\Big|_{\bbeta=\bbeta_{0i}}
=
\frac{2n^{-1}\sum_{t=1}^n\bZ_{t-1}\varepsilon_{t,i}}{2\del[1]{\widehat Q_i(\bbeta_{0i})}^{1/2}}
=
\frac{n^{-1}\sum_{t=1}^n\bZ_{t-1}\varepsilon_{t,i}}{ \del[1]{n^{-1}\sum_{t=1}^n\varepsilon_{t,i}^2}^{1/2}}
=
\frac{0.5\cdot\widehat\bUpsilon^0_i\bS_{n,i}}{s_{\varepsilon_i}}
,	
\end{align*}
it follows from convexity of~$\bbeta\mapsto \del[1]{\widehat{Q}_i(\bbeta)}^{1/2}$, the subgradient inequality, and H{\"o}lder's inequality that	%
\begin{align*}
\del[1]{\widehat Q_i(\dot\bbeta_i)}^{1/2}-\del[1]{\widehat Q_i(\bbeta_{0i})}^{1/2}
\geqslant 
\sbr[2]{-\nabla_{\bbeta}\del[1]{\widehat Q_i(\bbeta)}^{1/2}\Big|_{\bbeta=\bbeta_{0i}}}^\top \dot\bdelta_i
\geqslant
-\frac{0.5\cdot\enVert[1]{\bS_{n,i}}_{\ell_\infty}\enVert[1]{\widehat{\bUpsilon}^0_{i}\dot\bdelta_i}_{\ell_1}}{s_{\varepsilon_i}}.
\end{align*}
Therefore, because~$\lambda/n\geqslant 0.5c\max_{i\in[p]}\Vert \bS_{n,i}\Vert_{\ell_{\infty}}$, we arrive at the lower bound 
\begin{align*}
\del[1]{\widehat Q_i(\dot\bbeta_i)}^{1/2}-\del[1]{\widehat Q_i(\bbeta_{0i})}^{1/2}
\geqslant 
-\frac{\lambda}{s_{\varepsilon_i} cn}\del[1]{\enVert[0]{(\widehat{\bUpsilon}^0_{i}\dot\bdelta_{i})_{T_i}}_{\ell_1}+\enVert[0]{(\widehat{\bUpsilon}^0_{i}\dot\bdelta_{i})_{T_i^c}}_{\ell_1}}.	
\end{align*}
Combining this inequality with the upper bound in~\eqref{eq:srtLassoaux} and~$1/c<\widehat\ell\leqslant\widehat u$, it follows that
\begin{align*}
\enVert[0]{(\widehat{\bUpsilon}^0_{i}\dot\bdelta_{i})_{T_i^c}}_{\ell_1}
\leqslant
\frac{\widehat{u}+ 1/c}{\widehat{\ell}-1/c}\enVert[0]{(\widehat{\bUpsilon}^0_{i}\dot\bdelta_{i})_{T_i}}_{\ell_1}.
\end{align*}
As a consequence,
\begin{align*}
\widehat\upsilon^0_{\min}\enVert[0]{\dot\bdelta_{iT_i^c}}_{\ell_1}
\leqslant
\enVert[0]{(\widehat{\bUpsilon}^0_{i}\dot\bdelta_{i})_{T_i^c}}_{\ell_1}
\leqslant
\frac{\widehat{u}+ 1/c}{\widehat{\ell}-1/c}\enVert[0]{(\widehat{\bUpsilon}^0_{i}\dot\bdelta_{i})_{T_i}}_{\ell_1}
\leqslant
\frac{\widehat{u}+ 1/c}{\widehat{\ell}-1/c}\widehat\upsilon^0_{\max}\enVert[0]{\dot\bdelta_{iT_i}}_{\ell_1},	
\end{align*}
and hence
\begin{align*}
\enVert[0]{\dot\bdelta_{iT_i^c}}_{\ell_1}
\leqslant 
\frac{\widehat{u}+ 1/c}{\widehat{\ell}-1/c}\cdot\frac{\widehat\upsilon^0_{\max}}{\widehat\upsilon^0_{\min}}\enVert[0]{\dot\bdelta_{iT_i}}_{\ell_1}
=
\widehat c_{0}\widehat{\mu}_{0}\enVert[0]{\dot\bdelta_{iT_i}}_{\ell_1}.
\end{align*}
Thus,~$\dot{\bdelta}_{i}$ lies in the restricted set $\mathcal{R}_{\widehat c_{0}\widehat{\mu}_{0},T_{i}}$ defined via (\ref{eq:RestrictedSet}) with $C=\widehat c_0\widehat\mu_0$ and $T=T_i$. In addition, the following relations hold:
\begin{align*}
\widehat{Q}_i(\dot\bbeta_i)-\widehat{Q}_i(\bbeta_{0i})
&=
\frac{1}{n}\sum_{t=1}(\bZ_{t-1}^\top\dot{\bdelta}_{i})^{2}-\frac{2}{n}\sum_{t=1}^{n}\varepsilon_{t,i}\bZ_{t-1}^\top\dot{\bdelta}_{i}
\geqslant
\Vert\dot{\bdelta}_{i}\Vert_{2,n}^{2}-\Vert \bS_{n,i}\Vert_{\ell_{\infty}}\Vert\widehat{\bUpsilon}_{i}^{0}\dot{\bdelta}_{i}\Vert_{\ell_{1}},\\
\widehat Q_i(\dot\bbeta_i)-\widehat Q_i(\bbeta_{0i})
&=
\sbr[2]{\del[1]{\widehat Q_i(\dot\bbeta_i)}^{1/2}+\del[1]{\widehat Q_i(\bbeta_{0i})}^{1/2}}\sbr[2]{\del[1]{\widehat Q_i(\dot\bbeta_i)}^{1/2}-\del[1]{\widehat Q_i(\bbeta_{0i})}^{1/2}}. 
\end{align*}
When combined with~\eqref{eq:srtLassoaux}, the previous display yields
\begin{align}
\Vert\dot{\bdelta}_{i}\Vert_{2,n}^{2}
&\leqslant	
\Vert \bS_{n,i}\Vert_{\ell_{\infty}}\Vert\widehat{\bUpsilon}_{i}^{0}\dot{\bdelta}_{i}\Vert_{\ell_{1}}\notag\\
&\qquad+\sbr[2]{\del[1]{\widehat Q_i(\dot\bbeta_i)}^{1/2}+\del[1]{\widehat Q_i(\bbeta_{0i})}^{1/2}}\frac{\lambda}{ns_{\varepsilon_i}}\sbr[2]{\widehat{u}\Vert(\widehat{\bUpsilon}^0_{i}\dot\bdelta_{i})_{T_i}\Vert_{\ell_{1}}-\widehat{\ell}\Vert(\widehat{\bUpsilon}^0_{i}\dot\bdelta_{i})_{T_i^c}\Vert_{\ell_{1}}}.\label{eq:sqrt_lasso_intermediate_bnd_1}
\end{align}
Furthermore, using that~$\dot\bdelta_i\in \mathcal{R}_{\widehat c_{0}\widehat{\mu}_{0},T_{i}}$,
\begin{align}
\enVert[0]{\del[0]{\widehat{\bUpsilon}^0_{i}\dot\bdelta_{i}}_{T_i}}_{\ell_{1}}
\leqslant \widehat\upsilon^0_{\max}\enVert[0]{\dot\bdelta_{i,T_i}}_{\ell_{1}}
\leqslant \widehat\upsilon^0_{\max}\sqrt{s}\enVert[0]{\dot\bdelta_{i,T_i}}_{\ell_{2}}
\leqslant\frac{\widehat\upsilon^0_{\max}\sqrt{s}}{\widehat\kappa_{\widehat c_{0}\widehat{\mu}_{0}}}\enVert[0]{\dot\bdelta_{i}}_{2,n},\label{eq:sqrt_lasso_restricted_set_consequence}
\end{align}
where we use the Cauchy-Schwarz inequality and the definition of $\widehat\kappa_{\widehat c_{0}\widehat{\mu}_{0}}$, cf.~\eqref{eq:RestrictedEigenvalue}. Inequalities \eqref{eq:srtLassoaux} and \eqref{eq:sqrt_lasso_restricted_set_consequence} combine to show
\begin{align*}
\del[1]{\widehat Q_i(\dot\bbeta_i)}^{1/2}
\leqslant
\del[1]{\widehat Q_i(\bbeta_{0i})}^{1/2}+\frac{\lambda\widehat{u}}{ns_{\varepsilon_i}}\Vert(\widehat{\bUpsilon}^0_{i}\dot\bdelta_{i})_{T_i}\Vert_{\ell_{1}}
\leqslant
\del[1]{\widehat Q_i(\bbeta_{0i})}^{1/2}+\frac{\lambda\widehat{u}\widehat\upsilon^0_{\max}\sqrt{s}}{ns_{\varepsilon_i}\widehat\kappa_{\widehat c_{0}\widehat{\mu}_{0}}}\Vert\dot\bdelta_{i}\Vert_{2,n},
\end{align*}
Inserting this last upper bound into \eqref{eq:sqrt_lasso_intermediate_bnd_1} and once more using~$\lambda/n\geqslant 0.5c\max_{i\in[p]}\Vert \bS_{n,i}\Vert_{\ell_{\infty}}$, we arrive at
\begin{align*}
\Vert\dot{\bdelta}_{i}\Vert_{2,n}^{2}	
&\leqslant	
\frac{2\lambda}{cn}\Vert\widehat{\bUpsilon}_{i}^{0}\dot{\bdelta}_{i}\Vert_{\ell_{1}}+\del[3]{2s_{\varepsilon_i}+\frac{\lambda\widehat{u}\widehat\upsilon^0_{\max}\sqrt{s}}{ns_{\varepsilon_i}\widehat\kappa_{\widehat c_{0}\widehat{\mu}_{0}}}\Vert\dot\bdelta_{i}\Vert_{2,n}}\frac{\lambda}{ns_{\varepsilon_i}}\sbr[3]{\frac{\widehat{u} \widehat\upsilon^0_{\max}\sqrt{s}}{\widehat\kappa_{\widehat c_{0}\widehat{\mu}_{0}}}\Vert\dot\bdelta_{i}\Vert_{2,n}-\widehat{\ell}\Vert(\widehat{\bUpsilon}^0_{i}\dot\bdelta_{i})_{T_i^c}\Vert_{\ell_{1}}}\\
&\leqslant
\frac{2\lambda}{cn}\Vert\widehat{\bUpsilon}_{i}^{0}\dot{\bdelta}_{i}\Vert_{\ell_{1}}-\frac{2\lambda}{n}\widehat{\ell}\Vert(\widehat{\bUpsilon}^0_{i}\dot\bdelta_{i})_{T_i^c}\Vert_{\ell_{1}}+\frac{2\lambda\widehat{u}\widehat\upsilon^0_{\max}\sqrt{s}}{n\widehat\kappa_{\widehat c_{0}\widehat{\mu}_{0}}}\Vert\dot\bdelta_{i}\Vert_{2,n}+\del[3]{\frac{\lambda\widehat{u}\widehat\upsilon^0_{\max}\sqrt{s}}{ns_{\varepsilon_i}\widehat\kappa_{\widehat c_{0}\widehat{\mu}_{0}}}\Vert\dot\bdelta_{i}\Vert_{2,n}}^2.
\end{align*}
Since~$1/c<\widehat\ell$ and~$\dot\bdelta_i\in \mathcal{R}_{\widehat c_{0}\widehat{\mu}_{0},T_{i}}$, the inequality \eqref{eq:sqrt_lasso_restricted_set_consequence} produces
\begin{align*}
\frac{\lambda}{cn}\Vert\widehat{\bUpsilon}_{i}^{0}\dot{\bdelta}_{i}\Vert_{\ell_{1}}-\frac{\lambda}{n}\widehat{\ell}\Vert(\widehat{\bUpsilon}^0_{i}\dot\bdelta_{i})_{T_i^c}\Vert_{\ell_{1}}
\leqslant
\frac{\lambda}{cn}\Vert(\widehat{\bUpsilon}_{i}^{0}\dot{\bdelta}_{i})_{T_i}\Vert_{\ell_{1}}
\leqslant
\frac{\lambda \widehat\upsilon^0_{\max} \sqrt{s}}{cn\widehat\kappa_{\widehat c_{0}\widehat{\mu}_{0}}}\enVert[0]{\dot\bdelta_{i}}_{2,n},
\end{align*}
so the previous two displays combine to yield
\begin{align*}
\Vert\dot{\bdelta}_{i}\Vert_{2,n}^{2}	
\leqslant		
2\left(\widehat{u}+\frac{1}{c}\right)\frac{\lambda\widehat\upsilon^0_{\max}\sqrt{s}}{n\widehat\kappa_{\widehat c_{0}\widehat{\mu}_{0}}}\Vert\dot\bdelta_{i}\Vert_{2,n}+\del[3]{\frac{\lambda\widehat{u}\widehat\upsilon^0_{\max}\sqrt{s}}{ns_{\varepsilon_i}\widehat\kappa_{\widehat c_{0}\widehat{\mu}_{0}}}\Vert\dot\bdelta_{i}\Vert_{2,n}}^2.
\end{align*}
From~$\frac{\lambda\widehat{u} \widehat\upsilon^0_{\max}\sqrt{s}}{ns_{\varepsilon_i}\widehat\kappa_{\widehat c_{0}\widehat{\mu}_{0}}}\leqslant\frac{1}{\sqrt{2}}$, the previous display allows us to deduce that
\begin{align*}
\frac{1}{2}\Vert\dot{\bdelta}_{i}\Vert_{2,n}^{2}
\leqslant 
 2\left(\widehat{u}+\frac{1}{c}\right)\frac{\lambda\widehat\upsilon^0_{\max}\sqrt{s}}{n\widehat\kappa_{\widehat c_{0}\widehat{\mu}_{0}}}\Vert\dot\bdelta_{i}\Vert_{2,n},
\end{align*}
which implies the claimed bound on~$\Vert\dot{\bdelta}_{i}\Vert_{2,n}$. Again using that~$\dot\bdelta_i\in \mathcal{R}_{\widehat c_{0}\widehat{\mu}_{0},T_{i}}$, the $\ell_{2}$-error
bound follows from~$\Vert\dot{\bdelta}_{i}\Vert_{\ell_{2}}\leqslant\Vert\dot{\bdelta}_{i}\Vert_{2,n}/\widehat\kappa_{\widehat c_{0}\widehat{\mu}_{0}}$ and the prediction error bound, and the $\ell_{1}$-bound subsequently follows from $\Vert\dot{\bdelta}_{i}\Vert_{\ell_{1}}\leqslant\left(1+\widehat c_{0}\widehat{\mu}_{0}\right)\Vert\dot{\bdelta}_{iT_{i}}\Vert_{\ell_{1}}\leqslant\left(1+\widehat c_{0}\widehat{\mu}_{0}\right)\sqrt{s}\Vert\dot{\bdelta}_{i}\Vert_{\ell_{2}}$ and the $\ell_2$-bound.
\end{proof}

\begin{proof}[\sc{Proof of Theorem~\ref{thm:Rates-for-sqrtLASSO-data-driven-loadings}}]
We verify that the assumptions of Lemma~\ref{lem:sqrtLasso} are satisfied with probability tending to one. By Lemma \ref{lem:Deviation-Bound}~$\P(0.5\lambda^\ast_n/n\geqslant 0.5c\max_{i\in[p]}\Vert \bS_{n,i}\Vert _{\ell_{\infty}})\to1$, and by Lemma \ref{lem:Ideal-Penalty-Loading-Control} there are constants $(\underline{\upsilon}^{0},\overline{\upsilon}^{0})\in(0,\infty)^2$, such that~$\P(\underline{\upsilon}^{0}\leqslant\widehat{\upsilon}_{\min}^{0}\leqslant\widehat{\upsilon}_{\max}^{0}\leqslant\overline{\upsilon}^{0})\to1$ and, thus,~$\P(\widehat{\mu}_{0}\leqslant\overline\upsilon^0/\underline\upsilon^0)\to1$. 

Recall that $\overline q=0$ implies that the innovations $\{\bepsilon_t\}_{t\in\Z}$ are i.i.d. Letting~$\sigma_{\varepsilon_i}^2:=\E[\varepsilon_{0,i}^2]$, the independence of~$\bepsilon_t$ and~$\bZ_{t-1}$ implies $\upsilon_{i,j}^{0}:=(\E[\varepsilon_{0,i}^{2}Z_{-1,j}^{2}])^{1/2}=\sigma_{\varepsilon_i}(\E[Z_{-1,j}^2])^{1/2}$ for $i\in[p]$ and $j\in[pq]$.  We now show that 
\begin{align}\label{eq:sqrtweigthslimit}
\max_{(i,j)\in[p]\times[pq]}\envert[1]{(s_{\varepsilon_i}\dot\upsilon_j)^2-(\upsilon_{i,j}^{0})^2}\overset{\P}{\to}0.
\end{align}
To this end, Lemma \ref{lem:Covariance-Consistency} yields $\max_{j\in[pq]}\envert[0]{\dot\upsilon_j^2-\E[Z_{-1,j}^2]}\to_\P0$,
and $\max_{i\in[p]}\envert[0]{s_{\varepsilon_i}^2-\sigma_{\varepsilon_i}^2}\to_\P0$
follows from arguments parallel to those used in the proof of Lemma~\ref{lem:Covariance-Consistency}. (Replace~$\cbr[0]{\bZ_{t}}_{t\in\Z}$ by~$\cbr[0]{\bepsilon_{t}}_{t\in\Z}$ and use that the latter are i.i.d.) Therefore, since~$\max_{j\in[pq]}\E[Z_{-1,j}^2]$ is bounded by Part \ref{enu:ZiUnifMomentBnd} of Lemma \ref{lem:ZandEpsZNormBnds} and~$\max_{i\in[p]}\sigma_{\varepsilon_i}^2$ is bounded by Assumption \ref{assu:Innovations}.\ref{enu:EpsSubWeibullBeta}, \eqref{eq:sqrtweigthslimit} now follows from a calculation, since~$(\upsilon_{i,j}^{0})^2=\sigma_{\varepsilon_i}^2\E[Z_{-1,j}^2]$. Furthermore, the proof of Lemma~\ref{lem:Ideal-Penalty-Loading-Control} (specifically, the conclusion of Lemma \ref{lem:DeltaIdealVanishes}) establishes
\begin{align*}
\max_{(i,j)\in[p]\times[pq]}\envert[1]{(\widehat{\upsilon}_{i,j}^{0})^2-(\upsilon_{i,j}^{0})^2}\overset{\P}{\to}0.
\end{align*}
The previous display and~$\min_{(i,j)\in[p]\times[pq]} (\upsilon_{i,j}^{0})^2\geqslant a_2d>0$ imply that~$\max_{(i,j)\in[p]\times[pq]}(1/\widehat{\upsilon}_{i,j}^{0})\lesssim_\P1$.
Combined with~\eqref{eq:sqrtweigthslimit}, these observations and a calculation combine to show
\begin{align*}
\max_{(i,j)\in[p]\times[pq]}\envert[3]{\frac{s_{\varepsilon_i}\dot\upsilon_j}{\widehat{\upsilon}_{i,j}^{0}}-1}\overset{\P}{\to}0,
\end{align*}
such that for the purpose of satisfying \eqref{eq:pensqrtLasso} one may choose~$\widehat\ell=1/c+0.5(1-1/c)\in(1/c,1)$ and~$\widehat u=2$, which satify $1/c<\widehat\ell\leqslant\widehat u$ by construction. Careful inspection of the proof of Lemma~\ref{lem:Restricted-Eigenvalue-Bound} shows that the only properties of~$\widehat\ell$ and~$\widehat u$ used there are~$\P(\widehat\ell>1/c)\to1$ and~$\widehat u\lesssim_\P1$, which suffice to show that $\widehat{c}_0=(\widehat{u}c+1)/(\widehat{\ell}c-1)\lesssim_\P1$. The argument presented there therefore shows~$\P(\widehat\kappa^2_{\widehat c_0\widehat\mu_0}\geqslant d/2)\to1$ for the~$\widehat\ell$ and~$\widehat u$ (and the implied~$\widehat{c}_0$) given here.

Next, using~$\widehat u \widehat\upsilon^0_{\max}\lesssim_\P1$, $\max_{i\in[p]}\envert[0]{s_{\varepsilon_i}^2-\sigma_{\varepsilon_i}^2}\to_\P 0$,~$\min_{i\in[p]}\sigma_{\varepsilon_i}^2\geqslant a_2$, $\lambda_n^\ast\lesssim \sqrt{n\ln(pn)}$ [which follows from $\Phi^{-1}(1-z)\leqslant\sqrt{2\ln(1/z)},z\in(0,1)$], and Assumption \ref{assu:RowSparsity}, we see that the qualifier in \eqref{eq:sqrt_lasso_bound_qualifier} holds with probability approaching one.
Using again~$\lambda_n^\ast\lesssim \sqrt{n\ln(pn)}$, the rates in~(\ref{eq:lasso_rates_1})--(\ref{eq:lasso_rates_3}) for sqrt-Lasso then follow from Lemma~\ref{lem:sqrtLasso} and uniformity in $i\in[p]$.
\end{proof}

\section{Further Simulation Output}\label{sec:appsim}
\begin{figure}[H]\caption{Raw average estimation errors for the simulation designs in Section \ref{sec:Simulations}}\label{fig:SimulationsAbsoluteErrors}
\centering{}
\includegraphics[height=.9\textheight]{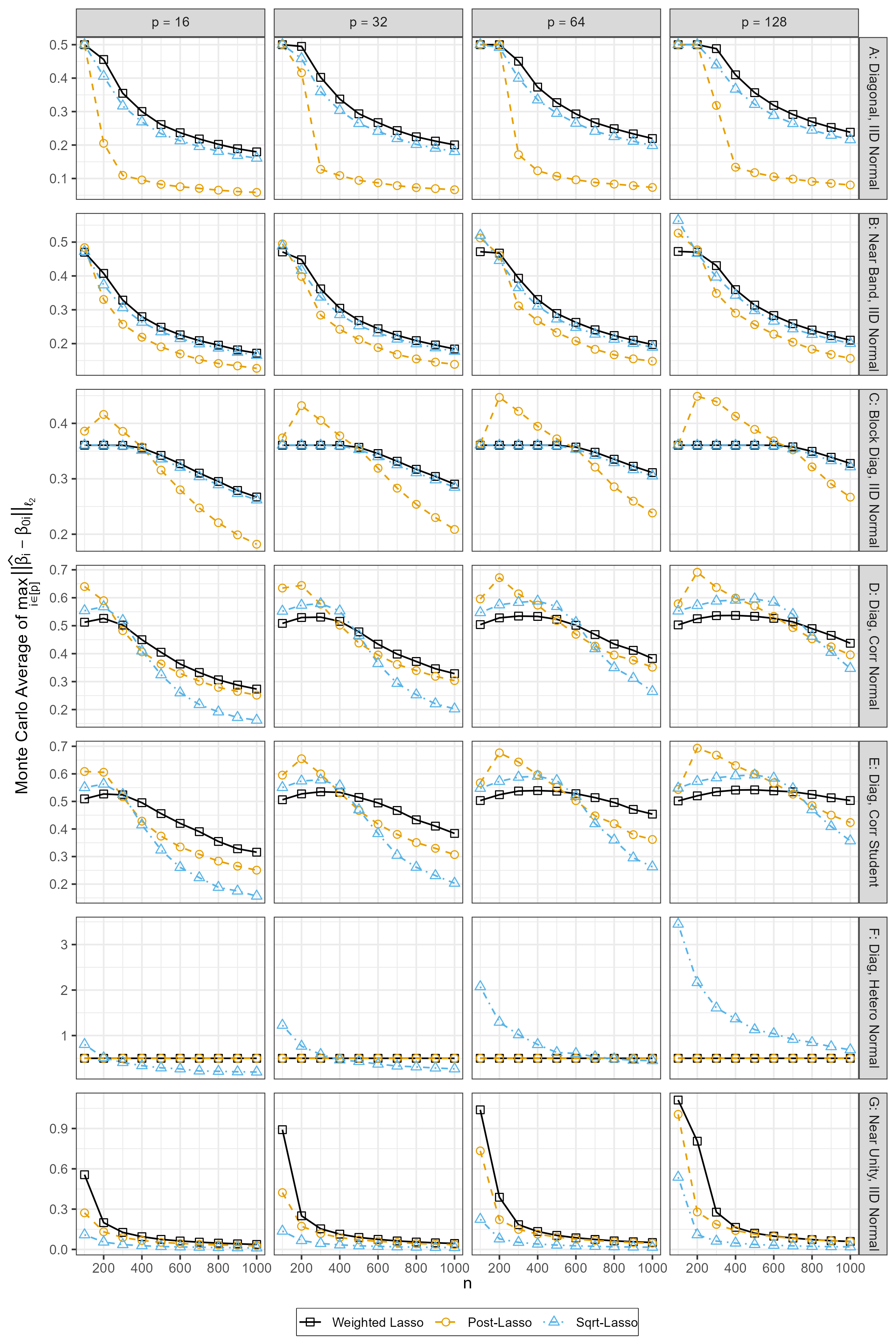}
\end{figure}

\begin{figure}\caption{Raw average estimation errors for simulation Design \protect\hyperlink{Design A}{A} in Section~\ref{sec:Simulations} and their dependence on~$c$. The same pattern is valid for the remaining simulation designs.}\label{fig:cchoice}
    \centering{}
    \includegraphics[width=0.9\linewidth]{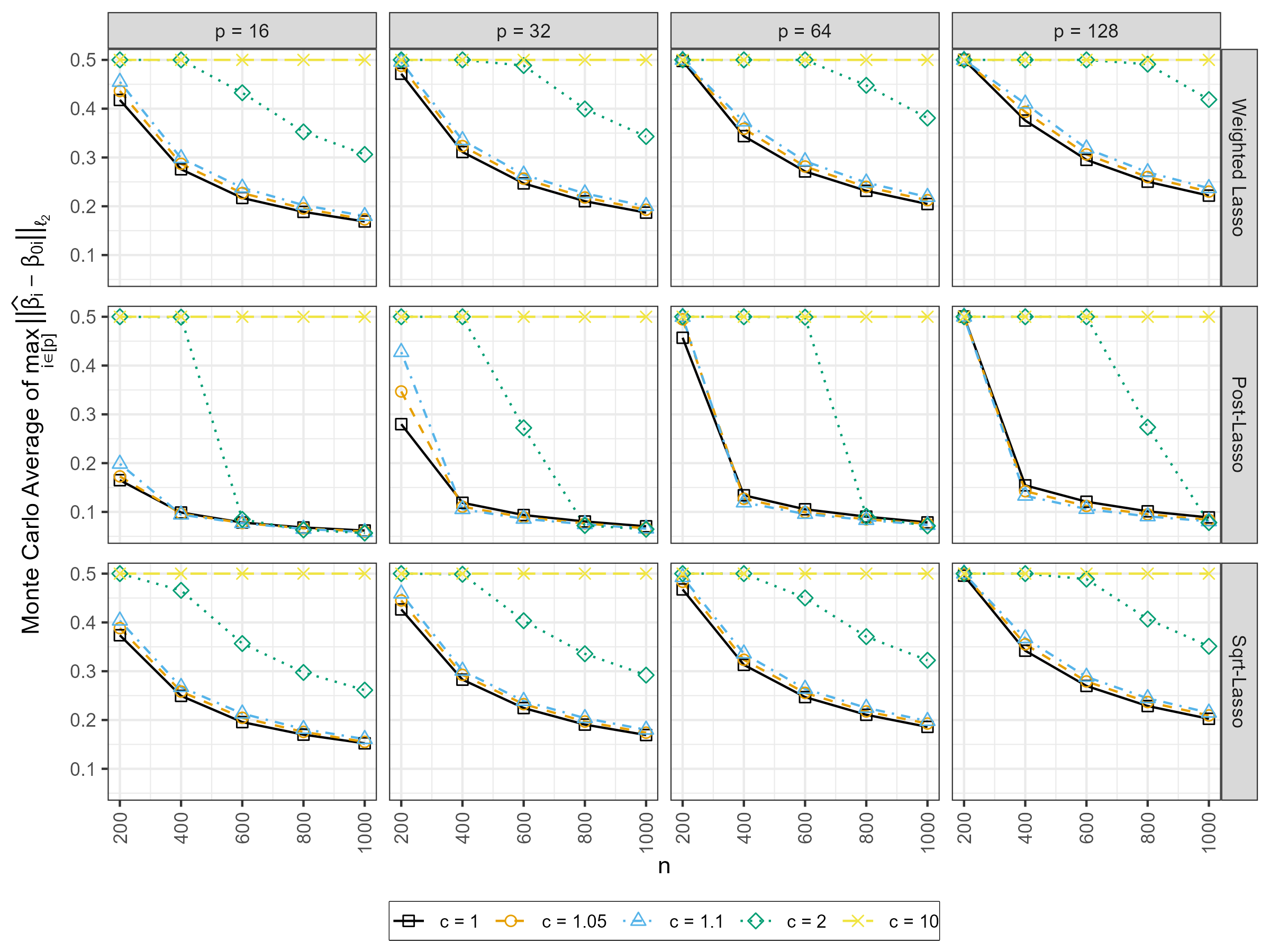}
\end{figure}   

\end{document}